\documentclass[10pt,article]{IEEEtran}
\usepackage{epsfig}
\usepackage[usenames]{color}
\usepackage[numbers]{natbib}
\usepackage{amsmath}
\usepackage{amssymb}
\usepackage{amsthm}
\usepackage{mathrsfs}
\usepackage{amsfonts}
\usepackage{lscape}
\usepackage{graphicx}
\usepackage{enumerate}
\usepackage{pgf,pgfarrows,pgfnodes,pgfautomata,pgfheaps,pgfshade}

\theoremstyle{definition}
\newtheorem{thm}{\textbf{Theorem}}
\newtheorem{cor}[thm]{\textbf{Corollary}}
\newtheorem{lem}[thm]{\textbf{Lemma}}
\newtheorem{prop}[thm]{\textbf{Proposition}}
\newcommand{\Capacity}{\mathcal{C}}
\newcommand{\C}{\mathcal{C}}

\newcommand{\Dcal}{D}
\newcommand{\Xcal}{\mathcal{X}}

\newcommand{\Zcal}{\mathcal{Z}}
\newcommand{\Fcal}{\mathcal{F}}
\newcommand{\Chi}{\Xcal}

\newcommand{\tL}{\Lambda}
\newcommand{\X}{X}

\newcommand{{\main}}{\hbox{\textsl{main}}}

\newcommand{\shift}{\hbox{shift}}

\newcommand{\E}{\mathbb{E}}

\newcommand{\PP}{\mathbb{P}}
\newcommand{\QQ}{\mathbb{Q}}

\newcommand{\intersect}{\cap}
\newcommand{\union}{\cup}

\newcommand{\eqs}[1]{\begin{equation*} #1 \end{equation*}}
\newcommand{\eqq}[1]{\begin{equation} #1 \end{equation}}
\newcommand{\alge}[1]{\begin{align} #1 \end{align}}
\newcommand{\algg}[1]{\begin{align*} #1 \end{align*}}

\newcommand{\cjrb}{C_{j,R}}
\newcommand{\cjrbh}{C_{j,R,h}}

\newcommand{\rr}{r}
\newcommand{\+}{\!+\!}
\newcommand{\rhalf}{r_0}
\newcommand{\sent}{sent}
\newcommand{\others}{other}
\newcommand{\trn}{\scriptscriptstyle\mathrm{T}}
\newcommand{\gcoef}{G^{coef}}
\newcommand{\M}{M}
\newcommand{\Rbb}{\mathbb{R}}
\newcommand{\cspc}{\mathcal{B}}
\newcommand{\fkup}[1]{\hat f^{up}_{#1}}
\newcommand{\aq}[1]{A_{#1}}
\newcommand{\af}[1]{B_{#1}}
\newcommand{\ah}[1]{D_{#1}}
\newcommand{\aee}[1]{E_{#1}}
\newcommand{\aeet}[1]{\tilde{E}_{#1}}
\newcommand{\hatdwght}{\hat{\delta}_{wght}}
\newcommand{\dwght}{\delta_{wght}}

\begin{document}

\title{Fast Sparse Superposition Codes have\\ Exponentially Small Error Probability for $R < \C$}

% use a multiple column layout for up to three different affiliations

\author{Antony~Joseph,~\IEEEmembership{Student~Member,~IEEE,}
        and~Andrew~R~Barron,~\IEEEmembership{Senior~Member,~IEEE}% <-this % stops a space
\thanks{Andrew R. Barron and Antony Joseph
are with the Department of Statistics, Yale University, New Haven, CT 06520 USA
e-mail:  \{andrew.barron,~antony.joseph\}@yale.edu.

Summary \cite{barron2010toward}, \cite{barron2011analysis} of this paper was presented at the IEEE International Symposium
on Information Theory, 2010, 2011 respectively.
 }

% <-this % stops a space

%This manuscript is not for distribution beyond ISIT reviewers as it will\\ subsequently
%Public disclosure to be withheld until the end of May, 2010.\\
%Not yet for public disclosure\\
}

% avoiding spaces at the end of the author lines is not a problem with
% conference papers because we don't use \thanks or \IEEEmembership

\maketitle

\begin{abstract}
For the additive white Gaussian noise channel with average codeword power constraint, sparse superposition codes are developed. These codes are based on the statistical high-dimensional regression framework. The paper [\textit{IEEE Trans. Inform. Theory} 55 (2012), 2541 -- 2557] investigated decoding using the optimal maximum-likelihood decoding scheme. Here a fast decoding algorithm, called  \textit{adaptive successive decoder}, is developed. %with which
%A fast decoding algorithm is presented.
%The non-zero coefficient magnitudes may be constant, or may vary in accordance with a particular assignment best suited for the decoder.
For any rate $R$ less than the capacity $\C$ communication is shown to be reliable with exponentially small error probability.
\vspace {-.4cm}
\end{abstract}


\section{Introduction}
\label{sec:intro}
The additive white Gaussian noise channel is basic to Shannon theory and underlies practical communication models.
Sparse superposition codes for this channel was developed in \cite{barron2010joseph}, where reliability bounds for the  optimal maximum-likelihood decoding were given. The present work provides comparable bounds for our fast adaptive successive decoder. %whereas here a fast decoding algorithm called \textit{adaptive successive decoder} is analyzed.
% A fast algorithm for the {\em sparse superposition codes}, introduced in \cite{barron2010joseph},  is analyzed.  %Theory and practice are linked by devising fast coding and decoding algorithms and by showing sparse superposition codes from moderate size dictionaries with these algorithms achieve nearly exponentially small error probability for any communication rate below the Shannon capacity.
 %The strategy and its analysis merges modern perspectives on statistical regression, model selection, and information theory.

In the familiar communication setup, an encoder maps length $K$ input bit strings $u\!=\!(u_1,u_2,\ldots,u_K)$ into codewords, which are length $n$ strings of real numbers $c_1,c_2,\ldots,c_n$, with power $(1/n)\sum_{i=1}^n c_i^2$. After transmission through the Gaussian channel, the received string $Y = (Y_1,\, Y_2,\ldots, Y_n)$ is modeled by,
$$Y_i = c_i + \epsilon_i\quad\text{for $i = 1,\ldots,n$},$$
where the $\epsilon_i$ are i.i.d. $N(0,\sigma^2)$. The decoder produces an estimates $\hat u$ of the input string $u$, using knowledge of the received string $Y$ and the codebook. %Block error is the event $\hat u \!\neq\! u$.
%, bit error at position $i$ is the event $\hat u_i \!\neq\! u_i$, and the bit error rate is $(1/K)\sum_{i=1}^K 1_{\{\hat u_i \neq u_i\}}$.
%When the input string is partitioned into sections, the section error rate is the fraction of sections not correctly decoded.
The decoder makes a block error if $\hat{ u} \,\neq\, u$. The reliability requirement is that, with sufficiently large $n$, that  the block error probability is small, when averaged over input strings $u$ as well as the distribution of $Y$. The communication rate $R = K/n$ %, the ratio of the input length to the codeword length,
is the ratio of the number of message bits to the number of uses of the channel required to communicate them.

The supremum of reliable rates of communication is the channel capacity $\Capacity \!=\! (1/2) \log_2 (1 \+ P/\sigma^2)$, by traditional information theory %Shannon
\cite{shannon2001mathematical},\cite{cover1991elements}. Here $P$ expresses a control on the codeword power.
%Gallager \cite{Gallager1968},
%or Cover and Thomas
  %This problem is also of interest in mathematics because of relationship to versions of the sphere packing problem as described in Conway and Sloane \cite{ConwaySloane}.
For practical coding the challenge is to achieve arbitrary rates below the capacity, while guaranteeing reliable decoding in manageable computation time.
%The present paper takes steps to meet this challenge.

%The development here is specific to the discrete-time channel for which $Y_i=c_i+\varepsilon_i$, for $i=1,2,\ldots,n$, with real-valued inputs and outputs and with independent Gaussian noise.  %Standard communication models, even in continuous-time, have been reduced to this discrete-time white Gaussian noise setting, or to parallel uses of such, when there is a frequency band constraint for signal modulation and when there is a specified spectrum of noise over that frequency band,
%as in \cite{gallager1968information}.
Solution to the Gaussian channel coding problem, when married to appropriate modulation schemes, is regarded as relevant to myriad settings involving transmission over wires or cables for internet, television, or telephone communications or in wireless radio, TV, phone, satellite or other space communications.

Previous standard approaches, as discussed in \cite{forney1998modulation}, entail a decomposition into separate problems of modulation, of shaping of a multivariate signal constellation, and of coding.  As they point out, though there are practical schemes with empirically good performance, theory for practical schemes achieving capacity is lacking.  In the next subsection we describe the framework of our codes. %In our approach, shaping is built directly into the superposition code design, amenable to the desired analysis.
\subsection{Sparse Superposition Codes}\label{sparsesuper}
 %The framework for superposition codes is the formation of specific forms of linear combinations of a given set of vectors.
 The framework here is as introduced in \cite{barron2010joseph}, but for clarity we describe it again in brief.
  The story begins with  a list $\X_1,\X_2,\ldots,\X_N$ of vectors, each with $n$ coordinates, which can be thought of as organized into a \textit{design}, or \textit{dictionary}, matrix $X$, where,
  $$X_{n \times N} = [X_1 : X_2 : \ldots : X_N].$$
  The entries of $X$ are drawn i.i.d. $N(0,1)$.
 The codeword vectors take the form of particular linear combinations of columns of the design matrix.

  More specifically, we assume $N = LM$, with $L$ and $M$ positive integers, and the design matrix $X$ is split into $L$ sections, each of size $M$. The codewords are of the form $X\beta$, where each $\beta \in \Rbb^N$   belongs to the set
  \algg{\cspc =\{\beta :\,\, &\text{$\beta$ has exactly one non-zero in each section,}\\
                          &\text{with value in section $\ell$ equal to $\sqrt{P_{(\ell)}}$}\}.}
This is depicted in figure \ref{fig:matpic}. The values $P_{(\ell)}$, for $\ell = 1,\,\ldots, L$, chosen beforehand, are positive and satisfy
  \eqq{\sum_\ell^L P_{(\ell)}=P,\label{eq:powersum}}
  where recall that $P$ is the power for our code.

   The received vector is in accordance with the statistical linear model $Y=X\beta+ \varepsilon,$
   where $\varepsilon$ is the noise vector distributed N($0,\sigma^2 I$).

\begin{figure}
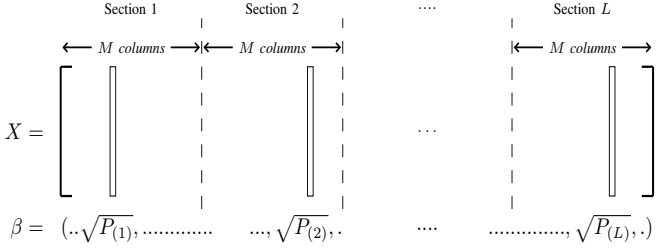

%\centerline{\includegraphics[height=3.5in]}
%\includegraphics[width=60mm]{matrixpic.tex}
\begin{pgfpicture}{0cm}{0cm}{12cm}{3.5cm}
%\pgfrect[stroke]{\pgfxy( 0 , 0 )}{\pgfpoint{ 16 cm}{ 6 cm}}
\begin{pgfmagnify}{.75}{.86}
\begin{pgftranslate}{\pgfpoint{.7cm}{0cm}}
\pgfputat{\pgfxy( -0.1 , 1.2 )}{\begin{pgfrotateby}{\pgfdegree{0}}\pgfbox[right,center]{$X = $}\end{pgfrotateby}}
\pgfsetlinewidth{1pt}
\scriptsize
\pgfline{\pgfxy( 0.2 , 0.2 )}{\pgfxy( 0.4 , 0.2 )}
\pgfline{\pgfxy( 0.2 , 0.2 )}{\pgfxy( 0.2 , 2.2 )}
\pgfline{\pgfxy( 0.2 , 2.2 )}{\pgfxy( 0.4 , 2.2 )}
%\color{red}
\pgfsetlinewidth{0}
\pgfrect[stroke]{\pgfxy( 1.075 , 0.2 )}{\pgfpoint{ 0.1 cm}{ 2 cm}}
\pgfrect[stroke]{\pgfxy( 4.575 , 0.2 )}{\pgfpoint{ 0.1 cm}{ 2 cm}}
\pgfrect[stroke]{\pgfxy( 9.925 , 0.2 )}{\pgfpoint{ 0.1 cm}{ 2 cm}}
%\color{black}
\pgfsetendarrow{\pgfarrowto}
\pgfsetendarrow{}
\pgfsetendarrow{}
\pgfputat{\pgfxy( 6.7 , 1.2 )}{\begin{pgfrotateby}{\pgfdegree{0}}\pgfbox[center,center]{$\ldots$}\end{pgfrotateby}}
\pgfsetlinewidth{1pt}

\pgfline{\pgfxy( 10.7 , 0.2 )}{\pgfxy( 10.7 , 2.2 )}
\pgfline{\pgfxy( 10.7 , 2.2 )}{\pgfxy( 10.5 , 2.2 )}
\pgfline{\pgfxy( 10.7 , 0.2 )}{\pgfxy( 10.5 , 0.2 )}
\pgfsetlinewidth{.8pt}
%\color{blue}
\pgfsetendarrow{\pgfarrowlargepointed{2pt}}
\pgfline{\pgfxy( 0.75 , 2.5 )}{\pgfxy( 0.25 , 2.5 )}
\pgfputat{\pgfxy( 1.45 , 2.5 )}{\begin{pgfrotateby}{\pgfdegree{0}}\pgfbox[center,center]{$\M$ \textit{columns}}\end{pgfrotateby}}
\pgfline{\pgfxy( 2.15 , 2.5 )}{\pgfxy( 2.65 , 2.5 )}
\pgfsetendarrow{}
%\pgfline{\pgfxy( 0.2 , 2.6 )}{\pgfxy( 0.2 , 2.4 )}
%\pgfline{\pgfxy( 2.7 , 2.6 )}{\pgfxy( 2.7 , 2.4 )}

\pgfsetendarrow{\pgfarrowlargepointed{2pt}}
\pgfline{\pgfxy( 3.25 , 2.5 )}{\pgfxy( 2.75 , 2.5 )}
\pgfputat{\pgfxy( 3.95 , 2.5 )}{\begin{pgfrotateby}{\pgfdegree{0}}\pgfbox[center,center]{$\M$ \textit{columns}}\end{pgfrotateby}}
\pgfline{\pgfxy( 4.65 , 2.5 )}{\pgfxy( 5.15 , 2.5 )}
\pgfsetendarrow{}
%\pgfline{\pgfxy( 5.2 , 2.6 )}{\pgfxy( 5.2 , 2.4 )}

\pgfsetendarrow{\pgfarrowlargepointed{2pt}}
\pgfline{\pgfxy( 8.75 , 2.5 )}{\pgfxy( 8.25 , 2.5 )}
\pgfputat{\pgfxy( 9.45 , 2.5 )}{\begin{pgfrotateby}{\pgfdegree{0}}\pgfbox[center,center]{$\M$ \textit{columns}}\end{pgfrotateby}}
\pgfline{\pgfxy( 10.15 , 2.5 )}{\pgfxy( 10.65 , 2.5 )}
\pgfsetendarrow{}
%\pgfline{\pgfxy( 8.2 , 2.6 )}{\pgfxy( 8.2 , 2.4 )}
%\pgfline{\pgfxy( 10.7 , 2.6 )}{\pgfxy( 10.7 , 2.4 )}
\pgfputat{\pgfxy( 1.45 , 3.1 )}{\begin{pgfrotateby}{\pgfdegree{0}}\pgfbox[center,center]{Section 1}\end{pgfrotateby}}
%\pgfputat{\pgfxy( 1.45 , 2.8 )}{\begin{pgfrotateby}{\pgfdegree{270}}\pgfbox[center,center]{$\left\{\rule{0cm}{ 1.35 cm}\right.$}\end{pgfrotateby}}
\pgfputat{\pgfxy( 3.95 , 3.1 )}{\begin{pgfrotateby}{\pgfdegree{0}}\pgfbox[center,center]{Section 2}\end{pgfrotateby}}

%\pgfputat{\pgfxy( 3.95 , 2.8 )}{\begin{pgfrotateby}{\pgfdegree{270}}\pgfbox[center,center]{$\left\{\rule{0cm}{ 1.35 cm}\right.$}\end{pgfrotateby}}
\pgfputat{\pgfxy( 6.7 , 3.1 )}{\begin{pgfrotateby}{\pgfdegree{0}}\pgfbox[center,center]{$....$}\end{pgfrotateby}}
\pgfputat{\pgfxy( 9.45 , 3.1 )}{\begin{pgfrotateby}{\pgfdegree{0}}\pgfbox[center,center]{Section $L$} \end{pgfrotateby}}
%\pgfputat{\pgfxy( 9.45 , 2.8 )}{\begin{pgfrotateby}{\pgfdegree{270}}\pgfbox[center,center]{$\left\{\rule{0cm}{ 1.35 cm}\right.$}\end{pgfrotateby}}
%\color{green}
\pgfsetlinewidth{0}
\pgfsetdash{{0.2cm}{0.2cm}}{0cm}
\pgfline{\pgfxy( 2.7 , 0 )}{\pgfxy( 2.7 , 2.95 )}
\pgfline{\pgfxy( 5.2 , 0 )}{\pgfxy( 5.2 , 2.95 )}
\pgfline{\pgfxy( 8.2 , 0 )}{\pgfxy( 8.2 , 2.95 )}
\normalsize

\color{black}
\pgfputat{\pgfxy( -0.1 , -0.3 )}{\begin{pgfrotateby}{\pgfdegree{0}}\pgfbox[right,center]{$\beta = $}\end{pgfrotateby}}
\pgfputat{\pgfxy( 0.2 , -0.3 )}{\begin{pgfrotateby}{\pgfdegree{0}}\pgfbox[left,center]{$(..\sqrt{P_{(1)}},.............$}\end{pgfrotateby}}
\pgfputat{\pgfxy( 5.2 , -0.3 )}{\begin{pgfrotateby}{\pgfdegree{0}}\pgfbox[right,center]{$...,\sqrt{P_{(2)}},.$}\end{pgfrotateby}}
\pgfputat{\pgfxy( 6.7 , -0.3 )}{\begin{pgfrotateby}{\pgfdegree{0}}\pgfbox[center,center]{$....$}\end{pgfrotateby}}
\pgfputat{\pgfxy( 10.7 , -0.3 )}{\begin{pgfrotateby}{\pgfdegree{0}}\pgfbox[right,center]{$..............,\sqrt{P_{(L)}},.)$}\end{pgfrotateby}}
\end{pgftranslate}
\end{pgfmagnify}
\end{pgfpicture}
\caption{Schematic rendering of the dictionary matrix $X$ and coefficient vector $\beta$. The vertical bars in the $X$ matrix indicate the selected columns from a section.}
\label{fig:matpic}
\end{figure}

Accordingly, with the $P_{(\ell)}$ chosen to satisfy \eqref{eq:powersum}, we have $\|\beta\|^2 = P$ and hence, $\E\|X\beta\|^2/n = P$, for each $\beta$ in $\cspc$. Here $\|. \|$ denotes the usual Euclidian norm. Thus the expected codeword power is controlled to be equal to $P$. Consequently, most of the codewords have power near $P$ and the average power across the $\M^L$ codewords, given by,
$$\frac{1}{\M^L}\sum_{\beta \in \cspc} \|X\beta\|^2/n$$
is concentrated at $P$.

Here, we study both the case of constant power allocation, where each $P_{(\ell)}$ is equal to $P/L$, and a variable power allocation where $P_{(\ell)}$ is proportional to $e^{-2\,{\Capacity}\,\ell/L}$. These variable power allocations are used in getting the rate up to capacity. This is a slight difference from the setup in \cite{barron2010joseph}, where the analysis was for the constant power allocation.

For ease in encoding, it is most convenient  that the section size $\M$ is a  power of two. Then
an input bit string $u$ of length $K\!=\! L \log_2 \M$ splits into $L$ substrings of size $\log_2 \M$ and the encoder becomes trivial. Each substring of $u$ gives the index (or memory address) of the term to be sent from the corresponding section.

As we have said, the rate of the code is $R=K/n$ input bits per channel uses and we arrange for arbitrary $R$ less than ${\Capacity}$.  For the partitioned superposition code this rate is $$R= \frac{L \log \M}{n}.$$ For specified $L$, $\M$ and $R$, the codelength $n=(L/R) \log \M$. Thus the block length $n$ and the subset size $L$ agree to within a log factor. %This paper focuses on partitioned superposition coding.

Control of the dictionary size is critical to computationally advantageous coding and decoding. %Possible dictionary sizes are between the extremes $K$ and $2^K$ dictated by the number and size of the sections.
At one extreme, $L$ is a constant, and section size $\M\!=\!2^{nR/L}$. However, its size, which is exponential in $n$, is impractically large.  At the other extreme $L = nR$ and $\M = 2$. However, in this case the number of non-zeroes of $\beta$ proves to be too dense to permit reliable recovery at rates all the way up to capacity. This can be inferred from recent converse results on information-theoretic limits of subset recovery in regression (see for eg. \cite{wainwright2009information}, \cite{akçakaya2010shannon}).

%At the other extreme there would be $L\!=\!K$ sections, each with two candidate terms in subset coding or two signs of a single term in sign coding with $\M\!=\!1$; in which case $X$ is the generator matrix of a linear code.

Our codes lie in between these extremes. We allow $L$ to agree with the blocklength $n$ to within a $\log$ factor, with $\M$ arranged to be polynomial in $n$ or $L$. For example, we may let $\M = n$, in which case $L = nR/\log n$, or we may set $\M = L$, making $n = (L\log L)/R$.
For the decoder we develop here, at rates below capacity, the error probability is also shown to be exponentially small in $L$.

%Design of the dictionary is guided by what is known from information theory concerning the distribution of symbols in the codewords. By analysis of the converse to the channel coding theorem (as in \cite{CoverThomas2006}), for a reliable code at rate near capacity, with a uniform distribution on the sequence of input bits, the induced empirical distribution on coordinates of the codeword must be close to independent Gaussian, in the sense that the resulting mutual information must be close to its maximum subject to the power constraint.

Optimal decoding for minimal average probability of error consists of finding the codeword $X\beta$ with coefficient vector $\beta \in \cspc$ that maximizes the posterior probability, conditioned on $X$ and $Y$.  This coincides, in the case of equal prior probabilities, with the maximum likelihood rule of seeking
$$\arg \min_{\beta \in \cspc} \|Y - X\beta\|.$$
 Performance bounds for such optimal, though computationally infeasible, decoding are developed in the companion paper \cite{barron2010joseph}.  Instead, here we develop fast algorithms for which we can still establish desired reliability and rate properties. We describe the intuition behind the algorithm in the next section. Section \ref{sec:modificabovealgo} describes the algorithm in full detail.

\subsection{Intuition behind the algorithm}\label{subsec:intuialgo}

From the received $Y$ and knowledge of the dictionary, we decode which terms were sent by an iterative algorithm. Denote as $$\sent = \{j : \beta_j \neq 0\} \quad\quad
 \text{and}\quad\quad \others = \{j : \beta_j = 0\}.$$
 The set $\sent$ consists of one term from each section, and denotes the set of correct terms, while $\others$ denotes the set of wrong terms.
 We now give a high-level description of the algorithm.

 The first step is as follows. For each term $X_j$ of the dictionary, compute the normalized inner product with the received string $Y$, given by,
 $$\Zcal_{1,j} = \frac{X_j^T Y}{\|Y\|},$$
 and see if it exceeds a positive threshold $\tau$. %Denote the associated event
%$$ \Hcal_{1,j} = \{\Zcal_{1,j} \ge \tau\}.$$

%As we shall see in the sequel, the distribution of $\Zcal_{1,j}$ is quite similar to that of a location shifted normal, where the shift is 0 for any $j$ in $\others$ and is a positive quantity for $j$ in $\sent$.  This positive shift for $j$ in $\sent$ is seen to depend on $P_{(\ell)}$, the signal-strength in section $\ell$; The larger the $P_{(\ell)}$, the more is this location shift.

%The above gives us a means to identify at least some of the correct terms (that is those in $sent$), in the first step. The threshold is chosen to be
%\eqq{\tau = \sqrt{2\log \M} + a.\label{eq:taudef}}
The idea of the threshold $\tau$ is that very few of the terms in $\others$ will be above threshold. Yet a positive fraction of the terms in $sent$ will be above threshold, and hence, will be correctly decoded on this first step. Denoting as $J = \{1,\,2,\ldots,\, N\}$,  take $$dec_1 = \{j \in J: \Zcal_{1,j} \geq \tau\}$$ as the set of terms detected in the first step. %Recall that this is also the set of terms with the test statistic above threshold.

%Let $thresh_1 = \{j \in J: 1_{\Hcal_j}=1\}$ be the set of terms with the test statistic above threshold and let $above_1$ denote the fraction of such terms. In the variable power case it is a weighted fraction $above_1 = \sum_{j \in thresh_1} P_j/P$, weighted by the power $P_j$.  We restrict decoding on the first step to terms in $thresh_1$ so as to avoid false alarms. The decoded set is either taken to be $dec_1=thresh_1$ or, more generally, a value $pace_1$ is specified and, considering the terms in $J$ in order of decreasing $\Zcal_{1,j}$, we include in $dec_1$ as many as we can with $\sum_{j\in dec_1} \pi_j$ not more than $\min\{pace_1,above_1\}$. Let $DEC_1$ denote the cardinality of the set $dec_1$.

Denoting $P_j = P_{(\ell)}$ if $j$ is in section $\ell$, the
 output of the first step consists of the set of decoded terms $dec_1$ and the vector $$F_1 = \sum_{j \in dec_1} \sqrt{P_j} \, X_j, $$ which forms the first part of the fit. %Notice that
% $F_1$ may also be expressed as $$\sum_{\ell = 1}^L \sum_{j \in section\,\, \ell} \sqrt{P_{(\ell)}} X_j.$$
The set of terms investigated in step 1 is $J_1=J$, the set of all columns of the dictionary. Then the set $J_2=J_1 -dec_1$ remains for second step consideration.  In the extremely unlikely event that $dec_1$ is already at least $L$ there will be no need for the second step.

For the second step, compute the residual vector
$$R_2\,=\, Y\,- F_1.$$
For each of the remaining terms, that is terms in $J_2$, compute the  normalized inner product  $$\Zcal_{2,j}^{res}=\frac{X_j^T R_2 }{\|R_2\|},$$ which is compared to the same threshold $\tau$. Then $dec_2$, the  set of decoded terms for the second step, is chosen in a manner similar to that in the first step. In other words, %defining $\Hcal_{2,j}^{res} = \{\Zcal_{2,j}^{res} \geq \tau\}$,
we take $$dec_2 = \{j \in J_2 : \Zcal_{2,j}^{res} \geq \tau\}.$$
From the set $dec_2$, compute the fit $F_2 = \sum_{j \in dec_2}\sqrt{P_j} \, X_j$, for the second step.

The third and subsequent steps would proceed in the same manner as second step. For any step $k$, we are only interested in
$$J_k = J - dec_1\cup dec_2 \ldots\cup dec_{k-1},$$ that is, terms not decoded previously.
One first computes the residual vector $R_k = Y - (F_1 + \ldots + F_{k-1})$. Accordingly, for terms in $J_k$, we take $dec_k$ as the set of terms for which $\Zcal_{k,j}^{res} = X_j^T R_k \, /\|R_k\|$ is above $\tau$. %The set of decoded terms is either taken to be $thresh_k$ or a subset of it.

%The decoding stops when the size of the cardinality of the set of all decoded term becomes $L$ or is there are no terms above threshold in a particular step. Also, if in the course of the algorithm, two terms are detected in a section, then we declare an error in that section.

We arrange the algorithm to continue until at most a pre-specified number of steps $m$, arranged to be of the order of $\log\M$. The algorithm could stop before $m$ steps if either there are no terms above threshold, or if $L$ terms have been detected. Also, if in the course of the algorithm, two terms are detected in a section, then we declare an error in that section.

Ideally, the decoder selects one term from each section, producing an output which is the index of the selected term. For a particular section, there are three possible ways a mistake could occur when the algorithm is completed. The first is an \textit{error}, in which the algorithm selects exactly one wrong term in that section. The second case is when two or more terms are selected, and the third is when no term is selected. We call the second and third cases \textit{erasures} since we know for sure that in these cases an error has occurred. Let $\hat{\delta}_{mis,error},\,\hat{\delta}_{mis,erasure}$ denote the fraction of sections with error, erasures respectively.
Denoting the section mistake rate,
\eqq{\hat{\delta}_{mis} =  2\,\hat{\delta}_{mis,error}\, + \,\hat{\delta}_{mis,erase}\label{eq:hatdeltamis},}
our analysis provides a good bound, denoted by $\delta_{mis}$, on $\hat{\delta}_{mis}$ that is satisfied with high probability.  %In the next subsection we describe results on the performance of the algorithm.

The algorithm we analyze, although very similar in spirit, is a modification of the above algorithm. The modifications are made so as to help characterize the distributions of $\Zcal_{k,j}^{res}$, for $k \geq 2$. These are described in section \ref{sec:modificabovealgo}. For ease of exposition, we first summarize the results using the modified algorithm.

%\subsection{Accounting for section mistakes} \label{subsec:mistakesaccount}

\subsection{Performance of the algorithm}\label{subsec:performadapt}

With constant power allocation, that is with $P_{(\ell)} = P/L$ for each $\ell$, the decoder is shown to reliably achieve rates up to a threshold rate $R_{0} = (1/2)P/(P\+\sigma^2)$, which is less than capacity. This rate $R_{0}$  is seen to be close to the capacity when the signal-to-noise ratio $snr$ is low. However, since it is bounded by $1/2$, it is substantially less than the capacity for larger $snr$. To bring the rate higher, up to capacity, we use variable power allocation with power
\eqq{P_{(\ell)}\propto e^{-2\,\Capacity \ell/L},\label{eq:powerweuse}} for sections $\ell$ from $1$ to $L$.

%Concerning the advantages of variable power in the partitioned code case, which allows our scheme to achieve rates near capacity, the idea is that the power allocations proportional to $e^{-2{\Capacity}\ell/L}$ give some favoring to the decoding of the higher power sections among those that remain each step. This produces more statistical power for the test initially as well as retaining enough discrimination power for subsequent steps.

As we shall review, such power allocation also would arise if one were attempting to successively decode one section at a time, with the signal contributions of as yet un-decoded sections treated as noise, in a way that splits the rate $\Capacity$ into $L$ pieces each of size ${\Capacity}/L$; however, such decoding would require the section sizes to be exponentially large to achieve desired reliability.  In contrast, in our adaptive scheme, many of the sections are considered each step.  %The power allocations do not change too much across many nearby sections, so that a sufficient distribution of decodings can occur each step.

For rate near capacity, it helpful to use a modified power allocation, where
\eqq{P_{(\ell)} \propto  \max\{e^{-2\Capacity \frac{\ell-1}{L}}, u\},\label{eq:modpower}}
with a non-negative value of $u$. However, since its analysis is more involved we do not pursue this here. Interested readers may refer to documents \cite{barron2010ajoseph}, \cite{josephPhD} for a more thorough analysis including this power allocation.

%Thus $u_{cut}$ can be slightly larger than $e^{-2\Capacity}$. This modification performs a slight leveling of the power allocation for $\ell/L$ near $1$. It helps ensure that, even in the end game, there will be sections for which the true terms are expected to have inner product above threshold.

%With power $P_{(\ell)}$ for $\ell$ from $1$ to $L$, we set $P_j=P_{\ell}$ for terms $j$ in section $\ell$, recalling that $\beta_j^2 = P_j \,1_{j \: sent}$.  Moreover we set weights $\pi_j = P_j/P$.  For any subset

\begin{figure}
\begin{center}
\includegraphics[width=3.00in]{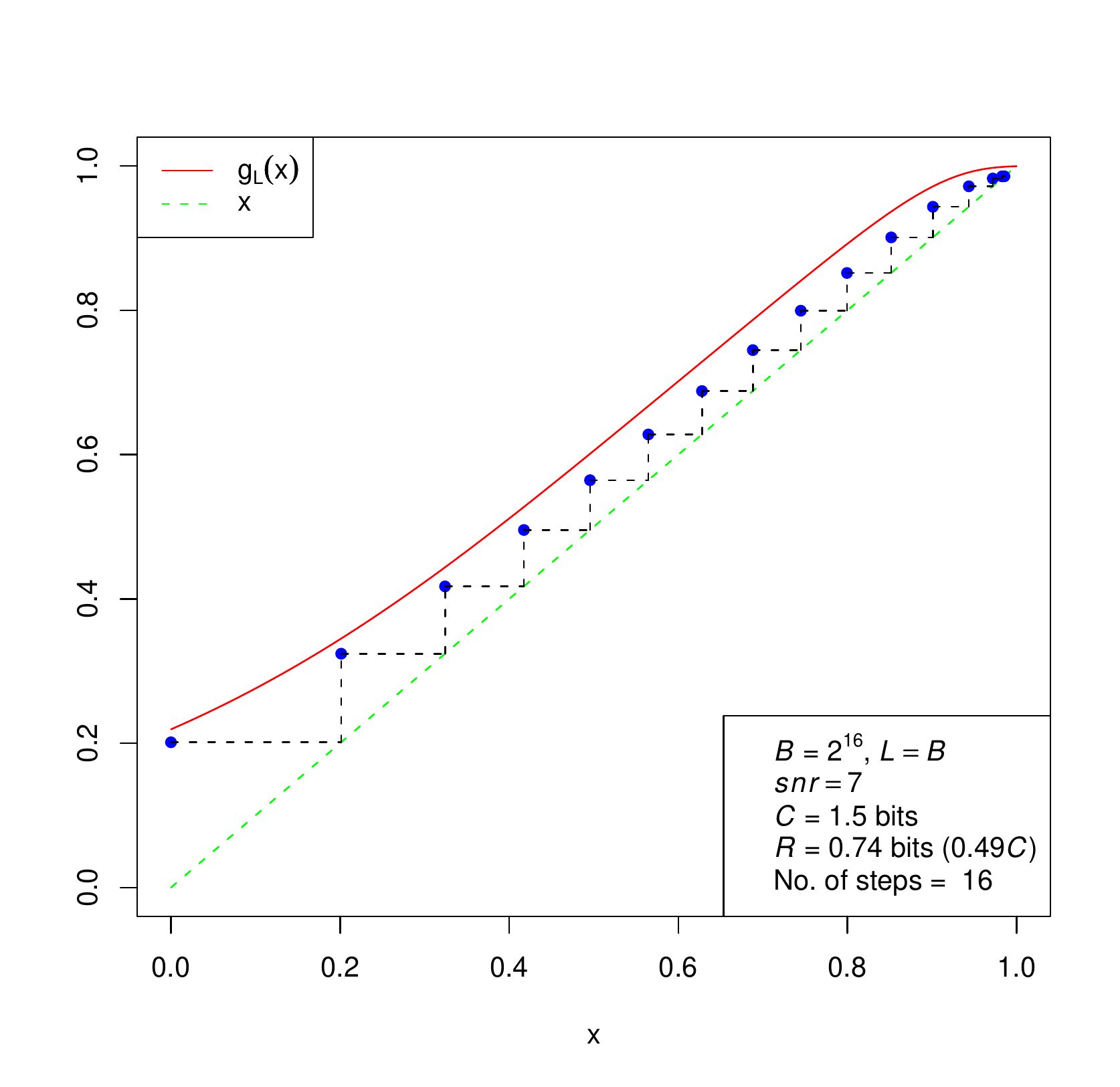}
\end{center}
\caption{Plot of the update function $g_L(x)$. The dots measure the proportion of sections correctly detected after a particular number of steps. Here $\M = 2^{16},\, snr = 7$, $R = 0.74$ and $L$ taken to be equal to $\M$. %Also for the above we take $a =0.8$, $c = 2$ and $\gamma = 0.79C$.
The height reached by the $g_L(x)$ curve at the final step corresponds to a $0.986$ proportion of section correctly detected, and a failed detection rate target of $0.013$.  The accumulated false alarm rate bound is $0.008$. %The weighted (un-weighted) correct detection rate is $0.994$ ($0.986$) and the false alarm and failed detection rates are $0.008$ and  $0.013$ respectively.
The probability of mistake rates larger than these targets is bounded by $1.5 \times 10^{-3}$.
}
\label{fig:prog7}
\end{figure}
The analysis leads us to a function $g_L : [0,1] \rightarrow [0,1]$, which depends on the power allocation and the various parameters $L,\,\M,\,snr$ and $R$, that proves to be helpful in understanding the performance of successive steps of the algorithm. If $x$ is the previous success rate, then $g_L(x)$ quantifies the expected success rate after the next step. An example of the role of $g_L$ is shown in Fig \ref{fig:prog7}.

%Let $\pi_j = P_j/P$, which sums to $1$ across $j$ in $\sent$, and sums to $\M\!-\!1$ across $j$ in $\others$.  Define in general \eqq{\hat q_k = \sum_{j \in \sent \intersect dec_k} \pi_j\label{eq:hatq}} for the step $k$ correct detections and \eqq{\hat f_k = \sum_{ j\in other \intersect dec_k} \pi_j\label{eq:hatf}}
%for the false alarms.
%In the case $P_j=P/L$ which assigns equal weight $\pi_j=1/L$, then $\hat q_k \, L$ is the increment to the number of correct detections on step $k$,
%$$ \, = \, \sum_{j \in \sent \intersect J_k} 1_{\Hcal_{k,j}}$$
%likewise $\hat f_k \, L$ is the increment to the number of false alarms. We define the weighted mistake rate after $m$ steps as
%\eqq{\hatdwght = (1 - \sum_{k = 1}^m \hat q_k ) + \sum_{k = 1}^m \hat f_k.\label{eq:hatdeltawghtdef}}

An outer Reed-Solomon codes completes the task of identifying the fraction of sections that have errors or erasures (see section VI of \citet{barron2010joseph} for details)   %Define section error rate $\hat \delta_{mis}$ as the fraction of sections undetected or incorrectly detected by the algorithm.
%The Reed-Solomon code corrects the small fraction $\hat \delta_{mis}$ of mistakes so that we end up not only with a small section mistake rate,
so that we end up with a small block error probability. If $R_{outer} = 1 - \delta$ is the rate of an RS code, with  $0 < \delta < 1 $, then a section mistake rate $\hat{\delta}_{mis}$ less than $\delta_{mis}$ can be corrected, provided $\delta_{mis} < \delta$. Further, if $R$ is the rate associated with our inner (superposition) code, then the total rate after correcting for the remaining mistakes is given by $R_{tot} = R_{outer}R$. The end result, using our theory for the distribution of the fraction of mistakes of the superposition code, is that the block error probability is %within a log factor of being
exponentially small. %exponentially small with exponent near a multiple of $(C-R_F)^2$ when $R_F$ is the final rate of the composite code.
One may regard the composite code as a superposition code in which the subsets are forced to maintain at least a certain minimal separation, so that decoding to within a certain distance from the true subset implies exact decoding.

%We allow not only for fixed rates $R < \Capacity$, but also for rates for which the gap from capacity is of the order of a polynomial in $1/\log \M$.
%There are two performance regimes. The first is for fixed rates $R<\Capacity$. A value can be given that is of the order of a polynomial in $1/\log \M$, with exponent depending on these parameters $R$ and $\Capacity$.  After a moderate number of steps, the fraction of section mistakes is %shown to be less than this value except in an event of probability exponentially small in $L$.
%The more refined conclusions presented here are for rate $R$ approaching capacity, with the slight leveling of the power allocation permitted, as discussed above.

In the proposition below, we assume that the power allocation is given by \eqref{eq:powerweuse}. Further, we assume that the threshold $\tau$ is of the form,
\eqq{\tau = \sqrt{2\log \M} + a \label{eq:taudef},}
with a positive $a$ specified in subsection \ref{choicepara}.

We allow rate $R$ up to $\Capacity^*$, where $\Capacity^*$ can be written as $$\Capacity^* = \frac{\Capacity}{1 + drop^*}.$$
Here $drop^*$ is a positive quantity given explicitly later in this paper. It is near%we mention here two approximations. Denote,
\eqq{\delta_\M = \frac{1}{\sqrt{\pi \log \M}} \label{eq:deltam},}
ignoring terms of smaller order.
%\begin{enumerate}
%  \item When $snr$ is large compared to $\log\log \M$, it can be approximated by,
%$$ drop^* \approx \frac{3\log\log \M  + 8\C + 2/\C + 4/(1+ 1/(2\C))}{2\log \M} + \delta_\M.$$
%\item For $snr$ near 1,
%$$ drop^* \approx \frac{3\log\log \M  +  2/\C + 4/(1+ 1/(2\C)}{2\log \M} + \delta_\M$$
%\end{enumerate}

Thus $\Capacity^*$ is within order $1/\sqrt{\log \M}$ of capacity and tends to $\Capacity$ for large $\M$. With the modified power allocation \eqref{eq:modpower}, it is shown in \cite{barron2010ajoseph}, \cite{josephPhD} that one can make $\C^*$ of order $1/\log \M$ of capacity.

%In summary form, the following expresses the main result of the performance of the iterative algorithm.

\begin{prop} \label{mainthm}
For any inner code rate $R < \Capacity^*$, express it in the form \eqq{R = \frac{\Capacity^*}{1 + \kappa/\log \M },\label{eq:ratedropcstar}}
with $\kappa \geq 0$. Then, for the partitioned superposition code,

\begin{enumerate}[I)]%for capital roman numbers.
\item The adaptive successive decoder admits fraction of section mistakes less than \eqq{\delta_{mis} = \frac{ 3\kappa + 5 }{8\Capacity \,\log \M} + \frac{\delta_\M}{2\C}  \label{eq:deltamisprac}}
except in a set of probability not more than
$$p_e = \kappa_1 e^{-\kappa_2L\,\min\left\{\kappa_3(\Delta^*)^2\,,\,\kappa_4(\Delta^*) \,\right\}},$$
where
$$\Delta^* = (\Capacity^* - R)/\Capacity^*.$$
Here  $\kappa_1$ is a constant to be specified later that is only polynomial in $\M$. Also, $\kappa_2,\, \kappa_3$ and $\kappa_4$ are constants that depend on the $snr$. See subsection \ref{proofmainthm} for details.  %(See Remarks for details).

%See subsection \ref{proofmainthm} for details.

\item After composition with an outer Reed Solomon code the decoder admits block error probability less than
$p_e$, with the composite rate being $R_{tot} = (1 - \delta_{mis})\,R$.

\end{enumerate}
\end{prop}

% The constants $\kappa_2,\, \kappa_3$ and $\kappa_4$ can be specified as follows. Define $\nu = snr/(1 + snr)$. We have that $\kappa_2$ is near $\nu/(2\Capacity)$. Further $$\kappa_3 = \min\left\{1/(32snr^2),\, 1/(16\Capacity^*)\right\} $$ and $$\kappa_4 = \frac{1}{(8 + 4/\C)snr^2 \log\M}.$$

The proof of the above proposition is given in subsection \ref{proofmainthm}. The following is an immediate consequence of the above proposition. %For the corollary below, we denote $R = R_{total}$ as the total rate, after composition with the Reed-Solomon code.

%Further, using the modified power allocation \eqref{eq:modpower}, the rate drop can be made of a much smaller order of $1/\log \M$, with $\kappa$ of constant order. The distribution analysis for power allocations analyzed here carries over for this modified power allocation; the difference being in the analysis of section \ref{commcapac}, which is much more involved using the modified power allocation. Interested readers please refer to
%\cite{barron2010ajoseph}, \cite{josephPhD} for analysis including this power allocation.

\begin{cor}\label{cor:cor}
  For any fixed total  communication rate $R_{tot} < \C$, there exists a dictionary $X$ with size $N = L\M$ that is polynomial in block length $n$, for which the sparse superposition code, with \textit{adaptive successive decoding} (and outer Reed-Solomon code), admits block error probability $p_e$ that is exponentially small in $L$. In particular,
    $$\lim_{n \rightarrow \infty} \frac{1}{L}\log(1/p_e) \geq const \min\{\Delta,\, \Delta^2\},$$
    where here $\Delta = (\C - R_{tot})/\C$ and $const$ denotes a positive constant.
\end{cor}

The above corollary follows since  if  $\kappa$ is of order $\sqrt{\log \M}$, then $\Delta^*$ is of order $1/\sqrt{\log \M}$. Further, as $\C^*$ is of order $1/\sqrt{\log M}$ below the capacity $\C$, we also get that $\Delta_{inner} = (\C - R)/\C$ is also $1/\sqrt{\log \M}$ below capacity.
From the expression for $\delta_{mis}$ in \eqref{eq:deltamisprac}, one sees that the same holds for the total rate drop, that is $\Delta = (\C - R_{tot})/\C$. A more rigorous proof is given in subsection \ref{proofmainthm}.

\subsection{Comparison with Least Squares estimator}\label{subsec:lstcomparison}

Here we compare the rate achieved here by our practical decoder with what is achieved with the theoretically optimal, but possibly impractical, least squares decoding of these sparse superposition codes shown in the companion paper \cite{barron2010joseph}.
%, subject to the constraint that there is one non-zero coefficients in each section.
%Such least squares decoding provides the stochastically smallest distribution of the number of mistakes, with a uniform distribution on the possible messages, but it has an unknown computation time. %For the target mistake rates we focus on here, the performance achieved with our practical decoder, entails a rate nearly as close to $\Capacity$ as is possible for the optimal least squares decoding.
%Getting very close to $\Capacity$ by this scheme is not possible with reasonable size

%In this direction, the results in the companion paper \cite{BarronJosephLeastSquares}, for least squares decoding of superposition codes, partially complement what we give here for our adaptive successive decoder.  To briefly describe, for optimum least square decoding, favorable properties are demonstrated, in the case that the power assignments $P/L$ are the same for each section. [Interestingly, the analysis techniques there are different and do not reveal rate improvement from use of variable instead of constant power with optimal least squares decoding.]  Suppose $\M \!\ge\! L^b$ for a specified section size rate $b$ depending only on the signal-to-noise ratio, where conveniently $b$ tends to $1$ for large signal-to-noise. For usefulness of the results there in comparison with our schemes here, restrict attention to moderate and large signal-to-noise ratios, as for computational reasons, it is desirable that $\M$ be not more than a low order polynomial in $L$.

Let $\Delta = (\Capacity \!-\!R)/\Capacity$ be the rate drop from capacity, with $R$ not more than $\Capacity$. The rate drop $\Delta$ takes values between 0 and 1. With power allocated equally across sections, that is with $P_{(\ell)} = P/L$, it was shown in \cite{barron2010joseph} that for any $\delta_{mis} \in [0,1)$, the probability of more than a fraction $\delta_{mis}$ of mistakes, with least squares decoding, is less than $$\exp\{-nc_1\min\{\Delta^2,\delta_{mis}\}\},$$ for any positive rate drop $\Delta$ and any size $n$. The error exponent for the above is correct, in terms of orders of magnitude, to the theoretically best possible error probability for any decoding scheme, as established by Shannon and Gallager, and reviewed for instance in \cite{polyanskiy2010channel}.

%, from which it is seen what size $n$ is needed to achieve target performance objectives.  Since $L$ agrees with $n$ to within a log factor and $\M$ is polynomially related to $L$, it reveals what we require for these size parameters as well.
%dependence on $\Delta= \Capacity\!-\!R$ than what we obtain for our practical decoder, as
The bound obtained for the least squares decoder is better than that obtained for our practical decoder in its freedom of any choice of mistake fraction, rate drop and size of the dictionary matrix $X$.
Here, we allow for rate drop $\Delta$ to be of order $1/\sqrt{\log \M}$. Further, from the expression \eqref{eq:deltamisprac}, we have $\delta_{mis}$ is of order $1/\sqrt{\log \M}$, when $\kappa$ is taken to be of $O(\sqrt{\log\M})$. Consequently, we compare the error exponents obtained here with that of the least squares estimator of \cite{barron2010joseph}, when both $\Delta$ and $\delta_{mis}$ are of order $1/\sqrt{\log \M}$.
%It shows theoretically that $n$ only needs to be of size  $[\log(1/\epsilon)]/[c_1\min\{\Delta^2,\delta_{mis}\}]$ to achieve probability $\epsilon$ of at least a fraction $\delta_{mis}$ mistakes, at rate that is $\Delta$ close to capacity.  With suitable target fractions of mistakes, the drop from capacity $\Delta$ is not more than $\sqrt{(1/c_1 n)\log 1/\epsilon}$. It is of order $1/\sqrt n$ if $\epsilon$ is fixed; whereas, for $\epsilon$ exponentially small in $n$, the associated drop from capacity $\Delta$ would need to be at least a constant amount.

Using the expression given above for the least squares decoder one sees that the exponent is of order $n/(\log \M)$, or equivalently $L$, using $n = (L\log \M)/R$.
 %, larger by only a loglog factor.
%of order $1/\sqrt{\log \M}$ from the fraction $\delta_u$ un-decoded. As we have seen for practical ranges of $\M$, our $\delta_u$ is comparable to $\delta_a$ at $a_{opt}$, where $\delta_a$ is within a loglog factor of $1/\log \M$.  In this regime one sees that the performance of our practical algorithm is roughly comparable to what is established for the least squares decoder.
%Of course one difference is that here we use a variable power allocation.
%In terms of performance, the difference is seen to be only a factor of order $\sqrt{\log \M}$ in the drop from capacity associated with the same code size and reliability level.
For our decoder, the error probability bound is seen to be exponentially small in $L/(\log \M)$ using the expression given in Proposition \ref{mainthm}.  This bound is within a $(\log \M)$ factor of what we obtained for the optimal least squares decoding of sparse superposition codes.

\subsection{Related work in Coding}

We point out several directions of past work that connect to what is developed here.
Modern day communication schemes, for example LDPC \cite{gallager1962low} and Turbo Codes \cite{berrou1993near}, have been demonstrated to have empirically good performance. However, a mathematical theory for their reliability is restricted only to certain special cases, for example erasure channels \cite{luby2001efficient}.

These LDPC and Turbo codes use message passing algorithms for their decoding. Interestingly, there has been recent work by \citet{bayati2010lasso} that has extended the use of these algorithms for estimation in the general high-dimensional regression setup with Gaussian $X$ matrices. Unlike our adaptive successive decoder, where we decide whether or not to select a particular term in a step, these iterative algorithms make soft decisions in each step.
However, analysis addressing rates of communication have not been given these works.
Subsequent to the present work, an alternative algorithm with soft-decision decoding for our partitioned superposition codes is proposed and analyzed by \citet{chobarron}. %However, in both these works, there is still a need to give finite sample characterizations of the distribution of statistics used in successive iterates for the purpose of error analysis for finite blocklengths.

A different approach to reliable and computationally-feasible decoding, with restriction to binary signaling, is in the work on {\em channel polarization} of \cite{arikan2009rate}, \cite{arikan2009channel}. These polar codes have been adapted to the Gaussian case as in \cite{abbe2011polar}, however, the error probability is exponentially small in $\sqrt{n}$, rather than $n$.

The ideas of {\em superposition codes}, {\em rate splitting}, and {\em successive decoding} for Gaussian noise channels began with \citet{cover1972broadcast} in the context of multiple-user channels. There, each section corresponds to the codebook for a particular user, and what is sent is a sum of codewords, one from each user.  Here we are putting that idea to use for the original Shannon single-user problem, with the difference that we allow the number of sections to grow with blocklength $n$, allowing for manageable dictionaries.

Other developments on broadcast channels by \citet{cover1972broadcast}, that we use, is that for such Gaussian channels, the power allocation can be arranged as in \eqref{eq:powerweuse} such that messages can be peeled off one at a time by successive decoding. However, such successive decoding applied to our setting would not result in the exponentially small error probability that we seek for manageable dictionaries. It is for this reason that instead of selecting the terms one at a time, we select multiple terms in a step adaptively, depending upon whether their correlation is high or not.

A variant of our regression setup was proposed by \citet{tropp2006just} for communication in the single user setup. However, his approach does not lead to communication at positive rates, as discussed in the next subsection.

There have been recent works that have used our partitioned coding setup for providing a practical solution to the Gaussian source coding problem, as in \citet{konton} and \citet{RVGaussianRD12}. A successive decoding algorithm for this problem is being analyzed by \citet{Venkatinprep}. An intriguing aspect of the analysis in \cite{Venkatinprep} is that the source coding proceeds successively, without the need for adaptation across multiple sections as needed here.

\subsection{Relationships to sparse signal recovery} \label{sec:sparseiter}

Here we comment on the relationships to high-dimensional regression. A very common assumption is that the coefficient vector is \textit{sparse}, meaning that it has only a few, in our case $L$, non-zeroes, with $L$ typically much smaller than the dimension $N$.   %or can be well approximated by a vector with few non-zeroes.
Note, unlike our communication setting, it is not assumed that the magnitude of the non-zeroes be known. Most relevant to our setting are works on \textit{support recovery}, or the recovery of the non-zeroes of $\beta$, when  $\beta$ is typically allowed to belong to a set with $L$ non-zeroes, with the magnitude of the non-zeroes being at least a certain positive value.

%Our first efforts in attempting to provide a practical decoder, reliable at rates up to capacity, involved trying to adapt existing results on convex optimization, sparse approximation, and compressed sensing. With focus on rate in comparison to capacity, the potential success and existing shortcomings of these approaches are discussed here.

%Relevant convex optimization concerns the problem of least squares convex projection onto the convex hull of a given set of vectors. If there is the freedom to multiply these vectors by a specified constant, then such convex projection is also called $\ell_1$-constrained least squares, basis pursuit \cite{chen2001atomic}, or the Lasso \cite{tibshirani1996regression}. Formulation as an $\ell_1$-penalized least squares is popular in cases of sparse statistical linear modeling and compressed sensing in which the non-zero coefficient values  are unknown, whereas $\ell_1$-constrained least squares is a more natural match to our setting in which the non-zero coefficient values are known.

Popular techniques for such problems involve relaxation with an $\ell_1$-penalty on the coefficient vector, for example in the basis pursuit \cite{chen2001atomic} and Lasso \cite{tibshirani1996regression} algorithms. %The idea with such optimization is to show with certain rate constraints and dictionary properties that the  convex projection is likely to concentrate its non-zero coefficients on the correct subset. Completion of  convex optimization to very high precision would entail a computation time in general of the order of $N^3$.
An alternative is to perform a smaller number of iterations, such as we do here, aimed at determining the target subset. Such works on sparse approximation and term selection concerns a class of iterative procedures which may be called relaxed greedy algorithms (including orthogonal matching pursuit or OMP) as studied in  \cite{jones1992simple}, \cite{barron1993universal},  \cite{pati1993orthogonal}, \cite{lee1996efficient},  \cite{barron2008approximation}, \cite{huang2008risk}, \cite{tropp2007signal}, \cite{zhang2008adaptive}, \cite{josephrandomorthogonal}. In essence, each step of these algorithms finds, for a given set of vectors, the one which maximizes the inner product with the residuals from the previous iteration and then uses it to update the linear combination. Our adaptive successive decoder is similar in spirit to these algorithms.

%Here by relaxed it is taken to mean that in updating the fit based on the  newly selected term and the terms selected in previous steps, the contribution of terms selected previously  are down-weighted.
%The relaxation property, in optimizing the linear combination with previous contributions, is that those contribution can be down-weighted in the presence of the new vector.
%These procedures %adapt to two purposes of relevance to
%the decoding task. First a relaxed greedy algorithm
%solves, to within specified precision, for the least squares convex
%projection onto the convex hull of a given set of vectors. A variant of it can also solve for the $\ell_1$-penalized least squares solution to within a given precision, as shown in \citet{huang2008risk}. %with quantification of
%the $\ell_2$ accuracy of sparse fits in $k$ steps compared to the approximation error of the least squares projection on the convex hull of vectors.

Results on support recovery can broadly be divided into two categories. The first involves giving, for a given $X$ matrix, uniform guarantees for support recovery. In other words, it guarantees, for any $\beta$ in the allowed set of coefficient vectors, that the probability of recovery is high. %More specifically, denoting  as $S$ the support $\beta$, and as $\hat S$ its estimate obtained using a certain procedure,  interest is mainly on conditions on $X$  so that
%\eqq{P_{err,\, X} = \sup_{\beta \in \cspc' }\PP_{\beta}\left(\Ecal |X\right)\label{eq:uniformrecov}}
%is small. Here $\Ecal$ is the event that $\hat S$ is not equal to $S$. Observe that if $P_{err, X}$  is small, it  gives  strong guarantees on support recovery, since it ensures that any $\beta \in \cspc'$ can be recovered with high probability.
The second category of research involves results where the probability of recovery is obtained after certain averaging, where the averaging is over a distribution of the $X$ matrix. % or over a pre-assigned distribution of $\beta$.
%In particular, one seeks to make the quantity
%\eqq{P_{err} = \sup_{\beta \in \cspc'}\PP_{\beta}\left(\Ecal \right)\label{eq:avgXsense}}
%small. Here $\PP_{\beta}\left(\Ecal \right) = E_X \PP_{\beta}\left(\Ecal |X\right)$, where the expectation on the right is over the distribution of $X$.  Notice that if \eqref{eq:uniformrecov} is small, it implies that \eqref{eq:avgXsense} is small. Thus the second category of research provides a somewhat weaker characterization of the error probability. We describe results on both approaches in the sequel.

For the first approach, a common condition on the $X$ matrix is the \textit{mutual incoherence condition}, which assumes that the correlation between any two distinct columns be small.  In particular, assuming that $\|X_j\|^2 = n$, for each $j = 1,\ldots, N$, it is assumed that,
\eqq{\frac{1}{n}\max_{j \neq j'}\left|X_j^{\trn}X_{j'}\right| \quad \mbox{\text{is $O(1/L)$}}. \label{eq:incoherence}}
Another related criterion is the \textit{irrepresentable criterion} \cite{tropp2004greed}, \cite{zhao2006model}. %, which assumes, for all subsets $T$ of size $L$, that
% \eqq{\|(X_{T}^{\trn}X_{T})^{-1}X_{T}^{\trn}X_j\|_1 < 1, \quad\mbox{for all} \quad j \in T^c \label{eq:irrepresentablecondpop}.}   Here $\|.\|_1$ denotes the $\ell_1$ norm.
  %It can be shown, see for example \cite{zhao2006model}, that conditions similar to that above are indeed necessary as well for support recovery.
  However, the above conditions are too stringent for our purpose of communicating at rates up to capacity. Indeed, for i.i.d $N(0,1)$ designs, $n$ needs to be  $\Omega(L^2\log \M)$ for these conditions to be satisfied. Here $n = \Omega(L^2\log \M)$ denotes that $n \geq const L^2\log \M$, for a  positive $const$ not depending upon $L$ or $\M$.  In other words, the rate $R$ is of order $1/L$, which goes to 0 for large $L$. Correspondingly, %we need to demonstrate support recovery under much weaker conditions.
  results from these works cannot be directly applied to our communication problem.

  As mentioned earlier, the idea of adapting techniques in compressed sensing to solve the communication problem began with \citet{tropp2006just}. However, since he used a condition similar to the irrepresentable condition discussed above, his results do not demonstrate communication at  positive rates.%, let alone rates up to capacity.

 We also remark that conditions such as \eqref{eq:incoherence}  are required by algorithms such as Lasso and OMP for providing uniform guarantees on support recovery. However, there are algorithms which provided guarantees with much weaker conditions on $X$. Examples include the iterative forward-backward algorithm \cite{zhang2008adaptive} and least squares minimization using concave penalties \cite{zhang2010nearly}. Even though these results, when translated to our setting, do imply communication at positive rates is possible, a demonstration that rates up to capacity can be achieved has been lacking.

 The second approach, as discussed above, is to assign a distribution for the $X$ matrix and analyze performance after averaging over this distribution.  \citet{wainwright2009sharp} considers $X$ matrices with rows i.i.d. $N(0, \Sigma)$, where $\Sigma$ satisfies certain conditions, and shows that recovery is possible with the Lasso with $n$ that is $\Omega(L\log \M)$. In particular his results hold for the i.i.d. Gaussian ensembles that we consider here. Analogous results for the OMP was shown by \citet{josephrandomorthogonal}. Another result in the same spirit of average case analysis is done by \citet{candès2009near} for the Lasso, where the authors assign a prior distribution to $\beta$ and study the performance after averaging over this distribution. The $X$ matrix is assumed to satisfy a weaker form of the incoherence condition that holds with high-probability for i.i.d Gaussian designs, with $n$ again of the right order.

A caveat in these discussions is that the aim of much (though not all) of the work on sparse signal recovery, compressed sensing, and term selection in linear statistical models  is distinct from the purpose of communication alone. In particular rather than the non-zero coefficients being fixed according to a particular power allocation, the aim is to allow a class of coefficients vectors, such as that described above, and still recover their support and estimate the coefficient values. The main distinction from us being that our coefficient vectors belong to a finite set, of $\M^L$ elements, whereas in the above literature the class of coefficients vectors is almost always infinite. This additional flexibility is one of the reasons why an exact characterization of achieved rate has not been done in these works.

Another point of distinction is that majority of these works focus on exact recovery of the support of the true of coefficient vector $\beta$. As mentioned before, as our non-zeroes are quite small (of the order of $1/\sqrt{L}$), one cannot get exponentially small error probabilities for exact support recovery. Correspondingly, it is essential to relax the stipulation of exact support recovery and allow for a certain small fraction of mistakes (both false alarms and failed detection). To the best of our knowledge,  there is still a need in the sparse signal recovery literature to provide proper controls on these mistakes rates to get significantly lower error probabilities.

Section \ref{sec:modificabovealgo} describes our adaptive successive decoder in full detail. Section \ref{sec:compres} describes the computational resource required for the algorithm. Section \ref{sec:analysisofalgo} presents the tools for the theoretical analysis of the algorithm, while section \ref{sec:reliabmulti} presents the theorem for reliability of the algorithm. Computational illustrations are included in section \ref{sec:compillus}. Section \ref{commcapac} proves results for the function $g_L$ of figure \ref{fig:prog7}, required for the demonstrating that one can indeed achieve rates up to capacity.
 %Section \ref{sec:reed} discusses the use of an outer Reed Solomon code to correct any mistakes from the inner decoder.
The appendix collects some auxiliary matters.

%\section{Minor Preliminaries}

%\input{algointroandperform0525.tex}
\section{The Decoder}\label{sec:modificabovealgo}

The algorithm we analyze is a modification of the algorithm described in subsection \ref{subsec:intuialgo}. The main reason for the modification is due to the difficulty in analyzing the statistics $\Zcal_{k,j}^{res}$, for $j \in J_k$ and for steps $k \geq 2$.

%that while it not hard to characterize the distribution of the statistic $\Zcal_{1,j}$, used in detection for the first step, it may not be possible to do so for the statistics $\Zcal_{k,j}^{res}$, for $j \in J_k$ and for steps $k \geq 2$.

The distribution of the statistic $\Zcal_{1,j}$, used in the first step, is easy, as will be seen below. This is because of the fact that the random variables
$$\{X_j,\, j \in J\} \quad \mbox{and} \quad Y$$
are jointly multivariate normal. However, this fails to hold for the random variables,
$$\{X_j,\, j \in J_k\} \quad \mbox{and} \quad R_k$$
used in forming $\Zcal_{k,j}^{res}$.

It is not hard to see why this joint Gaussianity fails. Recall that $R_k$ may be expressed as,
$$ R_k = Y - \sum_{j \in dec_{1,k-1}}\sqrt{P_j} X_j.$$
Correspondingly, since  the event $dec_{1,k-1}$ is not independent of the $X_j$'s, the quantities $R_k$, for $k \geq 2$, %essentially involve a random sum of the $X_j$'s, and hence
are no longer normal random vectors. It is for this reason the we introduce the following two modifications.

\subsection{The first modification: Using a combined statistic}

We overcome the above difficulty in the following manner. Recall that each
\eqq{R_k = Y - F_1 + \ldots - F_{k-1}\label{eq:rkform},}
is a sum of $Y$ and $-F_1,\, \ldots,\, -F_{k-1}$. Let $G_1 = Y$ and denote $G_k$, for $k \geq 2$, as the part of $-F_{k-1}$ that is orthogonal to the previous $G_k$'s.  In other words, perform Grahm-Schmidt orthogonalization on the vectors $Y,\, -F_1,\,\ldots,\, -F_{k-1}$, to get $G_{k'}$, with $k' = 1,\ldots,\,k$. Then, from \eqref{eq:rkform},
$$\frac{R_k}{\|R_k\|} = weight_{1}\, \frac{G_1}{\|G_1\|} + weight_2\, \frac{G_2}{\|G_2\|} + \ldots + weight_{k}\, \frac{G_{k}}{\|G_{k}\|},$$
for some weights, denoted by $weight_{k'} = weight_{k',k}$, for $k' = 1,\, \ldots,\, k$. More specifically,
$$weight_{k'} = \frac{R_k^{\trn}G_{k'}}{\|R_k\|\|G_{k'}\|},$$
and,
$$weight_1^2 + \ldots + weight_k^2 = 1.$$
Correspondingly, the statistic $\Zcal_{k,j}^{res} = X_j^{\trn} R_k/\|R_k\|$, which we want to use for $k$ th step detection, may be expressed as,
$$ \Zcal_{k,j}^{res} = weight_1\, \Zcal_{1,j} + weight_2\, \Zcal_{2,j} + \ldots + weight_{k-1}\, \Zcal_{k,j},$$
where,
\eqq{\Zcal_{k,j}= X_j^T G_k/\|G_k\|\label{eq:zcaldef}.}
Instead of using the statistic $\Zcal_{k,j}^{res}$, for $k \geq 2$, we find it more convenient to use statistics of the form,
\eqq{\Zcal_{k,j}^{comb} = \lambda_{1,k} \,\Zcal_{1,j} + \lambda_{2,k} \,\Zcal_{2,j} + \ldots +\lambda_{k,k} \,\Zcal_{k,j},\label{eq:zcombdef}}
where $\lambda_{k',k}$, for $k' = 1,\, \ldots,k$ are positive weights satisfying,
$$\sum_{k' = 1}^k \lambda_{k',k}^2  =1.$$

For convenience, unless there is some ambiguity, we suppress the dependence on $k$ and denote $\lambda_{k',k}$ as simply $\lambda_{k'}$.  Essentially, we choose $\lambda_{1}$ so that it is a deterministic proxy for $weight_1$ given above. Similarly, $\lambda_{k'}$ is a proxy for $weight_{k'}$ for $k' \geq 2$.
%This particular form of the $\lambda_{k'}$'s will be more clear in Section \ref{sec:grhm}.
The important modification we make, of replacing the random $weight_{k}$'s by proxy weights, enables us to give an explicit characterization of the distribution of the statistic $\Zcal_{k,j}^{comb}$, which we use as a proxy for $\Zcal_{k,j}^{res}$ for detection of additional terms in successive iterations.

We now describe the algorithm after incorporating the above modification. For the time-being assume that for each $k$ we have a vector of deterministic weights,
$$(\lambda_{k',k} : k' = 1,\,\ldots,\, k),$$
satisfying $\sum_{k' = 1}^k \lambda_{k'}^2 = 1$, where recall that for convenience we denote $\lambda_{k',k}$ as $\lambda_{k'}$. Recall $G_1 =Y$.%, be the normalized version for $G_1$.

For step $k = 1$, do the following
\begin{itemize}
%\item Let  $dec_0$, $dec_{1,0} = \emptyset$ and $t_{-1} = 1$ .  Start with $F_0, G_1 = Y$ and let $G_1 = Y/\|Y\|$.
%Start with $k = 1$.
%Perform the following iteratively, for steps $m \geq 1$ of the algorithm.
\item For $j \in J$, compute
$$\Zcal_{1,j} = X_j^{\trn} G_1/\|G_1\|.$$
%which is also equal to $X_j^{\trn} G_{1}^{norm}$.
To provide consistency with the notation used below, we also denote $\Zcal_{1,j}$ as
$\Zcal_{1,j}^{comb}$.
\item Update
\eqq{dec_{1} =\{j \in J: \Zcal_{1,j}^{comb} \geq \tau\}, \label{eq:dec1}}
which corresponds to the set of decoded terms for the first step. Also let $dec_{1,1} = dec_{1}.$ Update
$$F_1 = \sum_{j \in dec_1}\sqrt{P_j}X_j.$$
\end{itemize}
This completes the actions of the first step. Next, perform the following steps for $k \geq 2$, with %$k$ at most a pre-defined value $m$.
the number of steps $k$ to be at most a pre-define value $m$.
\begin{itemize}
\item Define $G_k$ as the part of $-F_{k-1}$ orthogonal to $G_1,\,\ldots,\, G_{k-1}$.
%Apply Gram-Schmidt procedure on the
%matrix
%\eqq{\left[G_{(0)} : \ldots : G_{(k-2)} : F_{(k-1)} \right],\label{toorth}}
%to obtain the matrix,
%$$\left[G_{(0)} : \ldots :  G_{(k-2)} : G_{(k-1)}\right],$$
%of orthogonal columns. Here $G_{(k-1)} = \left[G_{t_{k-2} + 1} : \ldots : G_{t_{k-1}}\right]$ is a matrix of  dimension $n\times i_{k-1}$.
%Further, let  $G_{(k-1)}^{norm}$ be the matrix obtained by normalizing the columns of
%$G_{(k-1)}$, that is
%\eqq{G_{(k-1)}^{norm} = \left[\frac{G_{t_{k-2} +  1}}{\|G_{t_{k-2} +  1}\|} : \ldots : \frac{G_{t_{k-1}}}{\|G_{t_{k-1}}\|}\right].\label{gkst}}
\item For $j \in J_k = J - dec_{1,k-1}$, calculate \eqq{ \Zcal_{k,j} = \frac{X_j^{\trn} G_{k}}{\|G_{k}\|} \label{zcalcomp}}
%In other words compute the vector,
%$$\left(\Zcal_{t_{k-2} + 1,j}, \ldots , \Zcal_{t_{k-1},j}\right) = X_j^{\trn} G_{(k-1)}^{norm}.$$
\item For $j \in J_k$, compute the combined statistic using the above $\Zcal_{k,j}$ and $\Zcal_{k',j},\, 0\leq k' \leq k-1$, given by,
$$\Zcal_{k,j}^{comb} = \lambda_{1}\,\Zcal_{1,j} +\lambda_{2}\,\Zcal_{2,j} +\ldots +\lambda_{k}\,\Zcal_{k,j},$$
where the weights $\lambda_{k'} = \lambda_{k,k'}$, which we specify later,  are positive and have sum of squares equal to 1.
\item Update \eqq{dec_{k} =\{j \in J_k: \Zcal_{k,j}^{comb} \geq \tau\},\label{abvthresh} }
which corresponds to the set of decoded terms for the $k$ th step. Also let $dec_{1,k} = dec_{1,k-1}\cup dec_{k},$ which is the set of terms detected after $k$ steps.%; It has cardinality $t_k$. %Let $k_m$ be the cardinality of $dec_{1,m}$. %Further update $J_{m+1} = J - dec_{1,m}$.

%Let $X_{dec_m}$ denote the set of columns corresponding to indices decoded in the $m$ th step. Denote them by $F_{m +  1},\ldots, \,F_{k_m}$,
%that is $[F_{k_{m-1} +  1} : \ldots: \,F_{k_m}] = X_{dec_m}.$
\item This completes the $k$ th step. Stop if either $L$ terms have been decoded, or if no terms are above threshold, or if $k =m$. Otherwise increase $k$ by 1 and repeat.
\end{itemize}

As mentioned earlier, part of what makes the above work is our ability to assign deterministic weights $(\lambda_{k,k'} : k' = 1,\ldots, k)$, for each step $k = 1,\ldots, m$. To be able to do so, we need good control on the (weigthed) sizes of the set of decoded terms $dec_{1,k}$ after step $k$, for each $k$. In particular, defining for each $j$, the quantity $\pi_j = P_j/P$, we define the size of the set $dec_{1,k}$ as $size_{1,k}$, where
\eqq{size_{1,k} = \sum_{j \in dec_{1,k}} \pi_j.\label{eq:sizeonek}}
Notice that $size_{1,k}$ is increasing in $k$, and is a random quantity which depends on the number of correct detections and false alarms in each step. As we shall see, we need to provide good upper and lower bounds for the $size_{1,1},\ldots,size_{1,k-1}$ that are satisfied with high probability, to be able to provide deterministic weights of combination, $\lambda_{k',k}$, for $k' = 1,\ldots,k$, for the $k$th step.

It turns out that the existing algorithm does not provide the means to give good controls on the $size_{1,k}$'s. To be able to do so, we need to further modify our algorithm.
%Assume that after step $k-1$, we can say that $size_{k-1}$, the weighted size of $dec_{1,k-1}$, is near a deterministic quantity then, it turns out that we can assign deterministic values for $(\lambda_{k,k'} : k' = 1,\ldots, k)$, which produce a good shift in the mean of the correct terms. To be able to do so, if provide a simple modification to the above algorithm.

\subsection{The second modification: Pacing the steps}\label{subsec:modpacing}

As mentioned above, we need to get good controls on the quantity $size_{1,k}$, for each $k$, where $size_{1,k}$ is defined as above. For this we modify the algorithm even further.

Assume that we have certain pre-specified positive values, which we call $q_{1,k}$, for $k = 1, \ldots,\,m$ .  Explicit expressions $q_{1,k}$, which are taken to be strictly increasing in $k$, will be specified later on. The weights of combination, $$(\lambda_{k',k} : k' = 1,\,\ldots,\, k),$$ for $k = 1, \ldots,\,m$, will be functions of these values.

For each $k$, denote
$$thresh_k = \{j: \Zcal_{k,j}^{comb} \geq \tau\}.$$
For the algorithm described in the previous subsection, $dec_k$, the set of decoded terms for the $k$ th step, was taken to be equal to $thresh_k$. We make the following modification:

For each $k$, instead of making $dec_k$ to be equal to $thresh_k$, take $dec_k$ to be a subset of $thresh_k$ so that the total size of the of the decoded set after $k$ steps, given by $size_{1,k}$ is near $q_{1,k}$.  The set $dec_k$ is chosen by selecting terms in $thresh_k$, in decreasing order of their $\Zcal_{k,j}^{comb}$ values, until $size_{1,k}$ nearly equal $q_{1,k}$.

In particular given $size_{1,k-1}$, one continues to add terms in $dec_k$, if possible, until
\eqq{q_{1,k} -1/L_\pi < size_{1,k} \leq q_{1,k}.\label{eq:paccriterion}}
Here  $1/L_\pi = \min_{\ell} \pi_{(\ell)}$, is the minimum non-zero weights over all sections. It is a small term of order $1/L$ for the power allocations we consider.

Of course the set of terms $thresh_k$ might not be large enough to arrange for $dec_k$ satisfying \eqref{eq:paccriterion}. Nevertheless, it is satisfied, provided
$$size_{1,k-1} + \sum_{j \in thresh_k} \pi_j \geq q_{1,k},$$
or equivalently,
\eqq{\sum_{j \in dec_{1,k-1}} \pi_j + \sum_{j \in J - dec_{1,k-1}} \pi_j 1_{\{\Zcal_{k,j}^{comb} \geq \tau\}} \geq q_{1,k}. \label{eq:pacingworkcond}}
Here we use the fact that $J_k = J - dec_{1,k-1}$.

Our analysis demonstrates that we can arrange for an increasing sequence of $q_{1,k}$, with $q_{1,m}$ near 1, such that condition \eqref{eq:pacingworkcond} is satisfied for $k = 1,\ldots, m$, with high probability.  Correspondingly, $size_{1,k}$ is near $q_{1,k}$ for each $k$ with high probability. In particular, $size_{1,m}$, the weighted size of the decoded set after the final step, is near $q_{1,m}$, which is near 1.

We remark that in \cite{barron2010ajoseph}, an alternative technique for analyzing the distributions of $\Zcal_{k,j}^{comb}$, for $j \in J_k$, is pursued, which does away with the above approach of pacing the steps. The technique in \cite{barron2010ajoseph} provides uniform bounds on the performance for collection of random variables indexed by the vectors of weights of combination. However, since the pacing approach leads to cleaner analysis, we pursue it here.

\section{Computational resource} \label{sec:compres}

%The main computation required of each step of the decoder is the computation of the inner products of the residual vectors with each column of the dictionary. Or one has computation of related statistics which require the same order of resource.

For the decoder described in section \ref{sec:modificabovealgo}, the vectors $G_k$ can be computed efficiently using the Grahm-Schmidt procedure. Further, as will be seen, the weights of combination are chosen so that, for each $k$,
$$\Zcal_{k,j}^{comb}  = \sqrt{1 - \lambda_{k,k}^2}\, \Zcal_{k-1,j}^{comb} +  \lambda_{k,k} \Zcal_{k,j}.$$
This allows us to computed the statistic $\Zcal_{k,j}^{comb}$ easily from the previous combined statistic. Correspondingly, for simplicity we describe here the computational time of the algorithm in subsection \ref{subsec:intuialgo}, in which one works with the residuals and accepts each term above threshold. Similar results hold for the decoder in section \ref{sec:modificabovealgo}.

 The inner products requires order $nL\M$ multiply and adds each step, yielding a total computation of order $nL\M m$ for $m$ steps. As we shall see the ideal number of steps $m$ according to our bounds is of order $\log \M$.

When there is a stream of strings $Y$ arriving in succession at the decoder, it is natural to organize the computations in a parallel and pipelined fashion as follows. One allocates $m$ signal processing chips, each configured nearly identically, to do the inner products. One such chip does the inner products with $Y$, a second chip does the inner products with the residuals from the preceding received string, and so on, up
to chip $m$ which is working on the final decoding step from the string received several steps before.
%Each signal processing chip has in parallel a number of simple processors, each consisting of a multiplier and and an accumulator, one for each stored column of the dictionary under consideration, with capability to provide pipelined accumulation of the required sum of products. This permits the collection of inner products to be computed online as the coordinates of the vectors are received.
After an initial delay of $m$ received strings, all $m$ chips are working simultaneously.

%Moreover, for each chip there is a collection of simple comparators, which compare the computed inner products to the threshold and store, for each column, a flag of whether it is to be part of the update. Sums of the associated columns are computed in updating the residuals (or related vectors) for the next step. The entries of that simple computation (sums of up to $L$ values) are to be provided for the next chip before processing the entries of the next received vector. If need be, to keep the runtime at a constant per received symbol, one arranges $2m$ chips, alternating between inner product chips and subset sum chips, each working simultaneously, but on strings received up to 2m steps before. The runtime per received entry in the string $Y$ is controlled by the time required  to load such an entry (or counterpart residuals on the additional chips) at each processor on the chip and perform in parallel the multiplication by the associated dictionary entries with result accumulated for the formation of the inner products.

%The terminology signal processing chip refers to computational units that run in parallel to perform the indicated tasks. Whether one or more of these computational units fit on the same physical computer chip depends on the size of the code dictionary and the current scale of circuit integration technology, which is an implementation matter not a concern at the present level of decoder description.

If each of the signal processing chips keeps a local copy of the dictionary $X$, alleviating the challenge of numerous simultaneous memory calls, the total computational space (memory positions) involved in the decoder is $nL\M m$, along with space for $L\M m$ multiplier-accumulators, to achieve constant order computation time per received symbol. Naturally, there is the alternative of increased computation time with less space; indeed, decoding by serial computation would have runtime of order $nL\M m$. Substituting $L = nR/ \log \M$ and $m$ of order $\log \M$, we may reexpress $nL\M m$ as $n^2\M$ . This is the total computational resource required (either space or time) for the sparse superposition decoder. %More precisely, to include the effect of the $snr$ on the computational resource, using the number of steps $m$ which arise in upcoming bounds, which is of order $log B$, and using $R$ upper bounded by capacity $\C$, we have the computational resource of $nL\M m$ memory positions bounded by $\C snr n^2\M$, and a number $L\M m$ of multiplier-adders bounded by $Csnr n \M$.

\section{Analysis}\label{sec:analysisofalgo}

%The following measures of performance of a step of the algorithm are important in characterizing the distribution of the statistics $\Zcal_{k,j}^{comb}$ for subsequent steps.

%$$\, = \, \sum_{j \in other \intersect J_k} 1_{\Hcal_{k,j}}.$$
%Their sum $size_k= \hat q_k  + \hat f_k $ matches $\sum_{j \in dec_k} \pi_j$.

Recall that we need to give controls on the random quantity $\hat \delta_{mis}$ given by \eqref{eq:hatdeltamis}. Our analysis leads to controls on the following weighted measures of correct detections and false alarms for a step.
Let $\pi_j = P_j/P$, where recall that $P_j = P_{(\ell)}$ for any $j$ in section $\ell$. The $\pi_j$  sums to $1$ across $j$ in $\sent$, and sums to $\M\!-\!1$ across $j$ in $\others$.  Define in general \eqq{\hat q_k = \sum_{j \in \sent \intersect dec_k} \pi_j\label{eq:hatq},} which provides a weighted measure for the number of correct detections in step $k$,  and \eqq{\hat f_k = \sum_{ j\in other \intersect dec_k} \pi_j\label{eq:hatf}}
for the false alarms in step $k$. Bounds on $\hat \delta_{mis}$ can be obtained from the quantities $\hat q_k$ and $\hat f_k$  as we now describe.

Denote
\eqq{\hatdwght = (1 - \sum_{k = 1}^m \hat q_k) + \sum_{k = 1}^m \hat f_k. \label{eq:hatdeltawght}}
An equivalent way of expressing $\hatdwght$ is the sum of $\ell$ from $1$ to $L$ of,
$$\displaystyle\sum_{j \in section\,\, \ell} \pi_j 1_{\{j\, \in\, \sent\cap dec_{1,m}^c\}}+ \pi_j1_{\{j\, \in\, \others\cap dec_{1,m}\}}.$$
In the equal power allocation case, where $\pi_j = 1/L$, one has $\hat{\delta}_{mis} \leq \hatdwght$. This can be seen by examining the contribution from a section to $\hat{\delta}_{mis}$ and $\hatdwght$. We consider the three possible cases. In the case when the section has neither an error or an erasure, its contribution to both $\hat{\delta}_{mis}$ and $\hatdwght$ would be zero. Next, when a section has an error, its contribution would be $2/L$ to $\hat{\delta}_{mis}$ from \eqref{eq:hatdeltamis}. Its contribution to $\hatdwght$ would also be $2/L$, with a $1/L$ contribution from the correct term (since it is in $\sent\cap dec_{1,m}^c$), and another $1/L$ from the wrong term. Lastly, for a section with an erasure, its contribution to $\hat{\delta}_{mis}$ would be $1/L$, while its contribution to $\hatdwght$ would be at least $1/L$. The contribution in the latter case would be greater than $1/L$ if there are multiple terms in $\others$ that were selected.

For the power allocation \eqref{eq:powerweuse} that we consider, bounds on $\hat{\delta}_{mis}$ are obtained by multiplying $\hatdwght$ by the factor $snr/(2\Capacity).$
To see this, notice that for a given weighted fraction, the maximum possible un-weighted fraction would be if we assume that all the failed detection or false alarms came from the section with the smallest weight. This would correspond to the section with weight $\pi_{(L)}$, where it is seen that $\pi_{(L)} = 2\Capacity/(L\,snr).$
Accordingly, if $\dwght$ were an upper bound on $\hatdwght$ that is satisfied with high probability, we take
\eqq{\delta_{mis} = \frac{snr}{2\Capacity}\, \dwght\label{eq:deltamisdeltawghtrel},}
so that $\hat \delta_{mis} \leq \delta_{mis}$ with high probability as well.

Next, we characterize, for $k \geq 1$, the distribution of $\Zcal_{k,j}$, for $j \in J_k$. As we mentioned earlier, the distribution of $\Zcal_{1,j}$ is easy to characterize. Correspondingly, we do this separately in the next subsection. In subsection \ref{subsec:distanal} we provide the analysis for the distribution of $\Zcal_{k,j}^{comb}$, for $k \geq 2$.

\subsection{Analysis of the first step} \label{subsec:firststepanalysis}
In Lemma \ref{lem:firststepdist} below we derive the distributional properties of $(\Zcal_{1,j} : j \in J)$.  Lemma \ref{laterstepdist},
in the next subsection, characterizes the distribution of $(\Zcal_{k,j} : j \in J_k)$ for steps $k \geq 2$ .

Define \eqq{\cjrb = n \,\pi_j \,\nu\label{eq:cjrb},}
where $\nu = P/(P + \sigma^2)$.
For the constant power allocation case, $\pi_j$ equals $1/L$. In this case $\cjrb = (R_{0}/R) \, 2\log \M$ is the same for all $j$.

For the power allocation \eqref{eq:powerweuse}, we have $$\pi_j = e^{-2{\Capacity}(\ell-1)/L} (1\!-\!e^{-2{\Capacity}/L})/(1\!-\!e^{-2\Capacity}),$$ for each $j$ in section $\ell$.  Let \eqq{\tilde {\Capacity}= (L/2)[1-e^{-2{\Capacity}/L}] \label{eq:ctilda},}
which is essentially identical to $\Capacity$ when $L$ is large.  Then for $j$ in section $\ell$, we have
\eqq{\cjrb \, = (\tilde \Capacity/R)\, e^{-2{\Capacity}(\ell-1)/L}(2\log\M).\label{eq:cjrbexpo}}
%whether the factor $\tilde\Capacity/R$ is greater than $1$ or less than $1$ depends on where the value of $R$ is relative to $\tilde \Capacity.$
%Note, the $C_{j,R}$ simplifies to the value $e^{-2{\Capacity}(\ell-1)/L}$ when the rate is $R=\tilde \Capacity$.

We now are in a position to give the lemma for the distribution of $\Zcal_{1,j}$, for $j \in J$. The lemma below shows that each $\Zcal_{1,j}$ is distributed as a shifted normal, where the shift is approximately equal to $\sqrt{\cjrb}$ for any $j$ in $\sent$, and is zero for $j$ in $\others$. Accordingly, for a particular section, the maximum of the $\Zcal_{1,j}$, for $j \in \others$, is seen to be approximately $\sqrt{2\log\M}$, since it is the maximum of $\M - 1$ independent standard normal random variables. Consequently, one would like $\sqrt{\cjrb}$ to be at least $\sqrt{2\log \M}$ for the correct term in that section to be detected.

\vspace{0.3cm}
\begin{lem}\label{lem:firststepdist} For each $j \in J$, the statistic $\Zcal_{1,j}$ can be represented as
$$ \sqrt{\cjrb}\, (\Chi_n/\sqrt{n}) 1_{\{j \: \sent\}} + Z_{1,j},$$
where $Z_1 = (Z_{1,j} : j \in J_1)$ is multivariate normal $N(0,\Sigma_{1})$,
with $\Sigma_1 = I - \delta_1 \delta_1^{\trn}/P$, where $\delta_1 = \nu \beta$.

 Also,
$$\Chi_n^2 = \frac{\|Y\|^2}{\sigma_Y^2}$$
is a Chi-square $n$ random variable that is independent of $Z_1 = (Z_{1,j} : j \in J)$. Here $\sigma_Y = \sqrt{P + \sigma^2}$ is the standard deviation of each coordinate of $Y$.
%Further, the covariance matrix . %can be expressed as $\Sigma_1 = I - \beta \beta^T/\sigma_Y^2$.
\end{lem}

\begin{proof}
 Recall that the $X_j$, for $j$ in $J$, are independent $N(0,I)$ random vectors and that $Y = \sum_j \beta_j X_j \, + \, \varepsilon$, where the sum of squares of the $\beta_j$ is equal to $P$

%Consider the decomposition of each random vector $X_j$ of the dictionary into a vector in the direction of the received $Y$ and a vector $U_j$ uncorrelated with $Y$.
The conditional distribution of each $X_j$ given $Y$ may be expressed as,
\eqq{X_j = \beta_j \, Y/\sigma_Y^2 \, + \, U_j,\label{eq:xjrep}}
where $U_j$ is a vector in $\Rbb^N$ having a multivariate normal distribution. Denote $b = \beta/\sigma_Y$. It is seen that $$U_j \sim N_n\left(0, (1 - b_j^2) I\right),$$ where $b_j$ is the $j$ th coordinate of $b$.

Further, letting $U = [U_1 : \ldots: U_N]$, it follows from the fact that the rows of $[X : \epsilon/\sigma]$ are i.i.d, that the rows of the matrix $U$ are i.i.d.
 %where the coefficient is $b_{1,j} = \E[X_{i,j}Y_i]/\sigma_Y= \beta_j/\sigma_Y$, which indeed makes each coordinate of $U_j$ uncorrelated with each coordinate of $Y$.  These coefficients collect into a vector $b=b_1 = \beta /\sigma_Y$ in $\Rbb^N$.  In the case that the magnitude of the non-zero coefficients is $\sqrt{P/L}$, then for terms $j$ sent, the square of $b_{1,j}$ is equal to $\,{\nu}/{L}$.
%, where $\nu = P/(\sigma^2\+P)$.
%In any case, since $\|\beta\|^2=P$, the sum of squares $\|b_1\|^2$ is equal to $$\nu = P/(\sigma^2\+P).$$

%The subscript $1$ on first step quantities such as %$b_1$ and $\delta_1$ and
%$b_{1,j}$ is to distinguish them from corresponding values that will arise on the subsequent steps.
%Likewise, when distinguishing $U_j=U_{1,j}$ from a corresponding $U_{k,j}$ for $k \geq 2$, we add the subscript.
%Subscripts may be dropped when clear from the context.

%These first step vectors $U_j = X_j - b_{1,j} Y/\sigma_Y$ along with $Y$ are linear combinations of joint normal random variables and so are also joint normal, with zero correlation implying that $Y$ is independent of the collection of $U_j$.
%Using this decomposition of $X_j$ we likewise have the decomposition of the inner product $$X_j \cdot Y = c_{j} |Y|^2 \, + \, U_j \cdot Y.$$
%The independence of $Y$ and $U_j$ facilitates development of distributional properties of the $U_j^T Y$.  For these purposes we need the characteristics of the joint distribution of the $U_j$ across terms $j$ (clearly there is independence for distinct time indices $i$).

 Further, for row $i$ of $U$, the random variables $U_{i,j}$ and $U_{i,j'}$ have mean zero and expected product $$1_{\{j=j'\}} - b_{j} b_{j'}.$$ %For the constant power allocation case, this value is $1- (\nu/L) \,1_{j \: \sent}$ when $j'\!=\!j$ and it is a small covariance $- (\nu/L) \, 1_{j,j' \: \sent}$ when $j' \neq j$.
 In general, the covariance matrix of the $i$th row of $U$ is given by $\Sigma_1$.

For any constant vector $\alpha\ne 0$, consider $U_j^T \alpha/\|\alpha\|$.
Its joint normal distribution across terms $j$ is the same for any such $\alpha$. Specifically, it is a normal $N(0,\Sigma_1)$, with mean zero and the indicated
covariances.
% $1_{\{j=j'\}} -\delta \,1_{j \: \sent}$.
%Letting $\sigma_\delta$ denote the collection of all these %variances and covariances, this is denoted the normal $N(0,\sigma_\delta)$ distribution.  It is also denoted $N_J (0,\sigma_\delta)$ when making explicit that it is the joint distribution of these variables for indices $j$ in $J$.

Likewise define %$Z_j = \sqrt {n} (U_j \cdot Y)/|Y|$.
$Z_{1,j} = U_j^T Y/\|Y\|$. %, also denoted $Z_{1,j}$ when making explicit that it is for the first step. %$= \sqrt n U_j \cdot Y /(\sqrt{P/L} |Y|)$.
Conditional on $Y$, one has that jointly across $j$, these $Z_{1,j}$ have the normal $N(0,\Sigma)$ distribution.
Correspondingly, $Z_1 = (Z_{1,j} : j \in J)$ is independent of $Y$, and has a $N(0, \Sigma_1)$ distribution unconditionally.
  %Indeed, since the $U_j$ are independent of $Y$, when we condition on $Y=\alpha$ we get the same $N(0,\Sigma)$ distribution, and since this conditional distribution does not depend on $Y$, it is the unconditional distribution as well.

Where this gets us is revealed via the representation of the inner product $\Zcal_{1,j} = X_j^T Y/\|Y\|$, which using \eqref{eq:xjrep}, is given by,
%\eqq{X_j^T Y \,=\, \beta_j \, \frac{\|Y\|^2}{\sigma_Y^2}
%\delta_L |Y|^2
%\, + \, \|Y\| \, Z_{1,j}.\label{eq:xjrep}}
%Correspondingly,
$$\Zcal_{1,j} =  \, \beta_j \, \frac{\|Y\|}{\sigma_Y^2}
%\delta_L |Y|^2
\, +  Z_{1,j}.$$
The proof is completed by noticing that for $j \in \sent$, one has $\sqrt{\cjrb} = \beta_j\sqrt{n}/\sigma_Y$.
%Accordingly, using the above decomposition of $X_j$, the inner product of interest has the representation $X_j \cdot Y = c_{L,j} |Y|^2 + U_j \cdot Y$.  Here $c_{L,j} |Y|^2 = (P/L) [|Y|^2 /(\sigma^2+P)] 1_{j \: \sent}$, approximately equal to $(P/L) 1_{j \: \sent}$, in concert with the shift of the distribution of the inner product depending on whether $j$ is in $\sent$.
%This identifies the distribution of the $X_j^T Y$ as that obtained as a mixture of the normal $Z_j$ with scale and location shifts determined by an independent random variable $\Chi_n^2= \|Y\|^2/\sigma_Y^2$, distributed as Chi-square with $n$ degrees of freedom.
%Divide through by $\|Y\|$ to normalize these inner products to a helpful scale and to simplify the distribution of the result to be only that of a location mixture of normals. The resulting random variables $\Zcal_{1,j} = X_j^T Y \,/\|Y\|$ %, also denoted $\Zcal_{1,j}$,
%take the form
%$$\Zcal_{1,j} \,=\, \sqrt n \, b_{1,j} \, {|Y|/\sigma_Y} \; +\; Z_j,$$
%where $|Y|/\sigma_Y = \Chi_n /\sqrt{n}$ is near $1$. Note that $\sqrt{n}b_{1,j} = \sqrt{n}\beta_j/\sigma_Y$
%which is $\sqrt{n\pi_j\nu}$ or $\sqrt{C_{j,R,\M}}$.
%This completes the proof of Lemma \ref{lem:firststepdist}.
\end{proof}
%Notice that for the constant power allocation case, $b_{1,j} = \sqrt{\nu/L} \, 1_{j \: \sent}$, %$b_{1,j} = \sqrt{(P/L)/(P\+\sigma^2)} \, 1_{j \: \sent}$,
%leading to
%$$\Zcal_{1,j} \,=\, \sqrt{C_{R,\M}} \, \frac{\|Y\|}{\sqrt{n}\sigma_Y} \, 1_{j \: \sent} \; +\; Z_j,$$
%where recall that $C_{R,\M} = 2 (R_{0}/R) \log \M$. 

\subsection{Analysis of steps $k \geq 2$} \label{subsec:distanal}

%We now analyze the above algorithm. The analysis uses the conditioning technique described in the paper \cite{Montanari}
%by Bayati and Montanari. It can also be analyzed by the techniques in \cite{BarronJosephFast}.

We need the characterize the distribution of the statistic $\Zcal_{k,j}^{comb},\, j\in J_k$, used in decoding additional terms for the $k$th step.
%The analysis uses the techniques described in the paper \cite{Montanari} by Bayati and Montanari. It can also be analyzed by the techniques in \cite{BarronJosephFast}.

The statistic $\Zcal_{k,j}^{comb},\, j\in J_k$, can be expressed more clearly in the following manner. For each $k \geq 1$, denote,
$$\Zcal_{k} = X^{\trn} \frac{G_k}{\|G_k\|}.$$
%Each $\Zcal_{k}$ is a matrix of dimension $N \times i_{k-1}$. We also express $\Zcal_{(k)}$ as $$\Zcal_{(k)} = [\Zcal_{t_{k-2} + 1} : \ldots : \Zcal_{t_{k-1}}],$$
%with the $j$ th element of the vector $\Zcal_l$ written as $\Zcal_{l,j}$.

Further, define $$\Zcal_{1,k} = [\Zcal_{1} : \Zcal_{2} : \ldots : \Zcal_{k}]$$ and let $\Lambda_{k} = (\lambda_{k,1},\lambda_{k,2}, \ldots,\,\lambda_{k,k})^{\trn}$ be the deterministic vector of weights of combinations used for the statistics $\Zcal_{k,j}^{comb}$. Then $\Zcal_{k,j}^{comb}$ is simply the $j$ th element of the vector $$\Zcal_{k}^{comb} = \Zcal_{1,k}\Lambda_k.$$ We remind that for step $k$ we are only interested in elements $j \in J_k$, that is, those that were not decoded in previous steps.

Below we characterize the distribution of $\Zcal_{k}^{comb}$ conditioned on the what occurred on previous steps in the algorithm. More explicitly, we define $\Fcal_{k-1}$ as
\eqq{\Fcal_{k-1} = (G_1,\, G_2,\ldots, G_{k-1}, \Zcal_1,\ldots, \Zcal_{k-1}),\label{eq:fkmone}}
or the associated $\sigma$-field of random variables.
%the sigma-field generated by the random variables, $G_1,\, G_2,\,\ldots,\, G_{k-1}$, as well as the statistics $\Zcal_{1},\,\ldots,\, \Zcal_{k-1}$. %Random variable measurable with respect to $\Fcal_{k-1}$ would
This represents the variables computed up to step $k-1$. Notice that from the knowledge of $\Zcal_{k'}$, for $k' = 1,...,k-1$, one can compute $\Zcal_{k'}^{comb}$, for $k' < k$. Correspondingly, the set of decoded terms $dec_{k'}$, till step $k-1$, is completely specified from knowledge of $\Fcal_{k-1}$.

Next, note that in $\Zcal_{1,k}$, only the vector $\Zcal_{k}$ does not belong to $\Fcal_{k-1}$. Correspondingly, the conditional distribution of $\Zcal_{k}^{comb}$ given $\Fcal_{k-1}$, is described completely by finding the distribution of $\Zcal_{k}$ given $\Fcal_{k-1}$. Accordingly,  we only need to characterize the conditional distribution of $\Zcal_k$ given $\Fcal_{k-1}$.

%Before doing so, we introduce the following notation.

%\textbf{Notation:}
%Let $H$ be an $N_1 \times N_2$ dimensional matrix. For a set $J'$ that is subset of $\{1,\, 2,\,\ldots,\, N_1\}$, we denote  $[H]_{J'}$ as
%the sub-matrix of $H$ with rows in $J'$. We also use $H_{J'}$  for $[H]_{J'}$. Similarly, $H_j$ and $[H]_j$ would denote the $j$ th row of $H$. If $N_2 = 1$, then $H_{J'}$  and $[H]_{J'}$ would denote the corresponding column vector with indices in $J'$.  For example, the vector corresponding to indices $j \in J_k$ of $\Zcal_k^{comb}$ is denoted by $[\Zcal_{k}^{comb}]_{J_k}$ or $\Zcal_{k,J_k}^{comb}$. Similarly, $[\Zcal_{k}]_{J_k}$ and $\Zcal_{k, J_k}$
%denotes the sub-matrix of $\Zcal_{k}$ comprising of rows with indices in $J_k$.

%If $H$ is a square $N_1 \times N_1$ dimensional matrix, we denote by $\left(H\right)_{J'}$ the $|J'|\times |J'|$ sub-matrix  of $H$ with rows and columns in $J'$.

%Next, we state a lemma showing that for $k \geq 2$, the distribution of $\Zcal_{k,j}$, with $j \in J_k$, can also be expressed in a manner similar to that in Lemma \ref{lem:firststepdist}. In particular, $\Zcal_{k,j}$ can be expressed as a normal random variable $Z_{k,j}$ plus a location shift depending on whether $j$ is in $\sent$ or not.

%Notice that we maintain the pattern used in Lemma \ref{lem:firststepdist} and use $\Zcal_{k,j}$ to denote the test statistics that incorporate the shift for $j$ in $\sent$ and standard font $Z_{k,j}$ to denote their counterpart mean zero normal random variables before the shift.

 Initializing with the distribution of $\Zcal_1$ derived in Lemma \ref{lem:firststepdist}, we provide the conditional distributions $$\Zcal_{k, J_k}= (\Zcal_{k,j} : j \in J_k),$$ for $k = 2,\ldots,\,n$.  As in the first step, we show that the distribution of $\Zcal_{k,J_k}$ can be expressed as the sum of a mean vector and a multivariate normal noise vector $Z_{k,J_k} = (Z_{k,j} : j \in J_k)$.
 The algorithm will be arranged to stop long before $n$, so we will only need these up to some much smaller final $k=m$. Note that $J_{k}$ is never empty because we decode at most $L$, so there must always be at least $(\M\!-\!1)L$ remaining.

 The following measure of correct detections in step, adjusted for false alarms, plays an important role in characterizing the distributions of the statistics involved in an iteration.
 Denote
\eqq{\hat q_k^{adj} = \frac{\hat q_k}{1 + \hat f_k/\hat q_k} \label{eq:qkadj},}
where $\hat q_k$ and $\hat f_k$ are given by \eqref{eq:hatq} and \eqref{eq:hatf}.

 In the lemma below we denote $N_{J_k} (0,\Sigma)$ to be multivariate normal distribution with dimension $|J_k|$, having mean zero and covariance matrix $\Sigma$, where $\Sigma$ is an $|J_k| \times |J_k|$ dimensional matrix. Further, we denote $\beta_{J_k}$ to be the sub-vector of $\beta$ consisting of terms with indices in $J_k$.

\begin{lem}\label{laterstepdist} For each $k \geq 2$, the conditional distribution of $\Zcal_{k,j}$ , for $j\in J_k$, given $\Fcal_{k-1}$ has the representation \eqq{\, \sqrt{\hat w_k \, \cjrb} \,(\Chi_{d_k}/\sqrt{n})\,1_{\{j\, \in\, \sent\}} \, +\, Z_{k,j}. \label{mainrel}}
Recall that $\cjrb = n\pi_j\nu$. Further, $\hat w_k$
%$ = (1\!-\!(k-1)/n) \hat w_k^{adj}$, where the $\hat w_k^{adj}$
$= \hat s_k\!-\!\hat s_{k-1}$, which are increments of a series with total
$$\hat w_1 +\hat w_2 + \ldots + \hat w_k \, = \,\hat s_k = \frac{1}{1-\hat q_{k-1}^{adj,tot} \, \nu},$$ where \eqq{\hat q_{k}^{adj,tot} = \hat q_1^{adj} + \ldots + \hat q_k^{adj}.\label{eq:qhatadjtot}}  Here $\hat w_1 = \hat s_1=1$.   The quantities $\hat q_k^{adj}$ is given by \eqref{eq:qkadj}.

The conditional distribution $\PP_{Z_{k,J_k}|\Fcal_{k-1}}$  %have the joint
is normal $N_{J_k}(0,\Sigma_{k})$, where the covariance $\Sigma_k$ has the representation $$\Sigma_k = I- \delta_k \delta_k^T/P,\quad\text{where $\delta_{k} = \sqrt{\nu_k}\, \beta_{J_k}$}.$$  %$=\Sigma(J_k,B_k)$ associated with the $N(0,\Sigma_k)$ distribution
%That is %$(\Sigma_k)_{j,j'} = 1_{j=j'} - B_{k,j}B_{k,j'}$, %$ 1_{j,j' \: \sent}$,
%$(\Sigma_k)_{j,j'}= 1_{j=j'} - \delta_{k,j} \delta_{k,j'}$,
%$1_{j=j'} - \delta_k 1_{j,j' \: \sent}$,
%for $j,j'$ in $J_k$, where the vector $\delta_k$ is in the direction $\beta$, with
Here $\nu_k = \hat s_k \nu$.

Define %$\Chi_{d_l}^2 = \|G_l\|^2/\sigma_l^2$, with
$\sigma_k^2 = \hat s_{k-1}/\hat s_k$. The $\Chi_{d_k}$ term appearing in \eqref{mainrel} is given by $$\Chi_{d_k}^2 = \frac{\|G_k\|^2}{\sigma_k^2}.$$
Also, the distribution of $\Chi_{d_k}^2$ given $\Fcal_{k-1}$, is chi-square with
$d_k = n-k + 1$ degrees of freedom, and further, it is independent of $Z_{k,J_k}$.

%Further, $\Chi_{d_l}^2 = \|G_l\|^2/\sigma_l^2$ is $\Fcal_{k-1}$ measurable. When conditioned on $\Fcal_{k-2,l-1}$ it has a Chi-square distribution with $d_k = n-l$ degrees of freedom.

%$\delta_k = \hat s_k \, \nu/L$.
\end{lem}

The proof of the above lemma is considerably more involved. It is given in Appendix \ref{sec:disanalzcalk}. From the above lemma one gets that $\Zcal_{k,j}$ is the sum of two terms - the `shift' term and the `noise' term $Z_{k,j}$. The lemma also provided that the noise term is normal with a certain covariance matrix $\Sigma_k$.

Notice that Lemma \ref{laterstepdist} applies to the case $k = 1$ as well, with $\Fcal_{0}$ defined as  empty, since $\hat w_k = \hat s_k = 1$. The definition of $\Sigma_1$ using the above lemma is the same as that given in Lemma \ref{lem:firststepdist}. Also note that the conditional distribution of $(\Zcal_{k,j} : j \in J_k)$, as given in Lemma \ref{laterstepdist}, depends on $\Fcal_{k-1}$ only through the $\|G_1\|,\, \ldots,\, \|G_{k-1}\|$ and $(\Zcal_{k',j} : j \in J_{k'})$ for $k' < k$.

In the next subsection, we demonstrate that $Z_{k,j}$, for $j \in J_k$, are very close to being independent and identically distributed (i.i.d.).

\subsection{The nearby distribution} \label{subsec:nearbydist}
Recall that since the algorithm operates only on terms not detected previously, for the $k$ step we are only interested in terms in $J_k$. The previous two lemmas specified conditional distributions of $\Zcal_{k,j}$, for $j \in J_k$. However, for analysis purposes we find it helpful to assign distributions to the $\Zcal_{k,j}$, for $j \in J - J_k$ as well. In particular, conditional on $\Fcal_{k-1}$, write
\eqs{\Zcal_{k,j} = \sqrt{\hat w_k \, \cjrb} \,\left(\frac{\Chi_{d_k}}{\sqrt{n}}\right)1_{\{j\in \sent\}}  + Z_{k,j} \quad\text{for $j \in J$}.}

Fill out of specification of the distribution assigned by $\PP$, via a sequence of conditionals $\PP_{Z_{k}|\Fcal_{k-1}}$ for $Z_{k} = (Z_{k,j} : j \in J)$, which is for all $j$ in $J$, not just for $j$ in $J_k$. %with $1_{\Hcal_j}=0$.
%Here $\Fcal_k^{full} = (Z_{k',J},\|G_{k'}\|^2/\sigma_{k'}^2 : k'=1,2,\ldots,k)$.
For the variables $Z_{k,J_k}$ that we actually use, the conditional distribution is that of $\PP_{Z_{k,J_k}|\Fcal_{k-1}}$ as specified in Lemmas \ref{lem:firststepdist} and \ref{laterstepdist}.  Whereas for the $Z_{k,j}$ with $j\in J-J_k$, given $\Fcal_{k-1}$, we conveniently arrange them to be independent standard normal. This definition is contrary to the  true conditional distribution of $\Zcal_{k,j}$ for $j \in J - J_k$, given $\Fcal_{k-1}$. However, it is a simple extension of the conditional distribution that shares the same marginalization to the true distribution of $(Z_{k,j} : j \in J_k)$ given $\Fcal_{k-1}$.

Further a simpler approximating distribution $\QQ$ is defined. Define $\QQ_{Z_{k}|\Fcal_{k-1}}$ to be independent standard normal. Also, like $\PP$, the measure $\QQ$ makes the $\Chi_{d_k}^2$ appearing in $\Zcal_{k,j}$, Chi-square$(n\!-\!k\+1)$ random variables independent of $Z_{k}$, conditional on $\Fcal_{k-1}$.

%As in Lemma 2, the $\alpha$ divergence between the conditional normal distributions $\PP_{Z_{2,J_2}|\Fcal_1}$ and $\QQ_{Z_{2,J_2}|\Fcal_1}$ has the upper bound $$D_\alpha (\PP_{Z_{2,J_2}|\Fcal_1} \| \QQ_{Z_{2,J_2}|\Fcal_1}) \le -\frac{1}{2} \log (1\!-\!\nu_2^{red}) \le c_0.$$ where $\nu_2^{red} = L_2 \delta_2$, with $L_2 = (1\!-\!\hat q_1)L$ equal to the number of terms in $J_2 \intersect \sent$. As we showed, $\nu_2^{red} \le \nu_1$, where $\nu_1 = P/(\sigma^2+P)$ and $c_0 =-(1/2)\log(1\!-\!\nu_1)$.

%An event $A$ is said to be determined by $\Fcal_k$ if its indicator is a function of $\Fcal_k$.  As $\Fcal_k =(Z_{k',J_{k'}},\Chi_{n-k'+1}^2 :k'=1,2,\ldots,k)$, where the index sets $J_k$ are random, given as functions of preceding $\Fcal_{k-1}$, it might be regarded as a tricky matter.  %Here define $A$ to be
%Alternatively a random variable may be said to be determined by $\Fcal_k$
%if it is measurable with respect to the collection of random variables $(\|G_{k'}\|^2/\sigma_{k'}^2, \, Z_{k',j} 1_{\{j \in dec_{1,k'-1}^c\}}, j \in J, \, 1\!\le\! k'\!\le\! k)$.  The multiplication by the indicator removes the effect on step $k'$ of any $Z_{k',j}$ decoded on earlier steps, that is, any $j$ outside $J_{k'}$.  Operationally, no advanced measure-theoretic notions are required, as we are working with sequences of conditional densities of explicit Gaussian form. %We have the following.

In the following lemma we appeal to a sense of closeness of the distribution $\PP$ to $\QQ$, such that events exponentially unlikely under $\QQ$ remain exponentially unlikely under the governing measure $\PP$.

\vspace{0.3cm}
\begin{lem}\label{measurebdd2} For any event $A$ that is determined by the random variables,
\eqq{\text{$\|G_1\|,\, \ldots,\, \|G_{k}\|$ and $(\Zcal_{k',j} : j \in J_{k'})$, for $k' \leq k$},\label{eq:randdet}}
one has
%(measurable with respect to) the collection of random variables $\Zcal_{1,j}$ and $Z_{2,j} 1_{H_j^c}$ for $j \in J$,
$$\PP[A] \le \QQ[A] e^{k c_0},$$
where $c_0 = (1/2)\log (1 + P/\sigma^2).$ %The analogous statement holds more generally for the expectation of any non-negative function of $\Fcal_k$.
\end{lem}

For ease of exposition we give the proof in Appendix \ref{apnearby}. Notice that the set $A$ is $\Fcal_k$ measurable, since the random variables that $A$ depends on are $\Fcal_k$ measurable.

\subsection{Separation analysis} \label{subsec:alterapproach}
Our analysis demonstrates that we can give good lower bounds for $\hat q_k$, the weighted proportion of correct detection in each step, and good upper bounds on $\hat f_k$, which is the proportion of false alarms in each steps.

 Denote the exception events $$\aq{k} = \{\hat q_k < q_k \} \quad\text{and}\quad\af{k} = \{\hat f_k > f_k\}.$$
Here the $q_k$ and $f_k$ are deterministic bounds for the proportion of correct detections and false alarms respectively, for each $k$. These will be specified in the subsequent subsection. %For convenience we suppress  the dependence on $k$ in these sets.

Assuming that we have got good controls on these quantities up to step $k - 1$, we now describe our characterization of $\Zcal_{k,j}^{comb}$, for $j \in J_k$, used in detection for the $k$th step. Define the exception sets $$\aq{1,k-1} =  \displaystyle\cup_{k' = 1}^{k-1} \aq{k'} \quad \text{and} \quad \af{1,k-1} =
\cup_{k' = 1}^{k-1}\af{k'}.$$
%From the arguments of the previous subsection one sees that outside of $A_q \cup A_f$, one has
%$$ \hat q_{k'} \geq q_{k'} \quad \text{for $k' = 1, \ldots, k-1 $} $$
%and
%$$ \hat f_{k'} \leq f_{k'} \quad \text{for $k' = 1, \ldots, k-1 $}.$$

%, where, $$ \hat q_{1,k}^{above} := \sum_{j \in \sent} \pi_j 1_{\Hcal_{k,j}}.$$ The satisfaction of this inequality is sufficient for the provision of sets of decoded terms satisfying the above requirement, which is a consequence of the observation that $$\sum_{j \in dec_{1,k-1}} \pi_j  + \sum_{j \in J - dec_{1,k-1}} \pi_j 1_{\Hcal_{k,j}} \geq \hat q_{1,k}^{above}.$$

%This $\hat q_{1,k}^{above}$ corresponds to the quantity studied in the previous section, giving the weighted total of terms in sent for which the combined statistic is above threshold, and it remains likely that it exceeds the purified statistic $\hat q_{1,k}$. What is different is the control on the size of the previously decoded sets allows for constant weights of combination.

The manner in which the quantities $\hat q_1, \ldots, \hat q_k$ and $\hat f_1,\ldots \hat f_k$ arise in the distributional analysis of Lemma \ref{laterstepdist} is through the sum
$$\hat q_k^{adj,tot}=\hat q_1^{adj} + \ldots + \hat q_k^{adj}$$
of the adjusted values $\hat q_k^{adj} = \hat q_k/(1+\hat f_k/\hat q_k)$. Outside of $\aq{1,k-1} \cup \af{1,k-1}$, one has \eqq{\hat q_{k'}^{adj} \geq q_{k'}^{adj}\quad\text{for $k' = 1, \ldots, k-1 $},\label{eq:adjrel}}
where, for each $k$,
$$q_k^{adj} = q_k/(1 + f_k/q_k).$$
Recall that from Lemma \ref{laterstepdist} that,
$$\hat w_k =\frac{1}{1- \hat q_{k-1}^{adj,tot} \nu} - \frac{1}{1- \hat q_{k-2}^{adj,tot} \nu}.$$
From relation \eqref{eq:adjrel}, one has $\hat w_{k'} \geq w_{k'}$, for $k' = 1,\ldots, k$, where $w_1 = 1$, and for $k >1$,
$$w_k = \frac{1}{1- q_{k-1}^{adj,tot} \nu} - \frac{1}{1- q_{k-2}^{adj,tot} \nu}.$$
Here, for each $k$, we take $q_{k}^{adj,tot} = q_1^{adj} + \ldots + q_k^{adj}$.

Using this $w_k$ we define the corresponding vector of weights $(\lambda_{k',k}: k' = 1,\ldots,k)$, used in forming the statistics $\Zcal_{k,j}^{comb}$, as
$$\lambda_{k',k}=\sqrt{\frac{w_{k'}}{w_1+w_2+\ldots+w_k}}.$$

Given that the algorithm has run for $k - 1$ steps, we now proceed to describe how we characterize the distribution of $\Zcal_{k,j}^{comb}$ for the $k$th step. Define the additional exception event $$\ah{1,k-1} = \union_{k'=1}^{k-1} \ah{k'}, \quad \text{with}\quad \ah{k} = \{\Chi_{d_{k}}^2/n \le 1\!-\!h\},$$
where $0 < h <1$. Here the term $\Chi_{d_k}^2$ is as given in Lemma \ref{laterstepdist}. It follows a Chi-square distribution with $d_k = n-k+1$ degrees of freedom. Define $$\aee{k-1} = \aq{1,k-1} \cup \af{1,k-1} \cup \ah{1,k-1}.$$
% or equivalently $\union_{k'=1}^k \{\Chi_{n\!-\!k'\+1}^2/(n\!-k'\!+1) \le (1\!-\!h_{k'})\}$ where $h_{k'}$ is related to $h$ by $(n\!-\!k'\+1)(1\!-\!h_{k'}) = n(1\!-\!h)$.
%Moreover we now arrange $$\tilde_h \, =\, 1\, - \, \left(1- h\,\right)^2.$$
%As a consequence of lemmas \ref{lem:firststepdist} and \ref{laterstepdist},
Notice that we have for $j \in \sent$ that
$$\Zcal_{k',j} = \sqrt{\hat w_{k'} \cjrb} \, (\Chi_{d_{k'}}/\sqrt{n}) + Z_{k',j}$$
%Define
%\eqq{\cjrbh = C_{j,R,\M}(1 - h)\label{eq:cjrbh}.}
and for $j \in \others$, we have $$\Zcal_{k',j} = Z_{k',j},$$
for $k' = 1,\ldots, k$. Further, denote $\cjrbh = \cjrb(1 - h)$. Then
on the set $\aee{k-1}^c\cap \ah{k}^c$, we have for $k' = 1,\ldots,k$ that%that,%that $\Zcal_{1,j}$ is at least %$\Zcal_{1,j}^{h}=
%$$\Zcal_{1,j} \geq \sqrt{\cjrbh} \, \,+\, Z_{1,j},$$ and for $k'\ge 2$,% that
%$\Zcal_{k',j}$ is at least %$\Zcal_{k',j}^{h}$  defined by
%$$\Zcal_{k',j}^{pure} =
$$\Zcal_{k',j} \geq \sqrt{ w_{k'}} \, \sqrt{\cjrbh} \,\, \,+\,\, Z_{k',j}\quad\text{for $j \in \sent$}.$$
Recall that, $$\Zcal_{k,j}^{comb}=\lambda_1\, \Zcal_{1,j}\, +\, \lambda_2\, \Zcal_{2,j} \,+\, \ldots\, +\, \lambda_k \,\Zcal_{k,j},$$
where for convenience we denote $\lambda_{k',k}$ as simply $\lambda_{k'}$.
Define for each $k$ and $j \in J$, the combination of the noise terms by
$$Z_{k,j}^{comb}=\lambda_1 Z_{1,j} + \lambda_2 Z_{2,j} + \ldots + \lambda_k Z_{k,j}.$$

From the above one sees that, for $j \in \others$ the $\Zcal_{k,j}^{comb}$ equals $Z_{k,j}^{comb}$, and for $j \in \sent$, on the set $\aee{k-1}^c\cap \ah{k}^c$, the statistic $\Zcal_{k,j}^{comb}$ exceeds
$$\left[ \lambda_1 \sqrt{w_1} +   \ldots \+ \lambda_k\sqrt{ w_k}\right]\sqrt{\cjrbh}  + \,Z_{k,j}^{comb},$$
%$$\sqrt{\lambda_1^*} \Zcal_{1,j}^{h} - \sqrt{\lambda_2^*} \Zcal_{2,j}^{h} - \ldots - \sqrt{\lambda_k^*} \Zcal_{k,j}^{h}$$
which is equal to
$$\sqrt{\frac{\cjrbh}{1-q_{k-1}^{adj, tot}\,\nu}} \,  \,+\, Z_{k,j}^{comb}.$$
%using $\hat w_k \geq w_k^*$.

Summarizing, $$\Zcal_{k,j}^{comb} = Z_{k,j}^{comb} \quad\quad\mbox{for $j \in \others$}$$
and, on the set $\aee{k-1}^c\cap \ah{k}^c$,
$$\Zcal_{k,j}^{comb} \geq \shift_{k,j} \,+\, Z_{k,j}^{comb},\quad\quad\mbox{for $j \in \sent$},$$
where  $$\shift_{k,j}=\sqrt{\frac{\cjrbh}{1-x_{k-1}\,\nu}},$$ with $x_0 = 0$ and $x_{k-1}=q_{k-1}^{adj,tot}$, for $k \geq 2$. Since the $x_k$'s are increasing, the $\shift_{k,j}$'s increases with $k$. It is this increase in the mean shifts that helps in additional detections.

For each $j \in J$, set $H_{k,j}$ to be the event,
 \eqq{
 H_{k,j} = \left\{\shift_{k,j} 1_{\{j\, \in\, \sent\}} \, + \, Z_{k,j} \geq \tau \right\} .\label{eq:hkj}
 }
Notice that \eqq{H_{k,j} = \{ \Zcal_{k,j}^{comb} \geq \tau\}\quad\quad\mbox{for $j \in \others$}.\label{eq:hkjotherlow}} On the set $\aee{k-1}^c\cap \ah{k}^c$, defined above, one has \eqq{H_{k,j} \subseteq \{ \Zcal_{k,j}^{comb} \geq \tau\}\quad\quad\mbox{for $j \in \sent$}.\label{eq:hkjsentlow}}
Using the above characterization of $\Zcal_{k,j}^{comb}$ we specify in the next subsection the values for $q_{1,k},\, f_k$ and $q_k$. Recall that the quantity $q_{1,k}$, which was defined is subsection \ref{subsec:modpacing}, gave  controls on $size_{1,k}$, the size of the decoded set $dec_{1,k}$ after the $k$ step.

\subsection{Specification of $f_k,\, q_{1,k}$, and $q_k$, for $k = 1,\ldots,m$} \label{subsec:targetfalse}

 Recall from subsection \ref{subsec:nearbydist} that under the $\QQ$ measure that $Z_{k,j}$, for $j \in J$, are i.i.d. standard normal random variables. Define the random variable \eqq{ \fkup{k}= \sum_{j \in \others} \pi_j 1_{H_{k,j}} \label{eq:fkhatup}.}
 Notice that $\hat f_k \leq \fkup{k}$ since %notice that
 \alge{\hat f_k &= \displaystyle\sum_{j \in dec_k \cap \others} \pi_j\nonumber\\
            &\leq \sum_{j \in \others} \pi_j 1_{\{\Zcal_{k,j}^{comb} \geq \tau\}}.\label{eq:tempsum}
       }
The above inequality follows since $dec_k$ is a subset of $thresh_k =  \{j:\Zcal_{k,j}^{comb} \geq \tau\}$ by construction. Further \eqref{eq:tempsum} is equal to $\fkup{k}$ using \eqref{eq:hkjotherlow}.

The expectation of $\fkup{k}$ under the $\QQ$-measure is given by,
$$\E_{\QQ}\left(\fkup{k}\right) = (\M\!-\!1)\bar \Phi(\tau),$$
where $\bar \Phi(\tau)$ is the upper tail probability of a  standard normal at $\tau$.  Here we use the fact that the $H_{k,j}$, for $j \in \others$, are i.i.d Bernoulli $\bar \Phi(\tau)$ under the $\QQ$-measure and that $\sum_{j \in \others} \pi_j$ is equal to $(M-1)$.

 Define $f^* = (\M\!-\!1)\bar \Phi(\tau),$ which is the expectation of $\fkup{k}$ from above. One sees that
\eqq{f^* \leq \frac{\exp\big\{-a \sqrt{2 \log \M} -(1/2)a^2\big\}}{(\sqrt{2 \log \M} + a)\sqrt{2\pi}} ,\label{eq:fstarbound}}
using the form for the threshold $\tau$ in \eqref{eq:taudef}. We also use that
 $\bar\Phi(x) \leq \phi(x)/x$ for positive $x$, with $\phi$ being the standard normal density. We take $f_k = f $ to be a value greater than $f^*$. We express it in the form $$f= \rho f^*,$$ with a constant factor $\rho > 1$. This completes the specification of the $f_k$.%Accordingly, since $\fkup{up}$ is a weighted sum of i.i.d. Bernoulli random variables, we will demonstrate through a large deviation inequality that the probability that $\fkup{k}$ exceeds $f$ is exponentially small in $L$.

Next, we specify the $q_{1,k}$ used in pacing the steps. Denote the random variable,
\eqq{\hat q_{1,k} = \sum_{j \: \sent} \pi_j 1_{H_{k,j}} \label{eq:qhatonek}.}
Likewise, define $q_{1,k}^*$ as the expectation of $\hat q_{1,k}$ under the $\QQ$ measure. Using \eqref{eq:hkj}, one has
$$q_{1,k}^* = \sum_{j \: \sent} \pi_j \bar \Phi(-\mu_{k,j}),$$
where $\mu_{k,j} = \shift_{k,j}-\tau$. Like before, we take $q_{1,k}$ to be a value less than $q_{1,k}^*$. More specifically, we take
\eqq{q_{1,k} = q_{1,k}^* - \eta\label{eq:qonekdef}}
for a positive $\eta$. %The probability that $\hat q_{1,k}$ is less than $q_{1,k}$ will be seen to be exponentially small in $L$.

This specification of $q_{1,k}^*$, and the related $q_{1,k}$, is a recursive definition. In particular, denoting
$$\mu_j(x) =\sqrt{1/(1-x\nu)}\,\, \sqrt{\cjrbh} \,-\, \tau.$$ Then $q_{1,k}^*$ equals the function
\eqq{g_L(x) = \sum_{j \: \sent} \pi_j \bar \Phi(-\mu_{j}(x))\label{eq:gdef}}
evaluated at $x_{k-1} = q_{k-1}^{adj,tot}$, with $x_0 = 0$.

For instance, in the constant power allocation case $\cjrbh =  (R_{0}(1- h)/R) \, 2\log \M$, is the same for all $j$.  This makes  $\shift_{k,j}$ the same for each $j$.  Consequently, $\mu_j(x) = \mu(x)$, where $\mu(x) = \sqrt{1/(1-x\nu)}\, \sqrt{\cjrbh} - \tau$. Then one has $q_{1,k}^*=\bar \Phi(-\mu(x_{k-1})).$
%which is the expectation, with respect to $\QQ$, of $1_{H_{k,j}}$ for each $j$ sent.  That is, for each $j$ sent, the probability that $Z_{k,j}^{comb}$ exceeds $-\mu_k$, which is the threshold adjusted down by the shift.
It obeys the recursion $q_{1,k}^* = g_L(x)$ evaluated at $x_{k-1} = q_{k-1}^{adj,tot}$, %$q_{1,k-1}^{adj}$
with $g_L(x)=\bar \Phi(-\mu(x))$.

%Where this will get us is demonstration that $q_{1,k}$ is a likely lower bound on $\hat q_{1,k} =\sum_{j \: \sent} \pi_j 1_{H_{k,j}}$, which, as we have said, is a likely lower bound on the total fraction of correct detections using \eqref{eq:qktotlowbdd}.  If, for a suitable number of steps, we have arranged sufficient growth in $q_{k-1}^{adj, tot}$, then $q_{1,k}$ will be near $1$ at the final $k$.

Further, we define the target detection rate for the $k$ th step, given by $q_k$, as
\eqq{q_k = q_{1,k} - q_{1,k-1} - 1/L_{\pi} - f\label{eq:qk},}
with $q_{1,0}$ taken to zero.  Thus the $q_k$ are specified from the $q_{1,k}$ and $f$. Also, $1/L_\pi = \min_{\ell = 1,\ldots, L} \pi_{(\ell)}$ is a quantity of order $1/L$. For the power allocation \eqref{eq:powerweuse}, one sees that $L_\pi = L\nu/(2\C)$.

\subsection{Building Up the Total Detection Rate} \label{subsec:builddetect}

The previous section demonstrated the importance of the function $g_L(x)$, given by \eqref{eq:gdef}. This function is defined on $[0,1]$ and take values in the interval $(0,1)$. Recall from subsection \ref{subsec:modpacing}, on pacing the steps, that the quantities $q_{1,k}$ are closely related to the proportion of sections correctly detected after $k$ steps, if we ignore false alarm effects. Consequently, to ensure sufficient correct detections one would like the $q_{1,k}$ to increase with $k$ to a value near 1. Through the recursive definition of $q_{1,k}$, this amounts to ensuring that the function $g_L(x)$ is greater than $x$ for an interval $[0,x_r]$, with $x_r$ preferably near 1.

%Let's demonstrate here how the likely total correct detection rate $q_{1,k}$ builds up to a value near $1$, followed by the corresponding conclusion of reliability.  Here we define the notion of correct detection being {\em accumulative}.  This notion holds for the power allocations we study. In this section we illustrate the matter in the case of constant power allocation and $R$ at most $R_0(1 - h)$.  Thereafter, we handle variable power allocation and rates $R$ up to the capacity.

%Recall that for our iterative algorithm, from the function $g(.)$, given by \eqref{eq:gdef},
%For our adaptive successive decoder there is a function $g(x)$ for $0\le x\le 1$ such that for each step %$k$,
%$$q_{1,k}^*= g(q_{k-1}^{adj,tot}),$$
%with which we then update the new $q_{1,k}$ by choosing it to be slightly less than $q_{1,k}^*$.  That can be done by setting a small constant $\eta$ for which $q_{1,k} = q_{1,k}^*-\eta$.  Slightly better alternative choices motivated by the reliability bounds are to arrange $\sqrt{1\!-\!q_{1,k}} = \sqrt{1-q_{1,k}^*} + \eta$ or $D(q_{1,k}\|q_{1,k}^*)= \eta^2$ where $D$ is the relative entropy between Bernoulli random variables of the indicated success probabilities.

%Let $x_r$ be any given value between $0$ and $1$, preferably near $1$.

\vspace{0.3cm}
\noindent
{\bf {Definition:}} A function $g(x)$ is said to be {\em accumulative} for $0\le x \le x_r$ with a positive $gap$, if $$g(x)-x \ge gap$$
for all $0\le x \le x_r$.  Moreover, an adaptive successive decoder is {\em accumulative} with a given rate and power allocation if corresponding function $g_L(x)$  satisfies this property for given $x_r$ and positive $gap$.
\vspace{0.3cm}

To detail the progression of the $q_{1,k}$ consider the following lemma.

\vspace{0.3cm}
\begin{lem}\label{progress} Assume $g(x)$ is accumulative on $[0,x_r]$ with a positive $gap$, and $\eta$ is chosen so that $gap - \eta$ is positive. Further, assume
\eqq{f \leq (gap - \eta)^2/8 - 1/(2L_\pi)\label{eq:fcriterion}.}
Then, one can arrange for an $m$ so that the $q_{1,k}$, for $k = 1,\ldots,m$, defined by \eqref{eq:qonekdef}, are increasing and
$$q_{1,m} \geq x_r + gap - \eta.$$
Moreover, the number of steps $m$ is at most $2/(gap - \eta)$.
%Take
%$$f \leq (gap - \eta)^2/4$$
%and the number of step $m$, be such that
%$$m \leq 1/(gap - \eta) \quad\text{and}\quad q_{m-1}^{adj,tot} \leq x_r.$$
%Then the increase $q_{1,k} -q_{1,k-1}$, for $k = 1,\ldots,m$, is at least
%$$\tL  = (gap - \eta)/2 - m/L_\pi.$$
\end{lem}

The proof of Lemma \ref{progress} is given in Appendix \ref{proofprogress}.
We now proceed to describe how we demonstrate the reliability of the algorithm using the quantities chosen above.

%\subsection{Simple device in bounding detections and false alarms} \label{subsec:simpledevice}

\section{Reliability of the Decoder} \label{sec:reliabmulti}

We are interested in demonstrating that the probability of the event $\aq{1,m}\cup\af{1,m}$ is small. This ensures that for each step $k$, where $k$ ranges from 1 to $m$, the proportion of correct detections $\hat q_k$ is at least $q_k$, and the proportion of false alarms $\hat f_k$ is at most $f_k = f$. We do this by demonstrating that the probability of the set
$$E_m = \aq{1,m}\cup\af{1,m}\cup \ah{1,m}$$
is exponentially small. The following lemma will be useful in this regard.
Recall that
$$\aq{1,m} = \cup_{k =1}^m \{\hat q_k < q_k\}\quad\text{and}\quad\af{1,m} = \cup_{k =1}^m \{\hat f_k > f\}.$$
\begin{lem} Let $q_{1,k}, \, q_k$ and $f$ be as defined in subsection \ref{subsec:targetfalse}.
  \label{lem:paccritlem}
%  For each $k = 1,\ldots,m$,
Denote $$\tilde{A}_{1,m} = \cup_{k = 1}^m \{\hat q_{1,k} < q_{1,k}\}\quad \text{and}\quad \tilde{B}_{1,m} = \cup_{k = 1}^m \{\fkup{k} > f\}.$$
Then,
  $$\aee{m} \subseteq  \tilde{A}_{1,m} \cup \tilde{B}_{1,m} \cup \ah{1,m}.$$
%  where recall that $\aeet{k-1} =  \aee{k-1} \cup \af{k}$, and the set on the right side of the above is equal to $\aq{k}^c$.
\end{lem}

For ease of exposition we provide the proof of this lemma in Appendix \ref{appfpaccritlem}. The lemma above described how we control the probability of the exception set $\aee{m}$. %We now describe
%how we control the probability of $\af{1,m}$.

%
We demonstrate that the probability of $\aee{m}$ is exponentially small by showing that the probability of
$\tilde{A}_{1,m} \cup \tilde{B}_{1,m} \cup \ah{1,m},$
which contains $\aee{m}$, is exponentially small.
 Also, notice that outside the set $\aee{m}$, the weighted fraction of failed detection and false alarms, denoted by $\hatdwght$ in \eqref{eq:hatdeltawght}, is bounded by
$$(1 - \sum_{k = 1}^m  q_k) + m f,$$
which, after recalling the definition of $q_k$ in \eqref{eq:qk}, can also be expressed as,
\eqq{1 - q_{1,m}  + 2m f + m/L_\pi.\label{eq:hatdeltawghtbound}}
Now, assume that $g_L$ is accumulative on $[0, x_r]$ with a positive $gap$. Then, from Lemma \ref{progress}, for $\eta < gap$, and $f> f^*$ satisfying \eqref{eq:fcriterion}, one has that \eqref{eq:hatdeltawghtbound} is upper bounded by
\eqq{\dwght = (1 - x_r) - (gap - \eta)/2,\label{eq:deltawghtweuse}}
using the bounds on $f$, $q_{1,m}$ and $m$ given in the lemma. Consequently, $\hat{\delta}_{mis}$ the mistake rate after $m$ steps, given by \eqref{eq:hatdeltamis}, is bounded by $\delta_{mis}$ outside of $\tilde{A}_{1,m} \cup \tilde{B}_{1,m} \cup \ah{1,m}$, where,
\eqq{\delta_{mis} = \frac{snr}{2\C}[(1 - x_r) - (gap - \eta)/2]\label{eq:deltamisweuse},}
via \eqref{eq:deltamisdeltawghtrel}.
%\eqq{\dwght = 1 - q_{1,m}  + 2m f + m/L_\pi\label{eq:deltawghtweuse}.}
%The $m/L_\pi$ is a negligible term, since we arrange the number of steps $m$ to be of order $\log\M$.
We then have the following theorem regarding the reliability of the algorithm.

%Here we establish, for any power allocation and rate for which the decoder is accumulative, the reliability with which the weighted fractions of mistakes are governed by the studied quantities $1-q_{1,m}$ plus $f_{1,m}$.  The bounds on the probabilities with which the fractions of mistakes are worse than such targets are exponentially small in $L$.  The implication is that if the power assignment and the communication rate are such that the function $g_L$ is accumulative on $[0,x^*]$, then for a suitable number of steps, the tail probability for weighted fraction of mistakes more than $\delta^*=1-g_L(x^*)$ is exponentially small in $L$.

%\vspace{0.3cm}
\begin{thm} \label{reliabddthm} Let the conditions of Lemma \ref{progress} hold, and let $\delta_{mis}$ be as in \eqref{eq:deltamisweuse}. Then, %the mistakes rate after $m$ steps, that is $\hat{\delta}_{mis}$, is at most $\delta_{mis}$ except in an event of probability less than
\algg{\PP(\hat{\delta}_{mis} > \delta_{mis})\,\, \leq & \,\,m e^{-2L_\pi \eta^2+m c_0} + m e^{-L_\pi f \Dcal(\rho)/\rho}\\
 &+ m e^{-(n\!-\!m\+1)h^2/2}.}
Here the quantities $\eta$ and $\rho$ are as defined in subsection \ref{subsec:targetfalse}, and $c_0$ is as given in Lemma \ref{measurebdd2}. Also $\Dcal(\rho)=\rho \log \rho - (\rho\!-\!1)$.
\end{thm}

%\vspace{0.3cm}
\begin{proof}[Proof of Theorem \ref{reliabddthm}]
From Lemma \ref{lem:paccritlem}, and the arguments above, the event $\{\hat{\delta}_{mis} > \delta_{mis}\}$ is contained in the event
$$\tilde{A}_{1,m} \cup \tilde{B}_{1,m} \cup D_{1,m}.$$
Consequently, we need to control the probability of the above three events under the $\PP$ measure.

We first control the probability of the event $D_{1,m}$, which is the union of Chi-square events $D_k = \{\Chi^2_{d_k}/n < 1 - h\}$. Now the event $D_k$ can be expressed as $\{\Chi^2_{d_k}/d_k < 1 - h_k\}$, where
$h_k = (nh - k + 1)/(n - k + 1)$. Using a standard Chernoff bound argument, one gets that
$$\PP(D_k) \leq e^{-(n - k + 1)h_k^2/2}.$$
The exponent in the above is at least $(n - k + 1)h^2/2 - kh$. Consequently, as $k \leq m$, one gets, using a union bound that
$$\PP(D_{1,m}) \leq me^{-(n - m + 1)h^2/2 + mh}.$$
Next, lets focus on the event $\tilde{B}_{1,m}$, which is the union of events $\{\fkup{k} > f\}$. Divide
$\fkup{k},\, f$, by $\M-1$ to get $\hat p_k,\, p$ respectively. Consequently, $\tilde{B}_{1,m}$ is also the union of the events $ \{\hat p_k > p \}$, for $k = 1,\ldots,m$, where
$$\hat p_{k} = \frac{1}{\M\!-\!1} \sum_{j \in other} \pi_j \,1_{H_{k,j}},$$
and $p = f/(M-1)$, with $f = \rho f^*$.

Recall, as previously discussed, for $j$ in $\others$, the event $H_{k,j}$ are i.i.d. Bernoulli($p^*$)  under the measure $\PP$, where $p^* = f^*/(M-1)$. Consequently, from
by Lemma \ref{bernoullisums} in the Appendix \ref{aptailber}, the probability of the event $\{\hat p_{k} \!\ge\! p\}$  is less than $e^{-L_\pi(\M\!-\!1)D(p\|p^*)}$.  Therefore,
$$\PP(\tilde{B}_{1,m}) \leq me^{-L_\pi(\M\!-\!1)D(p\|p^*)}.$$
To handle the exponents $(\M\!-\!1)D(p\|p^*)$ at the small values $p$ and $p^*$, we use the Poisson lower bound on the Bernoulli relative entropy, as shown in Appendix \ref{lowerboundD}. This produces the lower bound $(\M\!-\!1)[p\log p/p^* + p^*-p]$, which is equal to
$$f \log f/f^* + f^* - f.$$
We may write this as $f^* \Dcal(\rho)$, or equivalently $f \Dcal(\rho)/\rho$, where the functions $\Dcal(\rho)$ and $\Dcal(\rho)/\rho = \log \rho + 1-1/\rho$ are increasing in $\rho$.

Lastly, we control the probability of the event $\tilde{A}_{1,m}$, which the is union of the events
$\{\hat q_{1,k} < q_{1,k}\}$, where
$$\hat q_{1,k} = \sum_{j \in \sent} \pi_j H_{k,j}.$$
We first bound the probability under the $\QQ$ measure. Recall that under $\QQ$, the $H_{k,j}$, for $j \in \sent$, are independent Bernoulli, with the expectation of $\hat q_{1,k}$ being $q_{1,k}^*$. Consequently, using Lemma \ref{bernoullisums} in Appendix \ref{aptailber}, we have
$$\QQ(\hat q_{1,k} < q_{1,k}) \leq e^{- L_\pi D(q_{1,k}\| q_{1,k}^*)}.$$
Further,  by the Pinsker-Csiszar-Kulback-Kemperman inequality, specialized to Bernoulli distributions, the expressions $D(q_{1,k}\|q_{1,k}^*)$ in the above exceeds $2(q_{1,k}-q_{1,k}^*)^2$, which is $2\eta^2$, since $q_{1,k}^* - q_{1,k} = \eta$.

Correspondingly, one has
$$\QQ(\tilde{A}_{1,m}) \leq m e^{-L_\pi 2\eta^2}.$$
Now, use the fact that the event $\tilde{A}_{1,m}$ is $\Fcal_m$ measurable, along with Lemma \ref{measurebdd2}, to get that,
$$\PP(\tilde{A}_{1,m}) \leq m e^{-L_\pi 2\eta^2 + m\,c_0}.$$
This completes the proof of the lemma.
\end{proof}

\section{Computational Illustrations} \label{sec:compillus}
We illustrate in two ways the performance of our algorithm. First, for fixed values $L,\,B,\,snr$ and rates
below capacity we evaluate detection rate as well as probability of exception set $p_{e}$ using the theoretical bounds given in Theorem \ref{reliabddthm}. Plots demonstrating the progression of our algorithm are also shown. These
highlight the crucial role of the function $g_L$ in achieving high reliability.

\begin{figure*}
\centerline{
\mbox{
\includegraphics[width=2.7in]{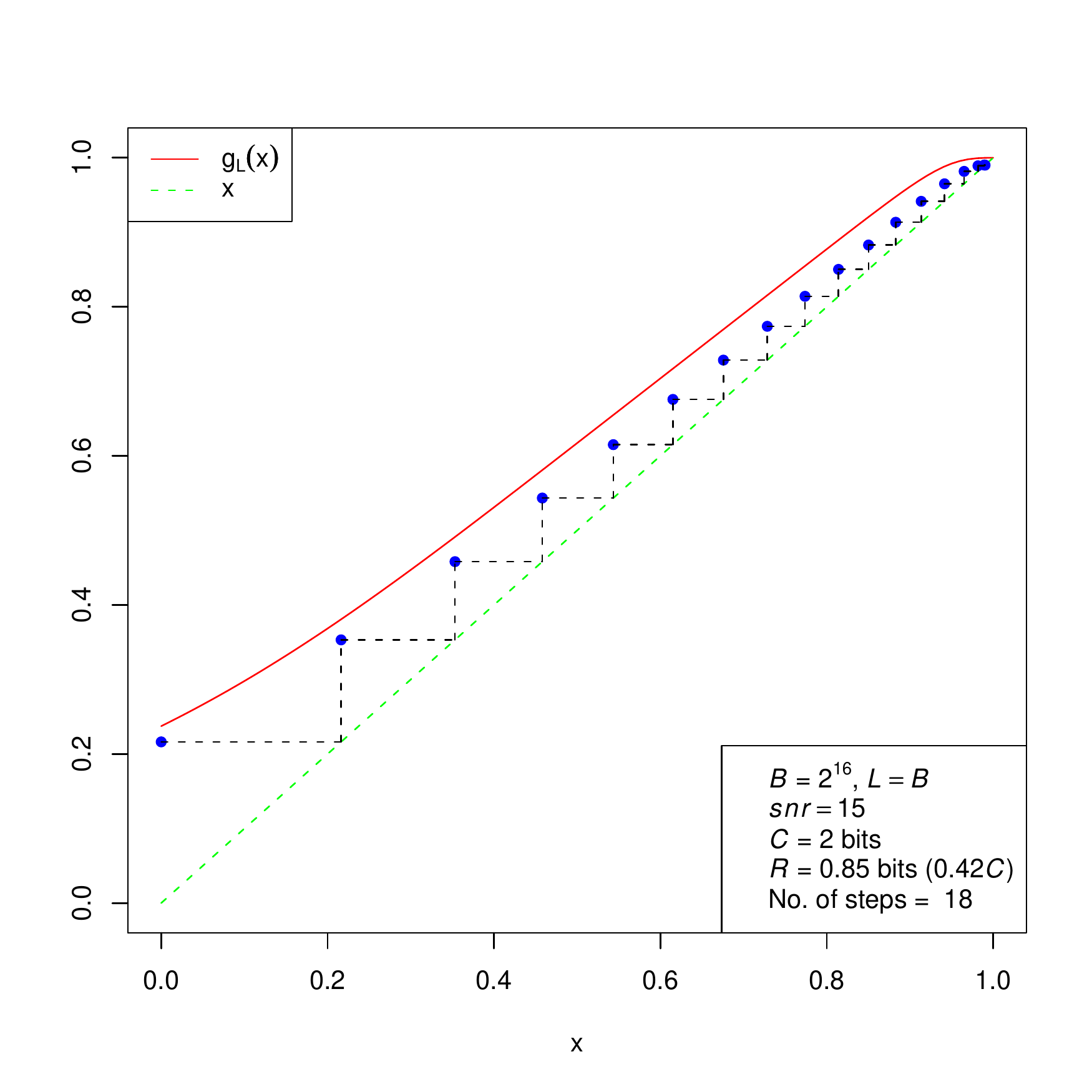}
}
\mbox{
\includegraphics[width=2.7in]{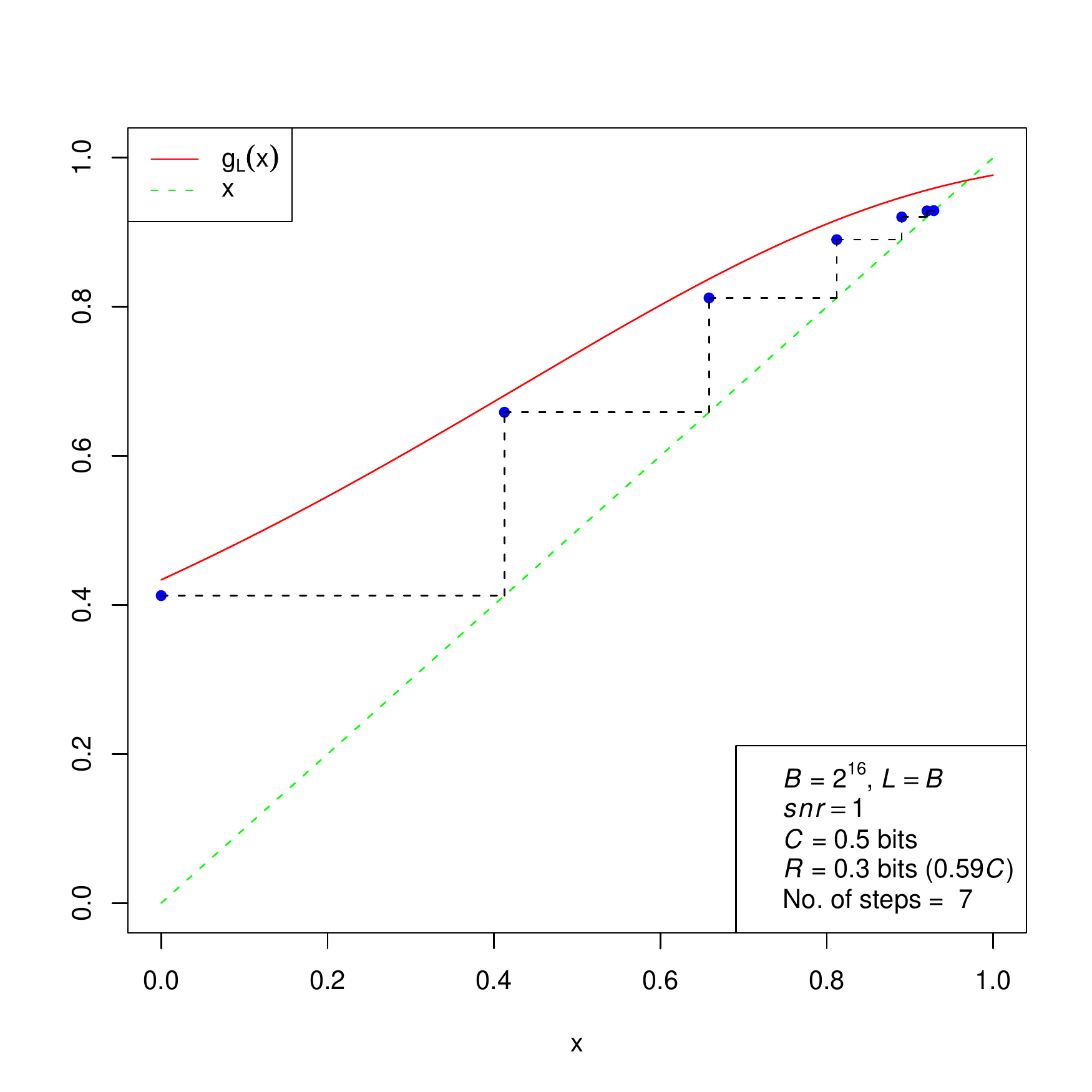}
}
}
\caption[Demonstration of progression of the iterative algorithm]{ Plots demonstrating progression of our algorithm. (Plot on left) $snr = 15$. %For the above we take $a =1.3375$, $c = 0.0041$ and $\gamma = 0.83C$.
The weighted (unweighted) detection rate is $0.995$ ($0.985$) for a failed detection rate of $0.014$ and the
false alarm rate is $0.005$. (Plot on right) $snr = 1$.
%Constant power allocation for section was used with $a =0.6625$.
The detection rate (both weighted and un-weighted) is $0.944$ and the
false alarm and failed detection rates are $0.016$ and  $0.055$ respectively.
}
\label{fig:progress2}
\end{figure*}

Figure \ref{fig:progress2}  presents the results of computation using the reliability bounds of Theorem \ref{reliabddthm} for fixed $L$ and $B$ and various choices of $snr$ and rates below capacity. The dots in these figures denotes $q_{1,k}$, for each $k$.% and the step function joining these dots highlight how $q_{1,k}$ is computed from $q_{1,k-1}^{adj}$. For large $L$ and $B$ these $q_{1,k}$'s are near $q_{1,k}$, our lower bound on the proportion of sections decoded after $k$ passes.
%In this extreme case $q_{1,k}$ would match $g_L(q_{1,k-1})$, so that the dots would lie on the function.
%We also compute the probability $P_{\mathcal{E}}$ for our choices of $q_{1,k}$'s and $f_k$'s using the formula given in the above theorem.

For illustrative purposes we take $B = 2^{16}$, $L = B$ and $snr$ values of $1,\, 7$ and $15$. The probability of error $p_e$ is set to be near $10^{-3}$. For each $snr$ value the maximum rate, over a grid of values, for which the error probability is less than $p_e$
is determined. With $snr = 1$ (Fig \ref{fig:progress2}), this rate $R$ is $0.3$ bits which is $59\%$ of capacity. When $snr$ is $7$ (Fig \ref{fig:prog7}) and $15$ (Fig \ref{fig:progress2}) , these rates correspond to $49\%$ and $42\%$ of their corresponding capacities.

For the above computations we chose power allocations of the form $$P_{(\ell)} \propto \max\{e^{-2\gamma l/L}, u\},$$ with $0 \leq \gamma \leq \mathcal{C}$, and $u > 0$.   Here the choices of $a,\, u$ and $\gamma$ are made, by computational search, to minimize the resulting sum of false alarms and failed detections, as per our bounds. In the  $snr = 1$ case the optimum $\gamma$ is 0, so we have constant power allocation in this case. In the other two cases, there is variable power across most of the sections. The role of a positive $u$ being to increase the relative power allocation for sections with low weights. %Note, in our analytical results for maximum achievable rates as a function of $B$, as given in Proposition \ref{mainthm} and in subsection \ref{choicepara} later on,  $\gamma$ is constrained to be equal to $\mathcal{C}$.

% Regarding the role of $\gamma$,  a value of $\gamma = C$ would make no difference in the $snr = 1$ case since this would be
%  countered using a larger $c$ to produce an equally weighted scheme for the sections. Even for the $snr = 7$ case, setting
% $\gamma = C$ is seen to produce only a small drop in total detection rate when optimized only over $a$ and $c$.
% However for $snr = 10,\, R = 0.95$ bps case, a $\gamma$ value of $C$ produces an failed detection rate of $23\%$, far greater
% than the $1.3\%$ in Fig. 3 where $\gamma$ of  $0.76C$ was used.\\
%Note, for the above computations we have incorporated the enhancement from the wedges and made use of an improved upper bound on false alarm rates. Computation details of these can be found in appendices.
\begin{figure*}
\centerline{
\mbox{
\includegraphics[width=2.25in]{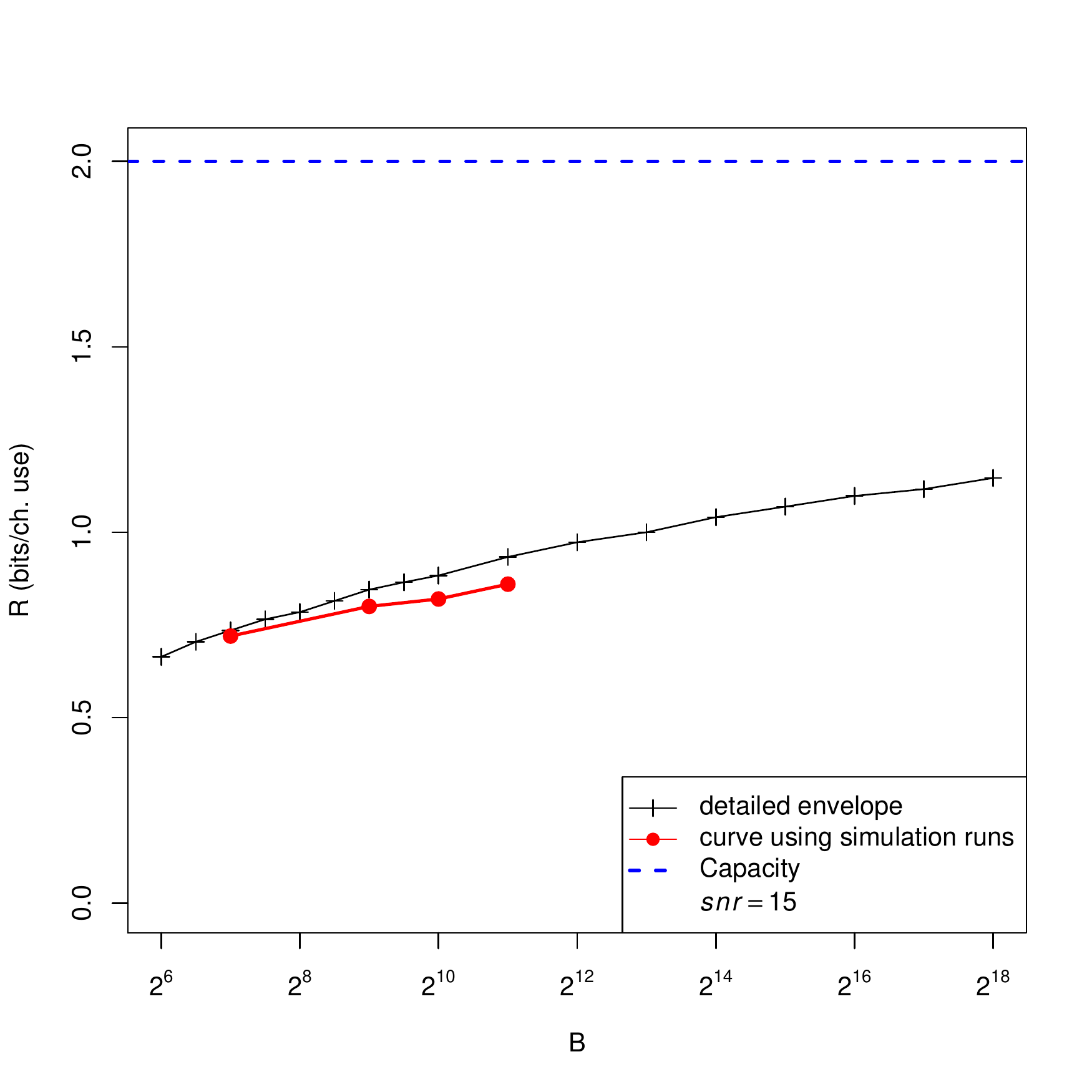}
}
\mbox{
\includegraphics[width=2.25in]{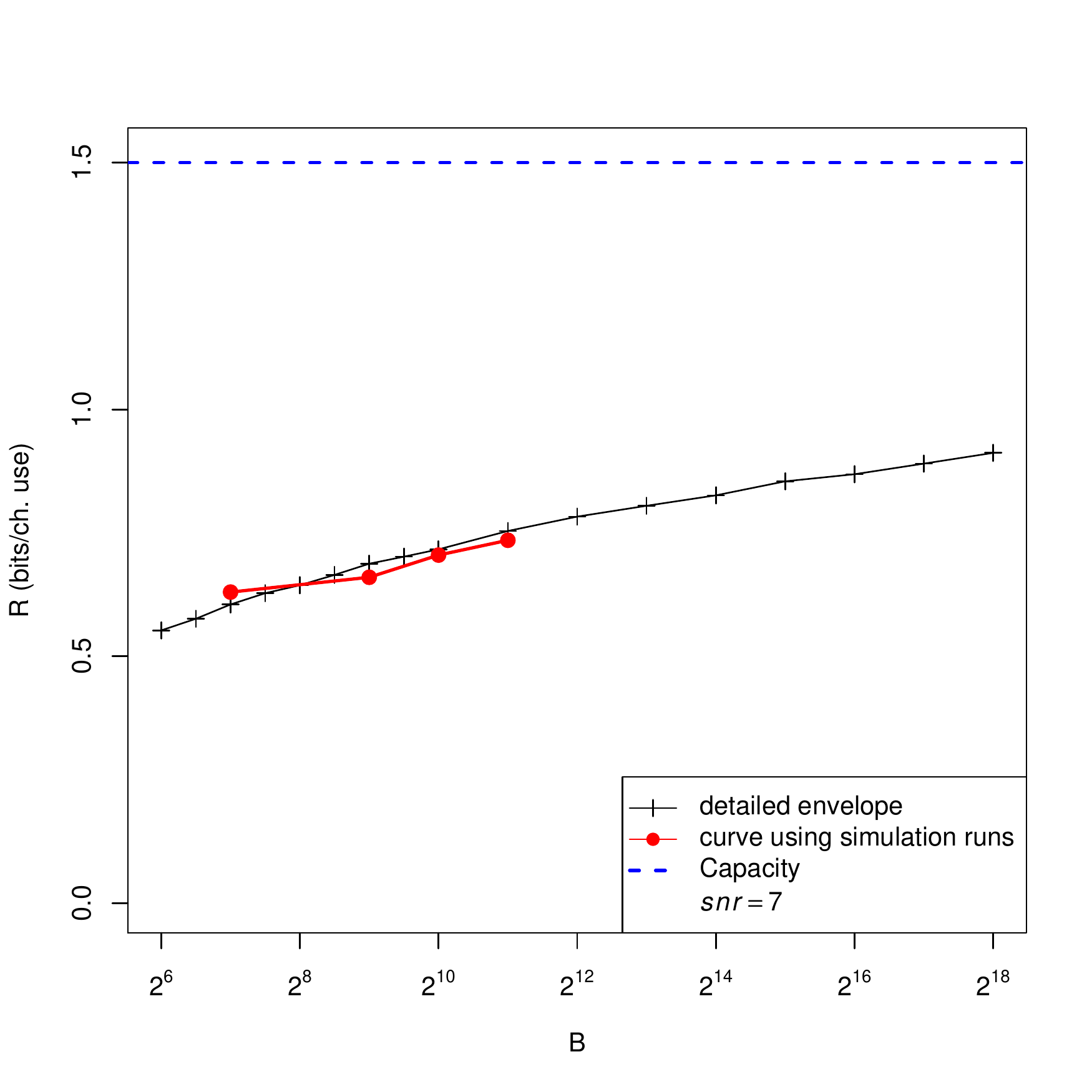}
}
\mbox{
\includegraphics[width=2.25in]{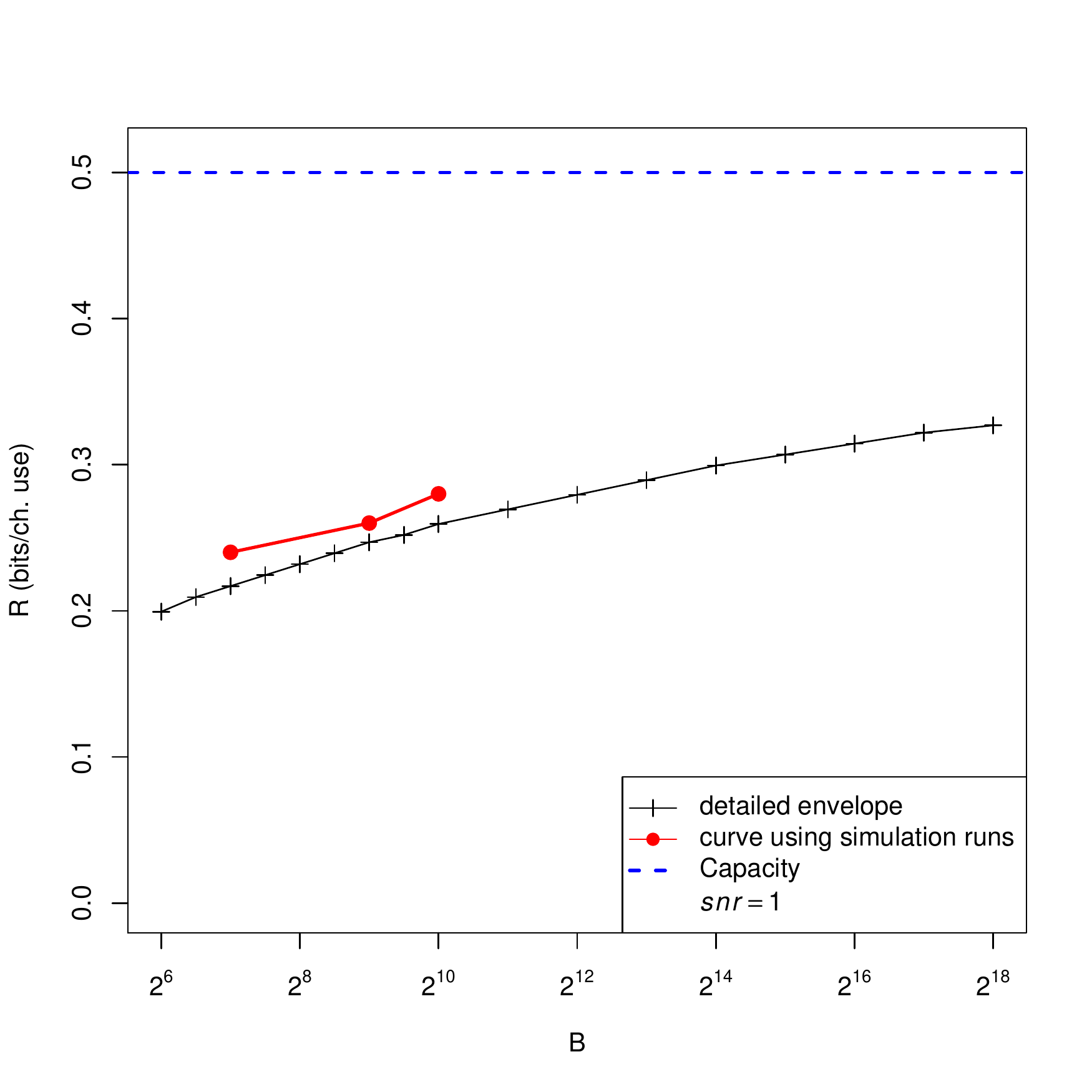}
}
}
\caption[Plots of achievable rates using the iterative algorithm]{ Plots of achievable rates as a function of $B$ for $snr$ values of $15,\, 7$ and $1$. Section error rate is controlled to be between 9 and $10\%$. For  the curve using simulation runs the error probability of making more than 10\% section mistakes is taken to be $10^{-3}$.
}
\label{fig:aymp1715}
\end{figure*}

%\vspace{0.2cm}
%\noindent
%{\bf{Achievable rates as a function of $B$:}}
Figure \ref{fig:aymp1715} gives plots of achievable rates  as a function of $B$. For each $B$, the points on the detailed envelope correspond to the numerically evaluated maximum inner code rate for which the section error is between 9 and 10\%.  Here we assume $L$  to be large, so that the $q_{1,k}$ and $f_k$ are replaced by the expected values $q_{1,k}^*$ and $f^*$, respectively. We also take $h = 0$. This gives an idea about the best possible rates for a given $snr$ and section error rate.

For the simulation curve, $L$ was fixed at 100 and for given $snr$, $B$, and rate values, $10^4$ runs of our algorithm were performed. The maximum rate over the grid of values satisfying section error rate of less than 10\% except in 10 replicates,
(corresponding to an estimated $p_e$ of $10^{-3}$) is shown in the plots. Interestingly, even for such small values of $L$, the curve is
is quite close to the detailed envelope curve, showing that our theoretical bounds are quite conservative. %Indeed, if we had
%used our theoretical bounds for the $L = 100$ case it would have been lower than the `$L = B$' curve.

\section{Achievable Rates approaching Capacity}  \label{commcapac}

We demonstrate analytically that rates $R$ moderately close to $\Capacity$ are attainable by showing that the function $g_L(x)$ providing the updates for the %half-space based measure of the fraction correctly decoder
fraction of correctly detected terms is indeed accumulative for suitable $x_r$ and $gap$.  Then the reliability of the decoder can be established via Theorem \ref{reliabddthm}. In particular, the matter of normalization of the weights $\pi_{(\ell)}$ is developed in subsection \ref{varpower}. An integral approximation $g(x)$ to the sum $g_L(x)$ is provided in subsection \ref{integapprox}, and in subsection \ref{ggx} we show that it is accumulative. Subsection
\ref{choicepara} addresses the issue of control of parameters that arise in specifying the code. In subsection \ref{proofmainthm}, we give the proof of Proposition \ref{mainthm}.

\subsection{Variable power allocations} \label{varpower}

As mentioned earlier, we consider power allocations $P_{(\ell)}$ proportional to $e^{-2\Capacity \ell/L}$. The function $g_L(x)$, given by \eqref{eq:gdef}, may also be expressed as
$$g_L(x)=\sum_{\ell = 1}^L \pi_{(\ell)} \, \Phi(\mu(x, u_{\ell}\,\C'/R)),$$
where $\pi_{(\ell)} = P_{(\ell)}/P$, and,
$$\mu(x,u) = (\sqrt{u/(1 - x\nu)} - 1)\sqrt{2\log\M} - a$$
and
$$u_\ell= e^{-2\Capacity(\ell-1)/L}\quad\text{and}\quad \C' = \tilde{\C}(1 - h),$$
with $\tilde \C$ as in \eqref{eq:ctilda}. Here we use the fact that $\cjrb$, for the above power allocation, is given by $u_\ell\, \tilde \C/R$ if $j$ is in section $\ell$, as demonstrated in \eqref{eq:cjrbexpo}.

Further, notice  that  $\pi_{(\ell)}= u_\ell/ sum$, with $sum= \sum_{\ell=1}^L u_\ell.$
One sees that $sum = L\nu/(2\tilde \Capacity)$, with $\nu = P/(P + \sigma^2)$.
Using this one gets that
\eqq{g_L(x) = \frac{2 \Capacity}{\nu L} \,\sum_{\ell=1}^L \, u_\ell \, \Phi\big(\mu(x,u_\ell\Capacity'/R)\big). \label{eq:gdefalternate}}

\subsection{Formulation and evaluation of the integral $g(x)$} \label{integapprox}

Recognize that the sum in \eqref{eq:gdefalternate} corresponds closely to an integral.  In each interval $\frac{\ell-1}{L}\le t < \frac{\ell}{L}$ for $\ell$ from $1$ to $L$, we have $e^{-2\Capacity \frac{\ell-1}{L}}$ at least $e^{-2 \Capacity t}$.  Consequently,
%recalling also that the ratio of the normalizing sum to a certain integral denoted $integ$ is equal to the ratio of $1\+\delta_{sum}^2$ and $1\+D(\delta_c)/snr$, it follows that
$g_L(x)$ is greater than $g(x)$
%$(integ/sum)g(x)$ where $g(x)=g_{num}(x)/(1+D(\delta_c))$ and
where %the numerator is %$g_{num}(x)$ is the integral
\eqq{g(x) = \frac{2 \Capacity}{\nu} \int_0^1 \!\!{e^{-2\Capacity t}} %{1\+D(\delta_c)/snr}
\Phi\big(\mu(x,e^{-2\Capacity t}\Capacity'/R)\big)  dt.\label{eq:gnumx}}
The $g_L(x)$ and $g(x)$ are increasing functions of $x$ on $[0,1]$.

Let's provide further characterization and evaluation of the integral $g(x)$.
Let
$$z_x^{low} = \mu(x,(1 - \nu)\,\Capacity'/R)\quad\text{and}\quad z_x^{max}= \mu(x,\Capacity'/R).$$
Further, let $\delta_a= a/\sqrt{2\log \M}.$
%From the expression for $\mu(z,u)$
For emphasis we write out that $z_x=z_x^{low}$ takes the form
\eqq{z_x = \left[ \frac{\sqrt{(1 - \nu)\Capacity'/R}}{\sqrt{1\,-\,x\nu}} - (1+\delta_a) \right]\sqrt{2\log \M}.\label{eq:zxlow}}

%Recall that the function $g_{num}(x)$ is
%$$\frac{2 \Capacity}{\nu} \!\int_0^1 \!\!\max\{e^{-2\Capacity t},u_{cut}\}  \, \Phi\big(\mu(x,\max\{e^{-2\Capacity t},u_{cut}\}\Capacity'/R)\big)  dt,$$

Change the variable of integration in \eqref{eq:gnumx} from $t$ to $u=e^{-2\Capacity t}$. Observing that $e^{-2\C} = 1 - \nu$, one sees that %to %produce the  simplified expression for the integral
%see that
$$g(x) = \frac{1}{\nu} \int_{1 - \nu}^1 \!\Phi\big(\mu(x,u\,\Capacity'/R)\big) du.$$
Now since $$\Phi(\mu)= \int 1_{\{z\le \mu\}}\,\phi(z)\,dz,$$
%and $\nu = 1-e^{-2\Capacity}$ is the length of the interval $[e^{-2\Capacity},1]$ of values of $u$,
it follows that
\eqq{g(x) = \int\int 1_{\big\{u_{cut} \le u\le 1\big\}}1_{\big\{z\le \mu(x,u\,\Capacity'/R)\big\}}\,\phi(z)\,dz\,du/\nu\label{eq:gnumoneinteg}.}
%which may be interpreted as the probability of the indicated subset of values of $z,u$ with respect to the product of the standard normal measure for $z$ and the uniform probability measure for $u$ over the interval $[e^{-2\Capacity},1]$ of length $\nu\!=\!1\!-\!e^{-2\Capacity}$. %The aim is to show that $g(x)$ nearly exceeds $x$.

%The $z_x^{low}$ and $z_x^{max}$, respectively, are the $z$ values where the right boundary of the graph of this subset of values for $(z,u)$ hits the values $u_{cut}$ and $1$ specified for $u$.

In \eqref{eq:gnumoneinteg}, the inequality
$$z \le \mu(x,u\,\Capacity'/R)$$
is the same as
$$\sqrt{u} \; \ge \; \sqrt{u_x R/\Capacity'} \, \big(1+ (z+a)/\sqrt{2\log \M} \, \big),$$
provided $z_x^{low} \leq z \leq z_x^{max}$. Here $u_x = 1-x\nu$. % is at least $1-\nu$.
Thereby, for all $z$, the length of this interval of values of $u$ can be written as
$$\left[1- \max\left\{u_x \frac{R}{\Capacity'}  \Bigl(1+\!\frac{z\+a}{\sqrt{2\log \M}}\Bigr)_+^{\,2} , 1 - \nu\right\} \right]_+.$$

Thus $g(x)$ is equal to,%the integral $g_{num}(x)$ may also be expressed as
\eqq{ \frac{1}{\nu} \int \left[1- \max\left\{u_x \frac{R}{\Capacity'}  \Bigl(1+\!\frac{z\+a}{\sqrt{2\log \M}}\Bigr)_+^{\,2} , 1 - \nu\right\} \right]_+\phi(z)dz\label{eq:gnumintegch}.}
%where $z_x=z_x^{low}$.

%\vspace{0.3cm}\noindent
\begin{lem}\label{dereval}  {\em Derivative evaluation.}  The derivative $g'(x)$ may be expressed as
\eqq{ \frac{R}{\Capacity'} \, \int_{z_x^{low}}^{z_x^{max}} \! \big(1\+\delta_a+\delta_z\big)^{2}\phi(z)dz.\label{eq:gderiv}}
Further, if
\eqq{R = \frac{\Capacity'}{[(1+\delta_a)^2 (1\+ \rr/\log \M)]},\label{eq:rateform}}
with $\rr \geq \rhalf$, where
\eqq{\rhalf = \frac{1}{2(1 + \delta_a)^2}\label{eq:rhalf},}
then the difference $g(x)-x$ is a decreasing function of $x$.
\end{lem}

%\vspace{0.3cm}\noindent
\begin{proof} The integrand in  \eqref{eq:gnumintegch} is continuous and piecewise differentiable in $x$, and its derivative is the integrand in \eqref{eq:gderiv}.
%the integral of its derivative is expressed above. %Direct evaluation confirms that
%It is in agreement with the derivative with respect to $x$ of the integral $g_{num}(x)$.

Further, \eqref{eq:gderiv} is less than,
 $$\frac{R}{\Capacity'} \int_{-\infty}^{\infty} \big(1\+\delta_a+\delta_z\big)^{2} \phi(z)dz = \frac{R}{\Capacity'}\left[(1\+\delta_a)^2 + 1/(2\log \M)\right],$$ which is less than $1$ for $\rr \geq \rhalf$.  Consequently, $g(x)-x$ is decreasing as it has a negative derivative.
\end{proof}

%\vspace{0.3cm}\noindent
\begin{cor}\label{lowerbound} {\em A lower bound.} The function $g(x)$ is at least %$g_{low}(x)$ given by
$$g_{low}(x) =
\quad\quad\quad\quad\quad\quad\quad\quad
\quad\quad\quad\quad\quad\quad\quad$$
$$\frac{1}{\nu} \!\int_{z_x^{low}}^{\infty}\!\!\left[1\!- (R/\Capacity')u_x \, \big(1+\!(z\+a)/\sqrt{2\log \M} \, \big)^2\right] \phi(z)dz.$$
This $g_{low}(x)$ is equal to
$$\Phi(z_x)\,+ \, \left[x  + \delta_R \frac{u_x}{\nu}\right] \big[1-\Phi(z_x)\big]$$
\eqq{- \;2(1\+\delta_a)\frac{R}{\Capacity'} \frac{u_x}{\nu} \, \frac{\phi(z_x)}{\sqrt{2\, \log \M}}\,- \, \frac{R}{\Capacity'} \frac{u_x}{\nu} \, \frac{z_x \phi(z_x)}{2\log \M}\label{eq:glow}.}
where
$$\delta_R = \frac{\rr -\rhalf}{\log \M +  r}.$$
Moreover, this $g_{low}(x)$  has derivative $g_{low}'(x)$ given by
$$ \frac{R}{\Capacity'} \, \int_{z_x}^{\infty} \! \big(1\+\delta_a+\delta_z\big)^{2}\phi(z)dz.$$
%where the first part vanishes when $c=0$.
\end{cor}

%namely there is no such positive part restriction to the whole integrand in the representation of $g_{low}(x)$, and accordingly $g_{low}(x)$ has the same derivative characterization but with no upper limit to the integration.

\begin{proof} The integral expressions for $g_{low}(x)$ are the same as for $g(x)$ except that the upper end point of the integration extends beyond $z_x^{max}$, where the integrand is negative. The lower bound conclusion follows from this negativity of the integrand above $z_x^{max}$.  The evaluation of $g_{low}(x)$ is fairly straightforward after using $z\phi(z) = -\phi'(z)$ and
$z^2\phi(z) = \phi(z) -(z\phi(z))'$.  Also use that $\Phi(z)$ tends to 1, while $\phi(z)$ and $z\phi(z)$ tend to $0$ as $z \rightarrow \infty$. This completes the proof of Corollary \ref{lowerbound}.
\end{proof}

\vspace{0.1cm}
\noindent
{\bf {Remark:}} What we gain with this lower bound is simplification because the result depends on $x$ only through $z_x = z_x^{low}$.

%\vspace{0.1cm} With the rate $R$ taken to be not more than $\Capacity'/(1\+\delta_a)^2$, we write it as
%\eqq{R=\frac{\Capacity'}{(1\+\delta_b)^2} = \frac{\Capacity'}{(1\+\delta_a)^2(1\+ \rr/\log \M)}\label{eq:rateform}}
%where $\delta_b=b/\sqrt{2\log \M}$.

\subsection{Showing $g(x)$ is greater than $x$} \label{ggx}

%This section shows that $g_L(x)$ is accumulative, with a suitable $x_r$ and $gap$, using the lower bound which is here explored further.

The preceding subsection established that $g_L(x) - x$ is at least $g_{low}(x) - x$.
We now show that  $$h_{low}(x) = g_{low}(x)-x,$$ is at least a positive value, which we denote as $gap$, on an interval  $[0, x_r]$, with $x_r$ suitably chosen.

%As above we use a rate of the form
%$$R \,=\, \frac{\Capacity'}{(1+\delta_a)^2(1\+ \rr/\log \M)} \,=\, \frac{\Capacity'}{(1+\delta_b)^2}.$$
%in which $\delta_b = b/\sqrt{2\log \M}$ controls the drop from capacity.
%This drop will be smaller in case $1$ than in case $0$.
%It is assumed throughout that $\M\ge 2$.
%Here $\Capacity'= \Capacity (1 -h')/(1+\delta_{sum}^2)$.

Recall that $z_x=z_x^{low}$, given by \eqref{eq:zxlow}, is a strictly increasing function of $x$, with values in the interval $I_0=[z_0,z_1]$ for $0\!\le\! x\!\le\! 1$. For values $z$ in $I_0$, let $x=x(z)$ be the choice for which $z_x=z$. With the rate $R$ of the form \eqref{eq:rateform}, let $x_r$ be the value of $x$ for which $z_x$ is 0. One finds that $x_r$ satisfies,
\eqq{1-x_r\,=\, \frac{1}{snr} \, \left[ \frac{\rr}{\log \M}\right].\label{eq:xrform}}
We now show that $h_{low}(x)$ is positive on $[0,x_r]$, for $r$ at least a certain value, which we call $r_1$.
%The method used in this subsection involves representation of $g_{low}(x)$ in terms of $z$.
%Also the distance of $x=x(z)$ from $1$ can be expressed as a function of $z$ as
%$$\,1-x = \frac{1}{snr} \, \left[ \frac{(1\+\delta_b)^2-(1\+\delta_a\+\delta_z)^2} {(1+\delta_a+\delta_z)^2}\right].$$

%Then if $z_{up}$ be the value of $z$ at which $(1\+\delta_a\+\delta_z)^2 = (1\+\delta_c)(1\+\delta_{a})^2$ and  $x_r=x(z_{up})$ be the corresponding value of $x$.  Then from the above we get that

%We remark that the value of $x_r$ is the same for all values of $c\ge 0$.
%The value of $x_r$ is used in giving an upper end point of a range of sufficient positivity of $g(x)-x$.
% As it is desirable that this upper endpoint be not far from $1$, one may restrict attention to cases with $snr \ge 1/\sqrt{2\log \M}$, say, so that $1-x_r$ remains not more than $2r/\sqrt{2\log \M}$.
\begin{lem}\label{posgdiff} {\em Positivity of $h_{low}(x)$ on $[0,x_r]$.}
%Equivalently, expressing the
Let rate $R$ be of the form \eqref{eq:rateform}, with
%the distance of $x_r=x_{null}$ from $1$ may be expressed as $$1-x_r = \frac{(1\!-\!\nu)}{\nu} \, \frac {\rr}{(1\+\delta_a)^2\log \M}.$$
$\rr > r_1$,  where
\eqq{ r_1 = \rhalf/2 + \frac{\sqrt{\log\M}}{\sqrt{\pi}(1 + \delta_a)}\label{eq:rone}.}
Then, for $0 \le x \le x_r$ the difference $h_{low}(x)$ is greater than or equal to
\eqq{gap = \frac{1}{snr} \, \left[\frac {r - r_1}{\log \M} \right].\label{eq:gap}}
%which is positive for $\dd$ at least $d_1-1/[4(1\+\delta_a)]$.
%Furthermore, it is obtained for all $0\le x\le x_r$,  that $g_L(x)-x$ is at least the following expression divided by $1+\delta_{sum}^2$, $$\frac{(1\!-\!\nu)}{\nu} \, \left[\frac {\dd-d_{comb}}{(1+\delta_a)^2 \,(\log \M)} \right] - \frac{\Capacity}{L\nu},$$ with $d_{comb} = d_1 -1/4$ In particular, arranging $\dd$ to be greater than $d_{comb}$, say $\dd=2d_{comb}$, and $L > 2 e^{2\Capacity} (1+\delta_a)^2\,(\log \M)/(d_{comb}\Capacity)$, then $$g_L(x)-x \ge \frac{(1\!-\!\nu)}{\nu} \, \left[\frac {d_{comb}}{2(1+\delta_{sum}^2)(1+\delta_a)^2 \,(\log \M)} \right],$$ for all $0 \le x \le x_r$, strictly positive with value of order $d_{comb}/\log \M$.  Accordingly, in this $c=0$ case, the drop from capacity, the portion $1-x_r$ un-decodable, and the amount by which $g_L(x)-x$ is above zero are all of order $1/\sqrt{2\log \M}$.
%In particular, expressing the rate as $R=\Capacity'/(1\+\delta_b)^2$, suppose \eqq{b= a \, +\, \frac{1}{\sqrt{2\pi}} + \frac{extra}{2\sqrt{2\log \M}}\label{eq:settingb}} with $extra > (1/2)(1-1/\pi)$.
%Then, for $0 \le x \le x_r$ the difference $g_{low}(x)-x$ is greater than or equal to \eqq{\frac{1}{snr} \, \left[ \frac{extra -(1/2)(1-1/\pi)}{(1\!+\delta_a)^2 \,(2\log \M)}\right].\label{eq:gapextra}}
\end{lem}

\begin{proof}[Proof of Lemma \ref{posgdiff}]
The function $g(x)$ has lower bound $g_{low}(x)$. By Corollary \ref{lowerbound}, $g_{low}(x)$ has derivative bounded by
%By Corollary 12, taking the case that $c=0$, the function $g(x)$ has derivative
%$$g_{low}'(x)= \int_{z_x^{low}}^{\infty} \frac{  \big(1\+\delta_a+\delta_z\big)^{2}}{(1+\delta_b)^2}\phi(z)dz,$$
%which like $g'(x)$ is less than
$$\int_{-\infty}^{\infty} \frac{  \big(1\+\delta_a+\delta_z\big)^{2}}{(1+\delta_a)^2(1 + r/\log\M)}\phi(z)dz = \frac{(1 + \rhalf/\log\M)}{(1 + \rr/\log\M)},$$
which is less than $1$ for $r \ge \rhalf$.  Thus $g_{low}(x)-x$ is decreasing as it has a negative derivative.  %Likewise by the correspondence between $x$ and $z$, it follows that $G(z)$ is decreasing in $z$, so $G(z)\ge G(0)$ for $z\le 0$.

To complete the proof, evaluate $g_{low}(x)-x$ at the point $x=x_r$.  The point $x_r$ is the choice where $z_x=0$.  After using \eqref{eq:glow}, it is seen that the value  $g_{low}(x_r)-x_r$ is equal to $gap$, where $gap$ is given by \eqref{eq:gap}.
%$$\frac{(d - d_1)(1\+\delta_a) + 1/4}{(1\+\delta_a)^2 \,\log \M}.$$
\end{proof}

%\vspace{0.3cm}
%In the above form of $b$, the terms are small enough that it provides a practical rate $\Capacity'/\big(1+b/\sqrt{2\log \M}\,\big)^2$ reasonably close to capacity with moderate $\M$.  Nevertheless, it would be nice to remove the $1/\sqrt{2\pi}$ part so that with suitable $a$ a rate closer to capacity is achieved for large $\M$.  Another way to say it is that we would like to arrange for $\rr$ to be of smaller order. For that reason we next take advantage of the modification to the power allocation in which it is slightly leveled using a positive $c$.

\subsection{Choices of $a$ and $r$ that control the overall rate drop} \label{choicepara}

Here we focus on the evaluation of $a$ and $r$ that optimize our summary expressions for the rate drop, based on the lower bounds on $g_L(x)\!-\!x$.  %Other subsections provide for numerical computation of values best for the exact form of $g_L(x)\!-\!x$.
%For the time being let's assume that for a particular $d_0$, yet to be specified and the conditions in part (a) of the above lemma are satisfied.
Recall that the rate of our inner code is
$$R= \Capacity \frac{1 - h}{(1\+\delta_a)^2 (1+\rr/\log \M)}.$$
Now, for $r > r_1$, the function $g_L(x)$ is accumulative on $[0,x_r]$, with positive $gap$ given by \eqref{eq:gap}.
Notice that $r_1$, given by \eqref{eq:rone}, satisfies,
\eqq{r_1 \leq 1/4 + \frac{\sqrt{\log \M}}{\sqrt{\pi}}\label{eq:ronebdd}.}

%With $0< \eta < gap$ and $f> f^*$ satisfying \eqref{eq:fcriterion}, per our theory, we reliably have that the weighted fraction mistakes $\dwght$ is given by \eqref{eq:deltawghtweuse}, which using the lower bound on $q_{1,m}$ and upper bound on $f$, given by Lemma \ref{progress}, can also be written as
%$$\dwght = (1 - x_r) - (gap - \eta)/2.$$

%Bounds on un-weighted fractions of failed detections and false alarms are obtained by multiplying the weighted fractions by the factor $$\fac = \frac{snr }{2\Capacity}.$$
%To see this notice that for a given weighted fraction, the maximum possible un-weighted fraction would be if we assume that all the failed detection or false alarms came from the section with the smallest weight. This would correspond to the section with weight $\pi_{(L)}$, where we recall that $\pi_{(L)} = 2\Capacity/(L\,snr).$

%For large $L$, $\delta_{sum}^2$ would be like $\delta_c^2/2$ so we can assume that $\fac$ is bounded by $snr/(2\Capacity).$  Notice that $snr/(2\Capacity)$ is greater that 1 and is near 1 for small $snr$.
%For large $L$, $\delta_{sum}^2$ would be like $\delta_c^2/2$ so we can assume that $\fac$ is bounded by $snr/(2\Capacity).$  Notice that $snr/(2\Capacity)$ is greater that 1 and is near 1 for small $snr$.

Consequently, from Theorem \ref{reliabddthm}, with high reliability, the total fraction of mistakes $\hat{\delta}_{mis}$ is bounded by
$$\delta_{mis} = \frac{snr}{2\Capacity}\left[(1-x_r) \,- \, (gap-\eta)/2\right].$$
%using \eqref{eq:deltawghtweuse} and \eqref{eq:deltamisdeltawghtrel}.
If the outer Reed-Solomon code has distance designed to be at least $\delta_{mis}$ then any occurrences of a fraction of mistakes less than $\delta_{mis}$ are corrected.  The overall rate of the code is $R_{total}$, which is at least $(1-\delta_{mis})R$.

Sensible values of the parameters $a$ and $r$ can be obtained by optimizing the above overall rate under a presumption of small error probability, using simplifying approximations of our expressions.  Reference values (corresponding to large $L$) are obtained by considering what the parameters become with $\eta = 0$, $f = f^*$, and $h = 0$.

%The value $a$ will depend on $r$ through \eqref{eq:aexpr}, where
Notice that $a$ is related to $f^*$ via the bound \eqref{eq:fstarbound}. Set $a$ so that \eqq{a \sqrt{2\log \M} = \log \left[1/\big(f^* \sqrt{2\pi} \sqrt{2\log \M}\big)\right].\label{eq:aexpr}}
We take $f^*$ as $gap ^2/8$ as per Lemma \ref{progress}. Consequently $a$ will depend on $r$ via the expression of $gap$ given by \eqref{eq:gap}.

Next, using the expressions for $1-x_r$ and $gap$, along with $\eta = 0$, yields a simplified approximate expression for the mistake rate given by
$$\delta_{mis} = \frac{r +  r_1}{4\Capacity \,\log \M} .$$

%using the expression for $\rhalf$.

Accordingly, the overall communication rate may be expressed as,
$$R_{total} = \left(1\!-\!\frac{r +  r_1}{4\Capacity \,\log \M}\right) \frac{\Capacity }{(1\+\delta_a)^2 (1\+ \rr/\log \M)}.$$

As per these calculations (see \cite{josephPhD} for details) we find it appropriate to take
$r$ to be $r^*$, where
%$$\rr^* = r_1 + 4/\omega_1,\quad\text{with $\omega_1 = 2 + 1/\C$}.$$
$$\rr^* = r_1 + 2/(1 + 1/\C).$$
Also, the corresponding $a$ is seen to be $$a = (3/2)\log(\log(\M))/\sqrt{2\log(\M)} + \tilde a,$$ where,
$$\tilde a = \frac{2\log\left[snr\, (1 + 1/\C) /((\pi)^{.25})\right]}{\sqrt{2\log(\M)}}.$$

Express $R_{total}$ in the form $\C/(1 + drop)$.
%The with the above choices of $r$ and $a$ one gets that $drop_{num}$ is $3\log\log \M$ plus the following
%$$4\log \big(snr \omega_1/2\big) -\log (\pi) + (\omega_1 + \omega_2)r_1 + 4.$$
%The above is bounded by,
%$$4\log \big(snr\omega_1/2\big) + 2(1 + 1/\C)r_1 + 1.47.$$
Then with the above choices of $r$ and $a$, and the bound on $r_1$ given in \eqref{eq:ronebdd}, one sees that $drop$  can be approximated by
$$\frac{ 3\log\log \M + 4\log(\omega_1\,snr)  + 1/(4\C) + 3.35}{2\log \M} + \frac{1 + 1/(2\C)}{\sqrt{\pi\log \M}},$$
where $\omega_1 = 1 + 1/\C$.

We remark that the above explicit expressions are given to highlight the nature of dependence of the rate drop
on $snr$ and $\M$. These are quite conservative. For more accurate numerical evaluation see section \ref{sec:compillus} on computational illustrations.

\subsection{Definition of $\Capacity^*$ and proof of Proposition \ref{mainthm}} \label{proofmainthm}

In the previous subsection we gave the value of $r$ and $a$ that maximized, in an approximate sense, the outer code rate for given $snr$ and $\M$ values and for large $L$. This led to explicit expressions for the maximal achievable outer code rate as a function of $snr$ and $\M$. We define $\Capacity^*$ to be the inner code rate corresponding to this maximum achievable outer code rate. Thus,
 $$\Capacity^*= \frac{\Capacity}{(1\+\delta_a)^2\,[1+ r^*/\log \M]}.$$
Similar to above, $\Capacity^*$ can be written a $\Capacity/(1 + drop^*)$ where $drop^*$ can be approximated by
$$\frac{ 3\log\log \M + 4\log(\omega_1 snr)  + 4/\omega_1  -2}{2\log \M} + \frac{1}{\sqrt{\pi\log \M}},$$
with $\omega_1 = 1 + 1/\C$. We now give a proof of our main result.

\vspace{.3cm}

\noindent\begin{proof}[ {\bf Proof of Proposition \ref{mainthm}}] Take $r = r^* + \kappa$. Using $$(1 + \kappa/\log \M)(1 + r^*/\log \M) \geq (1 + r/\log \M),$$ we find that for the rate $R$ as in Proposition \ref{mainthm}, $gap$ is at least $(r - r_1)/(snr\log \M)$ for $x \leq x_r$, with $x_r = r/(snr \log \M)$.

Take $f^* = (1/8)(r^* - r_1)^2/ (snr \log \M)^2$, so that $a$ is the same as given in the previous subsection. Now, we need to select $\rho > 1 $ and $\eta >0$,
so that $$f = \rho f^* \leq (gap - \eta)^2/8 - 1/(2L_\pi).$$
Take $\omega = (1 + 1/\C)/2$, so that $r^* = r_1 + 1/\omega$. One sees that we can satisfy the above requirement by taking  $\eta$ as  $(1/2)\kappa/(snr \log \M)$ and $\rho  =  (1 +  \kappa\omega/ 2)^2 - \epsilon_L$
$$\epsilon_L  =  \frac{(2\omega snr \log\M)^2}{L_\pi},$$
 is of order $(\log M)^2/L$, and hence is negligible compared to the first term in $\rho$. Since it has little effect on the error exponent, for ease of exposition, we ignore this term. We also assume that
 $f = (gap - \eta)^2/8$, ignoring the $1/(2L_\pi)$ term.

%This choice makes $f_{1,m}$ at most $(gap - \eta)/2$ which is $\sqrt{c}/(2\omega snr \log \M).$
We select $$h = \frac{\kappa}{(2\log \M)^{3/2}}.$$
%which is indeed at most $(2 f_{1,m}\, snr)/\sqrt{2\log \M}$ as required.

The fraction of mistakes, $$\delta_{mis} = \frac{snr}{2\Capacity}\left[ \frac{r}{snr \log \M} - (gap - \eta)/2\right]$$
is calculated as in the previous subsection, except here we have to account for the positive $\eta$. Substituting the
expression for $gap$ and $\eta$ gives the expression for $\delta_{mis}$ as in the proposition.

Next, let's look at the error probability. The error probability is given by
$$m e^{-2L_\pi \eta^2+m c_0} + m e^{-L_\pi f D(\rho)/\rho} + m e^{m h} e^{ - n h^2/2}.$$
Notice that $n h^2/2$ is at least $(L_\pi \log \M)h^2/(2\Capacity^*)$, where we use that $L \geq L_\pi$ and $R \leq \Capacity^*$.% which is at most $snr/2$.
 Thus the above probability is less than
$$\kappa_1  \exp\{-L_\pi  \min\{2\eta^2,\, f^* D(\rho),\,  h^2\log \M/(2\Capacity^*)\}\} $$
with $$\kappa_1 = 3m\,e^{m \max\{c_0, 1/2\}},$$
where for the above we use $h < 1$.

Substituting, we see that $2\eta^2$ is $(1/2)\kappa^2/(snr \log \M)^2$ and $h^2\log \M/(2\Capacity^*)$ is $$\frac{1}{16 \Capacity^*}\frac{\kappa^2}{(\log \M)^2}.$$
Also, one sees that $ D(\rho)$ is at least $2(\sqrt{\rho} - 1)^2/\rho$. Thus the term $f^* D(\rho)$ is at least
$$\frac{\kappa^2}{(4 snr \log \M)^2 (1 + \kappa\omega/2) }.$$
We bound from below the above quantity  by considering two cases viz. $\kappa \leq 2/\omega$ and $\kappa > 2/\omega$.  For the first case we have $1 + \kappa\omega/2 \leq 2$, so this quantity is bounded from below by
$(1/2)\kappa^2/(4 snr\log \M)^2.$ % which using $\omega \geq 1$ is at least  $\kappa^2/(32 snr\log \M).$
For the second case use $\kappa/(1 + \kappa \omega/2)$ is bounded from below by
$1/\omega$, to get that this term is at least $(1/\omega)\kappa/(4 snr\log \M)^2.$

Now we bound from below the quantity $\min\{2\eta^2, f^* D(\rho),\,  h^2\log \M/(2\Capacity^*)\}$ appearing in the exponent. For $\kappa \leq 2/\omega$ this quantity is bounded from below by $$\kappa_3\frac{\kappa^2}{(\log \M)^2},$$
where
$$\kappa_3 = \min\left\{1/(32snr^2),\, 1/(16\Capacity^*)\right\}.$$
For $\kappa > 2/\omega$ this is quantity is at least $$ \min\left \{\kappa_3\frac{\kappa^2}{(\log \M)^2}, \kappa_4\frac{\kappa}{\log \M} \right\},$$
with
$$\kappa_4 = \frac{1}{8(1 + 1/\C)snr^2 \log\M}.$$

Also notice that $\Capacity^* - R$ is at most $\Capacity^*\kappa/\log \M$.  Thus we have that
$$\min\{2\eta^2, f^* D(\rho),\,  h^2\log \M/(2\Capacity^*)\}$$ is at least
$$ \min\left \{\kappa_3(\Delta^*)^2,\, \kappa_4\Delta^* \right\}.$$
%Denoting the first minimum as $\tilde \kappa_2$, that above is at least $\tilde \kappa_2\min\left \{(\Delta^*)^2, \Delta^* \right\}$.% where $$\Delta^* = (\Capacity^* - R)/\Capacity^*.$$ %and
%$$\tilde \kappa_2 =  \min\left\{1/(32 snr), 1/(2snr^2)\right\}.$$
Further, recalling that $L_{\pi} = L\nu /(2\Capacity)$, we get that
$\kappa_2  = \nu /(2\Capacity),$ which is near $\nu/(2\Capacity)$.

Regarding the value of $m$, recall that $m$ is at most $2/(gap - \eta)$. Using the above we get that
$m$ is at most $(2\omega snr)\log \M$. Thus ignoring the $3m$, term $\kappa_1$ is polynomial in $\M$ with power
$2\omega snr \max\{c_0, 1/2\}$.

Part II is exactly similar to the use of Reed-Solomon codes in section  VI of our companion paper \cite{barron2010joseph}.
\end{proof}

\vspace{.3cm}
In the proof of Corollary \ref{cor:cor}, we let $\zeta_i$, for integer $i$, be constants that do not depend on $L,\, \M$ or $n$.
\begin{proof}[{\bf Proof of Corollary \ref{cor:cor}}] %Note, unlike in Proposition \ref{mainthm},
%here we denote $R$ as the overall rate of communication, and denote $\Delta = (\C - R)/\C$.
Recall $R_{tot} = (1 - \delta_{mis})R$.
 Using the form of $\delta_{mis}$ and $\C^*$ for Proposition \ref{mainthm}, one sees that $R_{tot}$ may be expressed as,
\eqq{R_{tot} = \left(1 - \zeta_1 \delta_M - \zeta_2 \frac{\kappa}{\log \M}\right)\C.\label{eq:rcor}}
Notice that $\M$ needs to be at least $\exp\{\zeta_3/\Delta^2\}$, where $\Delta = (\C - R_{tot})/\C$, for above to be satisfied. For a given section size $\M$, the size of $\kappa$ would be larger for a larger $\C - R_{tot}$. Choose $\kappa$ so that $\zeta_2\kappa/\log\M$ is at least $\zeta_1 \delta_M$, so that by \eqref{eq:rcor}, one has,
\eqq{\Delta \geq 2\zeta_2 \frac{\kappa}{\log\M}.\label{eq:deltalowbdd}}
Now following the proof of Proposition \ref{mainthm}, since the error exponent is of the form $const \min\{\kappa/\log\M, (\kappa/\log\M)^2\}$, one sees that it is at least $const \min\{\Delta, \Delta^2\}$ from \eqref{eq:deltalowbdd}.
\end{proof}

\section{Discussion}

The paper demonstrated that the sparse superposition coding scheme, with the adaptive successive decoder and outer Reed-Solomon code, allows one to communicate at any rate below capacity, with block error probability that is exponentially small in $L$. It is shown in \cite{barron2010ajoseph} that this exponent can be improved  by a factor of $\sqrt{\log\M}$ from using a Bernstein bound on the probability of the large deviation events analyzed here.
%, which accounts for the variance in the sums of the indicators of random variables. See for analysis including this bound.

For fixed section size $\M$, the power allocation \eqref{eq:powerweuse} analyzed in the paper, allows one to achieve any $R$ that is at least a drop of $1/\sqrt{\log \M}$ of $\C$. In contrast, constant power allocation allows us to achieve rates up to a threshold rate $R_0 = .5 snr/(1 + snr)$, which is bounded by $1/2$, but is near $\C$ for small $snr$. In \cite{josephPhD}, \cite{barron2010ajoseph} the alternative power allocation \eqref{eq:modpower} is shown to allow for rates that is of order $\log \log \M/\log\M$ from capacity. Our experience shows that it is advantageous to use different power allocation schemes depending on the  regime for $snr$. When $snr$ is small, constant power allocation works better. The power allocation with leveling \eqref{eq:modpower} works better for moderately large $snr$, whereas \eqref{eq:powerweuse} is appropriate for larger $snr$ values.

One of the requirements of the algorithm, as seen in the proof of Corollary \ref{cor:cor}, is that for fixed rate $R_{tot}$, the section size $\M$ is needed to be exponential in $1/\Delta$, using power allocation \eqref{eq:powerweuse}. Here $\Delta$ is the rate drop from capacity. Similar results hold for the other power allocations as well. However, this was not the case for the optimal ML--decoder, as seen in \cite{barron2010joseph}. Consequently, it is still an open question whether there are practical decoders for the sparse superposition coding scheme which do not have this requirement on the dictionary size.

%The error probability here is exponentially small in $L$ for any rate $R$ below $\C$.

\appendices
\section{Proof of Lemma \ref{laterstepdist}} \label{sec:disanalzcalk}

For each $k \geq 2$, express $X$ as,
$$X = \frac{G_1}{\|G_1\|} \Zcal_1^{\trn} \,+\, \ldots \,+\, \frac{G_{k-1}}{\|G_{k-1}\|} \Zcal_{k-1}^{\trn} \,+\, \xi_k V_k,$$
where $\xi_k = [\xi_{k,k} : \ldots : \xi_{k,n}]$ is an $n \times (n - k + 1)$ orthonormal matrix, with columns $\xi_{k,i}$, for $i = k, \ldots, n$, being orthogonal to $G_1, \ldots, G_{k-1}$. There is flexibility in the choice of the $\xi_{k,i}$'s, the only requirement being that they depend on only $G_1,\ldots, G_{k-1}$ and no other random quantities. For convenience, we take these $\xi_{k,i}$'s to come from the Grahm-Schmidt orthogonalization of $G_1, \ldots, G_{k-1}$ and the columns of the identity matrix.

The matrix $V_k$, which is $(n - k+1)\times N$ dimensional, is also denoted as,
$$V_k = [V_{k,1} : V_{k,2} : \ldots : V_{k,N}].$$
The columns $V_{k,j}$, where $j = 1,\ldots, N$ gives the coefficients of the expansion of the column $X_j$ in the basis $\xi_{k,k}, \xi_{k, k +1}, \ldots, \xi_{k, n}$. We also denote the entries of $V_k$ as $V_{k,i,j}$, where $i = k, \ldots, n$ and $j = 1,\ldots, N$.

%Consider the general $k \ge 2$ case. Focus on the sequence of coefficients $$\Zcal_{1,j},  \Zcal_{2,j}, \ldots, \Zcal_{k-1,j}, \; V_{k,k,j},V_{k+1,k,j}, \ldots, V_{n,k,j}$$ used to represent $X_j$ for $j$ in $J_{k-1}$ in the basis $$\frac{G_1}{\|G_1\|}, \; \frac{G_2}{\|G_2\|} , \ldots ,  \frac{G_{k-1}}{\|G_{k-1}\|}, \; \xi_{k,k},\xi_{k+1,k}, \ldots \xi_{n,k},$$ where the $\xi_{i,k}$ for $i$ from $k$ to $n$ are orthogonal vectors in $\Rbb^n$, orthogonal to the $G_1, G_2, \ldots, G_{k-1}$.  These are associated with the representation $$X_j =\sum_{k'=1}^{k-1} \Zcal_{k',j} \, G_{k'}/\|G_{k'}\| +V_{k,j}$$ described above, except that here $V_{k,j}$ is represented as $\sum_{i=k}^n V_{i,k,j}\, \xi_{i,k}$.

We prove that conditional on $\Fcal_{k -1}$, %the distribution of the $V_{k,i,j}$ is independent across $i$ from $k$ to $n$, and for each such $i$ the joint distribution of
the distribution of $(V_{k,i,j}: j \in J_{k-1})$, for $i = k,\ldots,n$, is i.i.d. Normal $N(0,\Sigma_{k-1})$.  The proof is by induction.

The stated property is true initially, at $k\!=\!2$, from Lemma \ref{lem:firststepdist}. Recall that the rows of the matrix $U$ in Lemma \ref{lem:firststepdist} are i.i.d. $N(0,\Sigma_1)$. Correspondingly, since $V_2 = \xi_2^{\trn} U$, and since the columns of $\xi_2$ are orthonormal, and independent of $U$, one gets that the rows of $V_2$ are i.i.d $N(0,\Sigma_1)$ as well.

%Correspondingly, since for each $i = 2,\ldots,n$, the element $V_{2,i,j}$ is simply the projection of the $j$ th column of $U$ in the direction $\xi_{2,i}$, where $j \in J$, it follows that the orthonormality of $\xi_{2,i}$'s that the stated property holds for $k = 2$.

Presuming the stated conditional distribution property to be true at $k$, we conduct analysis, from which its validity will be demonstrated at $k + 1$.  Along the way the conditional distribution properties of $G_k$, $Z_{k,j}$, and $\Zcal_{k,j}$ are obtained as consequences.  As for $\hat w_k$ and $\delta_k$ we first obtain them by explicit recursions and then verify the stated form.

Denote as
\eqq{\gcoef_{k,i} =  -\sum_{j \in dec_{k-1}} \sqrt{P_j} \, V_{k,i,j} \quad\text{for $i = k,\ldots,n$ }.\label{eq:gcoeff}}
Also denote as,
$$\gcoef_k = (\gcoef_{k,k} ,\, \gcoef_{k,k+1}, \ldots, \gcoef_{k,n})^{\trn}.$$
The vector $\gcoef_k$ gives the representation of $G_k$ in the basis consisting of columns vectors of $\xi_k$. In other words, $G_k = \xi_k \gcoef_k$.

 Notice that,
\eqq{\Zcal_{k,j} =  V_{k,j}^{\trn} \gcoef_k \, /\|\gcoef_k\|.\label{eq:repzcalbasisone}}

%From the representation of $V_{k,j}$ in the basis given above, we have representation of $G_k$ in the same basis as $G_{i,k} = \sum_{j \in dec_{k-1}} \sqrt{P_j} \, V_{i,k,j}$ for $i$ from $k$ to $n$.  The coordinates less than $k$ are $0$, since the $V_{k,j}$ and $G_k$ are orthogonal to $G_1,\ldots, G_{k-1}$.  The value of $\Zcal_{k,j}$ is $V_{k,j}^T G_k \, /\|G_k\|$ where the inner product (and norm) may be computed in the above basis from sums of products of coefficients for $i$ from $k$ to $n$.

%To resolve components of $\Zcal_{k,j}$
Further, since $V_{k,j}$ and $\gcoef_k$ are jointly normal conditional on $\Fcal_{k-1}$, one gets, through conditioning on $\gcoef_k$ that,
$$V_{k,j} = b_{k- 1,j} \, \gcoef_k/\sigma_k + U_{k,j}.$$
%is used with values of $b_{k-1,j}$ following an update rule that will be specified (depending on $\Fcal_{k-1}$).%, where $U_{k,j}$ has coefficients $U_{i,k,j} = V_{i,k,j} - b_{k,j} \, G_{i,k}/\sigma_k$ for $i$ from $k$ to $n$, when representing the vector $U_{k,j}$ in terms of the basis built from the $\xi_{i,k}$.
Denote as $U_k = [U_{k,1} : U_{k,2} : \ldots : U_{k,N}]$, which is an $(n - k + 1)\times N$ dimensional matrix like $V_k$.  The entries of $U_k$ are denoted as $U_{k,i,j}$, where $i = k, \ldots, n$ and $j = 1,\ldots, N$. The matrix  $U_{k}$ is independent of $\gcoef_k$,  conditioned on $\Fcal_{k-1}$. Further, from the representation \eqref{eq:repzcalbasisone}, one gets that
\eqq{\Zcal_{k,j} \,=\,  b_{k-1,j} \, \|\gcoef_k\|/\sigma_k \,+ \, Z_{k,j},\label{eq:repzcalbasistwo}}
with,
$$Z_{k,j} = U_{k,j}^{\trn}\gcoef_k/\|\gcoef_k\|.$$

For the conditional distribution of $\gcoef_{k,i}$ given $\Fcal_{k-1}$, independence across $i$, conditional normality and conditional mean $0$ are properties inherited from the corresponding properties of the $V_{k,i,j}$.
To obtain the conditional variance of $\gcoef_{k,i}$, given by \eqref{eq:gcoeff}, use the conditional covariance $$\Sigma_{k-1}\!=\! I\!-\! \delta_{k-1}\delta_{k-1}^T$$ of $(V_{k,i,j}: j \in J_{k-1})$.  The identity part contributes $\sum_{j \in dec_{k-1}} P_j$ which is $(\hat q_{k-1} + \hat f_{k-1})P$; whereas, the $\delta_{k-1}\delta_{k-1}^T$ part, using the presumed form of $\delta_{k-1}$, contributes
an amount seen to equal $\nu_{k-1} [\sum_{j \in sent \intersect dec_{k-1}} P_j/P]^2 \, P$ which is $\nu_{k-1} \hat q_{k-1}^2 P$.
%for distinct $j,j'$ in $sent \intersect J_{k-1}$ with $\Hcal_{k-1,j}$, $\Hcal_{k-1,j'}$ true.  The rest %of the $\hat f_{k-1} \, L$ terms in the sum are uncorrelated.  Accumulating the terms and simplifying using $\nu_{k-1} = L \delta_{k-1}$,
It follows that the conditional expected square for $\gcoef_{k,i}$, for $i = k, \ldots, n$ is
%$$\sigma_k^2 \,=\, L \hat q_{k-1} (1\!-\!\delta_{k-1})+ L \hat f_{k-1} - L \hat q_{k-1}(L \hat q_{k-1}\!-\!1)\delta_{k-1}.$$ Using $\nu_{k-1}= L\delta_{k-1}$, it simplifies to
$$\sigma_k^2\,=\, \big[\,\hat q_{k-1} +\hat f_{k-1} - \hat q_{k-1}^2 \, \nu_{k-1}\,] \, P.$$

Conditional on $\Fcal_{k-1}$, the distribution of $$\|\gcoef_k\|^2= \sum_{i=k}^n (\gcoef_{k,i})^2$$ is that of $\sigma_k^2 \, \Chi_{n-k+1}^2$, a multiple of a Chi-square with $n\!-\!k+1$ degrees of freedom.

Next we compute $b_{k-1,j}$ in \eqref{eq:repzcalbasistwo}, which is the value of $$\E[V_{k,i,j} \gcoef_{k,i} |\Fcal_{k-1}]/\sigma_k$$ for any of the coordinates $i=k,\ldots,n$.  Consider the product $V_{k,i,j} \, \gcoef_{k,i}$ in the numerator. Using the representation of $\gcoef_{k,i}$ in \eqref{eq:gcoeff}, one has
$\E[V_{k,i,j} \gcoef_{k,i} |\Fcal_{k-1}]$ is
$$-\sum_{j' \in dec_{k-1}} \sqrt{P_{j'}} \big[1_{j'=j} - \delta_{k-1,j} \delta_{k-1,j'}\big],$$ which simplifies to $-\sqrt{P_j}\, \big[ 1_{j \in dec_{k-1}} - \nu_{k-1} \hat q_{k-1} 1_{j \: sent}\big]$.
%for terms with $\Hcal_{k-1,j}$ satisfied, this product is $V_{i,k,j}^2$ plus, if $j$ is in $sent \intersect J_{k-1}$, there are $L \hat q_{k-1} - 1$ occurrences of $V_{i,k,j} V_{i,k,j'}$ that have $\Hcal_{k-1,j'}$ satisfied and $j'$ in $sent \intersect J_{k-1}$ distinct from $j$.  On the other hand if $\Hcal_{k-1,j}$ is not satisfied, there are $L \hat q_{k-1}$ occurrences of such $V_{i,k,j} V_{i,k,j'}$. Thus $\E[\,V_{i,k,j} G_{i,k}\, |\Fcal_{k-1}]$, for $j$ in $J_{k-1}$, is found to equal $1_{\Hcal_{k-1,j}} \, - \, \hat q_{k-1} \, L \, \delta_{k-1} \, 1_{j \: sent}$.
So for $j$ in $J_k = J_{k-1}-dec_{k-1}$, we have the simplification
%$b_{k,j} \,=\, - \, b_k \, 1_{j \: sent}$ with
\eqq{b_{k-1,j} =  \,\frac{\hat q_{k-1} \, \nu_{k-1} \beta_j}{\sigma_k}.\label{eq:bkm1j}}
%$$b_k =  \frac{L_{k-1} \, \delta_{k-1} }{\sigma_k}.$$
%where $\nu_{k-1} = L\delta_{k-1}$.
%Accordingly, in terms of $\hat q_{k-1}$ and $\nu_{k-1}$, the square of $b_{k,j}$
Also, for $j,j'$ in $J_k$, the product takes the form
$$b_{k-1,j} b_{k-1,j'}= \delta_{k-1,j} \delta_{k-1,j'} \frac{ \hat q_{k-1} \nu_{k-1}}{1+ \hat f_{k-1}/\hat q_{k-1} - \hat q_{k-1} \nu_{k-1}}.$$
Here the ratio simplifies to $\hat q_{k-1}^{adj} \nu_{k-1}/(1- \hat q_{k-1}^{adj} \nu_{k-1})$.

Now determine the features of the joint normal distribution of the $$U_{k,i,j}= V_{k,i,j} - b_{k-1,j} \, \gcoef_{k,i}/\sigma_k,$$ for $j \in J_k$, given $\Fcal_{k-1}$.  Given $\Fcal_{k-1}$, the $(U_{k,i,j}: j \in J_{k})$ are i.i.d across choices of $i$, but there is covariance across choices of $j$ for fixed $i$.  This conditional covariance $\E[U_{k,i,j} U_{k,i,j'} |\Fcal_{k-1}]$, by the choice of $b_{k-1,j}$, reduces to $\E[V_{k,i,j}V_{k,i,j'} |\Fcal_{k-1}] - b_{k-1,j} b_{k-1,j'}$ which, for $j \in J_k$, is $$1_{j=j'}-\delta_{k-1,j}\delta_{k-1,j'} - b_{k-1,j} b_{k-1,j'}.$$
That is, for each $i$, the $(U_{k,i,j} : j \in J_k)$ have the joint $N_{J_k}(0,\Sigma_{k})$ distribution, conditional on $\Fcal_{k-1}$, where $\Sigma_k$ again takes the form $1_{j,j'}-\delta_{k,j} \delta_{k,j'}$ where
% $\delta_k$ takes the form $\delta_{k-1} + b_{k}^2$ equal to
$$\delta_{k,j} \delta_{k,j'}\, = \, \delta_{k-1,j} \delta_{k-1,j'} \left\{ 1\,+\, \frac{\hat q_{k-1}^{adj} \,\nu_{k-1}}{1 - \hat q_{k-1}^{adj} \nu_{k-1}} \right\},$$
for $j,j'$ now restricted to $J_{k}$. The quantity in braces simplifies to $1/(1-\hat q_{k-1}^{adj} \nu_{k-1})$.
Correspondingly, the recursive update rule for $\nu_k$ %$= L \delta_k$
is
$$\nu_k \, = \, \frac{\nu_{k-1}}{1\,-\, \hat q_{k-1}^{adj} \, \nu_{k-1}}.$$
%$$\delta_k \, = \, \delta_{k-1} \left\{ 1\,+\, \frac{\hat q_{k-1} \,\nu_{k-1}}{1 + (\hat f_{k-1} /\hat q_{k-1}) - \hat q_{k-1} \nu_{k-1}} \right\}.$$
%Accordingly, $\nu_k$ defined by L_k \delta_k  =$ $(1\!-\!\hat q_{k-1})L \delta_k$ is given by $$\nu_k= \, (1\!-\hat q_{k-1}) \, \nu_{k-1} \, \left\{ 1\, + \,  \frac{\hat q_{k-1} \,\nu_{k-1}}{1 + (\hat f_{k-1} /\hat q_{k-1}) - \hat q_{k-1} \nu_{k-1}} \right\},$$It yields the upper bound $\nu_k \le \nu_{k-1} (1\!-\!\hat q_{k-1})/(1\!-\!\hat q_{k-1} \nu_{k-1})$, which is not more than $\nu_{k-1}$.

Consequently, the joint distribution for $(Z_{k,j} : j \in J_k)$ is determined, conditional on $\Fcal_{k-1}$.  It is also the normal $N(0,\Sigma_{k})$ distribution and $(Z_{k,j} : j \in J_k)$ is conditionally independent of the coefficients of $\gcoef_k$, given $\Fcal_{k-1}$. After all, the $$Z_{k,j} = U_{k,j}^{\trn} \gcoef_k \,/\|\gcoef_k\|$$ have this $N_{J_k}(0,\Sigma_{k})$ distribution, conditional on $\gcoef_k$ and $\Fcal_{k-1}$, but since this distribution does not depend on $\gcoef_k$ we have the stated conditional independence.

%From this, it follows that the joint distribution of  $\Zcal_{J_k}= (\Zcal_{k,j} : j \in J_k)$ given $\Fcal_{k-1}$ is determined.
%Now $\Zcal_{k,j} = X_j^T  G_k \, /\|G_k\|$ reduces to $V_{k,j}^T \gcoef_k \,/\|\gcoef_k\|$ by the orthogonality of the $G_1, \ldots, G_{k-1}$ and $\xi_k$.  So using the representation $V_{k,j} = b_{k-1,j} \, \gcoef_k/\sigma_k + U_{k,j}$ one obtains  $$\Zcal_{k,j} \,=\,  b_{k-1,j} \, \|\gcoef_k\|/\sigma_k \,+ \, Z_{k,j}.$$
%where the $(Z_{k,j}: j \in J_k)$ and $G_k$ are conditionally independent with the indicated distributions given $\Fcal_{k-1}$.
This makes the conditional distribution of the $\Zcal_{k,j}$, given $\Fcal_{k-1}$, as given in \eqref{eq:repzcalbasistwo}, a location mixture of normals with distribution of the shift of location determined by the Chi-square distribution of $\Chi_{n-k+1}^2=\|\gcoef_k\|^2/\sigma_k^2$.
Using the form of $b_{k-1,j}$, for $j$ in $J_k$, the location shift $b_{k-1,j} \, \Chi_{n-k+1}$ may be written $$ \sqrt {\hat w_{k} \, \cjrb} \, \big[\, \Chi_{n-k+1}/\sqrt{n} \, \big]\, 1_{j \: sent},$$ where $$\hat w_{k} = \frac{n\, b_{k,j}^2}{\cjrb}.$$ The numerator and denominator has dependence on $j$ through $P_j$, so canceling the $P_j$ produces a value for $\hat w_{k}$.  Indeed, $\cjrb=(P_j /P)\nu (L/R)\log \M$ equals $n(P_j/P)\nu$ %expressed as $2(R_{th}/R)\log \M = \nu n/L$
and $b_{k-1,j}^2 = P_j \hat q_{k-1}^{adj} \, \nu_{k-1}^2/[1-\hat q_{k-1}^{adj} \nu_{k-1}]$. %$ = [\nu_{k-1}/L]\hat q_{k-1} \,\nu_{k-1}/\big[1 + (\hat f_{k-1} /\hat q_{k-1}) - \hat q_{k-1} \nu_{k-1}\big]$,
So this $\hat w_{k}$ may be expressed as %[CHECK THIS]
$$\hat w_k \, = \,%\frac{n\!-\!k\+1}{n} \,\,
\frac{\nu_{k-1}}{\nu} \,\, \frac{\hat q_{k-1}^{adj} \,\nu_{k-1}}{1 - \hat q_{k-1}^{adj} \nu_{k-1}},$$
which, using the update rule for $\nu_{k-1}$, is seen to equal
$$\hat w_k = %\left( 1 - \frac{k\!-\!1}{n} \right)
\frac{ \nu_{k-1} - \nu_k}{\nu}.$$
Further, repeatedly apply $\nu_{k'}/\nu_{k'-1} \!=\! 1/(1\!-\!\hat q_{k'-1}^{adj} \, \nu_{k'-1})$, for $k'$ from $k$ to $2$, each time substituting the required expression on the right and simplifying to obtain
$$\frac{\nu_{k}}{\nu_{k-1}} \,\, = \,\, \frac{1 \,- \, (\hat q_1^{adj} + \ldots + \hat q_{k-2}^{adj} ) \,\, \nu \:\:\quad}{1 - (\hat q_1^{adj} + \ldots + \hat q_{k-2}^{adj} + \hat q_{k-1}^{adj} ) \, \nu}.$$
This yields $\nu_k = \nu \hat s_k$, which, when plugged into the expressions for $\hat w_k$, % and $\delta_k = \nu_k/L$,
establishes the form of $\hat w_k$ as given in the lemma. %We also see that $\nu_k^{red} = [1-(\hat q_1^{adj} + \ldots  + \hat q_{k-1}^{adj} )] \,\nu \,\hat s_k$ is less than $\nu$.
%Now the vector $V_{k,j}$ updates to $V_{k+1,j} = V_{k,j} - \Zcal_{k,j} G_k/\|G_k\|$, as we have said, in accordance with the representation $V_{k,j} = \Zcal_{k,j} G_k/\|G_k\|+ V_{k+1,j}$ with $V_{k+1,j}$ orthogonal to $G_k$ as well as orthogonal to $G_1,\ldots $G_{k-1}$.

We need to prove that conditional on $\Fcal_k$ that the rows of $V_{k+1}$, for $j \in J_k$, are i.i.d. $N_{J_k}(0, \Sigma_k)$. Recall that $V_{k+1} = \xi_{k+1}^{\trn} X$. Since the column span of $\xi_{k+1}$ is contained in that of $\xi_k$, one may also write $V_{k+1}$ as $\xi_{k+1}^{\trn}\xi_k V_k$. Similar to the representation $G_k = \xi_k \gcoef_k$, express the columns of $\xi_{k+1}$ in terms of the columns of $\xi_k$ as $\xi_{k+1} = \xi_k \xi_{k}^{coef}$, where $\xi_k^{coeff}$ is an $(n - k + 1)\times (n - k)$ dimensional matrix. Using this representation one gets that $V_{k+1} = (\xi_k^{coef})^{\trn} V_k$.

Notice that $\xi_k$ is $\Fcal_{k-1}$ measurable and that $\xi_{k+1}$ is $\sigma\{\Fcal_{k-1} , G_k\}$ measurable. Correspondingly, $\xi_k^{coef}$ is also $\sigma\{\Fcal_{k-1} , G_k\}$ measurable. Further, because of the orthonormality of $\xi_k$ and $\xi_{k+1}$, one gets that the columns of $\xi_k^{coef}$ are also orthonormal. Further, as $G_k$ is orthonormal to $\xi_{k+1}$, one has the $\gcoef_k$ is orthogonal to the columns of $\xi_k^{coef}$ as well.

Accordingly, one has that $V_{k+1} = (\xi_k^{coef})^{\trn} U_k$. Consequently, using the independence of $U_k$ and $\gcoef_k$, and the above, one gets that conditional on $\sigma\{\Fcal_{k-1} , G_k\}$, for $j \in J_k$, the rows of $V_{k+1}$ are i.i.d. $N_{J_k}(0, \Sigma_k)$.

We need to prove that conditional on $\Fcal_k$, the distribution of $V_{k+1}$ is as above, where recall that $\Fcal_k = \sigma\{\Fcal_{k-1}, G_k,  \Zcal_k\}$  or equivalently, $\sigma\{ \Fcal_{k-1}, G_k, Z_k\}$. This claim follows from the conclusion of the previous paragraph by noting that $V_{k+1}$ is independent of $Z_k = (\gcoef_k)^{\trn}U_k$, conditional on $\sigma\{\Fcal_{k-1}, G_k\}$ as $\gcoef_k$ is orthogonal to $\xi_k^{coeff}$.

This completes the proof of the Lemma \ref{laterstepdist}.

%\section{Conclusion}

\section{The Method of Nearby Measures}\label{apnearby}

%The proof of the lemma follows immediately after the following two lemmas.
Let $b \in \Rbb^n$, be such that $\|b\|^2 = \nu <1$. Further, let $\PP$ be the probability measure of a
$N(0, \Sigma)$ random variable, where $\Sigma = I - bb^T$, and let $\QQ$ be the measure of a $N_n(0, I)$ random variable.  Then we have,
\begin{lem}\label{renyi} %\emph{The method of nearby measure.}
$$\PP[A] \, \le \, \QQ[A] e^{c_0},
$$
where $c_0 = -(1/2)\log(1 - \nu)$.
\end{lem}
\begin{proof} If $p(z),\, q(z)$, denote the densities of the random variables with measures $\PP$ and $\QQ$ respectively,  then $\max_z p(z)/q(z)$ equals $1/(1 - \nu)^{1/2}$, which is also $e^{c_0}$. From the densities $N(0,I\!-\!b b^T)$ and $N(0,I)$ this claim can be established from noting that after an orthogonal transformation these measures are only different in one variable, which is either $N(0,1\!-\!\nu)$ or $N(0,1)$, for which the maximum ratio of the densities occurs at the origin and is simply the ratio of the normalizing constants.

Correspondingly,
\algg{\PP(A) &= \int_A p(z)\, dz\\
             &\leq e^{c_0} \int_A q(z) \, dz =  \QQ(A) e^{c_0}.
}
This completes the proof of the lemma.
\end{proof}

%With $\nu \!=\! P/(\sigma^2\+P)$ this limit $-(1/2) \log [1\!-\! \nu]$ which we have denoted as $c_0$ is the same as $(1/2) \log[1+P/\sigma^2]$.  That it is the same as the capacity ${\Capacity}$ appears to be coincidental, as we do not have any direct communication rate interpretation of the operation of taking the log of the $L_{\infty}$ norm of the ratio of the densities that arise here.

\vspace{0.3cm}
\noindent{\bf Proof of Lemma \ref{measurebdd2}:}
We are to show that for events $A$ determined by the random variables \eqref{eq:randdet}, the probability $\PP[A]$ is not more than $\QQ[A] e^{k c_0}$. Write the probability as an iterated expectation conditioning on $\Fcal_{k-1}$.
%that determines all $\Hcal_j=\{\Zcal_{1,j} > \tau\}$.
That is, $\PP[A] = \E \left[ \PP[A|\Fcal_{k-1}]\right]$. To determine membership in $A$, conditional on $\Fcal_{k-1}$, we only need $Z_{k,J_k} = (Z_{k,j} : j \in J_k)$ where $J_k$ is determined by $\Fcal_{k-1}$. Thus
$$\PP[A] = \E_{\PP} \left[ \PP_{\Chi_{d_k}^2,Z_{k,J_k}|\Fcal_{k-1}}\big[A]\right],$$
where we use the subscript on the outer expectation to denote that it is with respect to $\PP$ and the subscripts on the inner conditional probability to indicate the relevant variables. For this inner probability switch to the nearby measure $\QQ_{\Chi_{d_k},Z_{k,J_k}|\Fcal_{k-1}}$.  These conditional measures agree concerning the distribution of the independent $\Chi_{d_k}^2$, %so the $\alpha$ relative entropy between them arises only from the normal distributions of the $Z_{k,J_{k}}$ given $\Fcal_{k-1}$. This $\alpha$ relative entropy is bounded by $c_0$.
so what matters is the ratio of the densities corresponding to $\PP_{Z_{k,J_k}|\Fcal_{k-1}}$ and $\QQ_{Z_{k,J_k}|\Fcal_{k-1}}$.

We claim that the ratio of these densities in bounded by $e^{c_0}$.
To see this, recall that from Lemma \ref{laterstepdist} that $\PP_{Z_{k, J_k}| \Fcal_{k-1}}$ is $N_{J_k}(0 , \Sigma_{k}) $, with $\Sigma_k = I - \delta_k\delta_k^T$.  Now $$||\delta_k||^2 = \nu_k\sum_{j \in sent \cap J_k} P_j/P $$
 which is $  (1 - (\hat q_1 + \ldots + \hat q_{k-1}))\nu_k$. Noting that $\nu_k = \hat s_k \nu$ and
$\hat s_k (1 - (\hat q_1 + \ldots + \hat q_{k-1})) $ is at most 1, we get that $||\delta_k||^2 \leq \nu$.
%Thus from Lemma \ref{normalclose}, for all $\alpha \geq 1$, the $\alpha$ relative entropy between $\PP_{Z_{k, J_k}| \Fcal_{k-1}}$ and the corresponding $\QQ$ conditional distribution is at most $c_0$.

So with the switch of conditional distribution, we obtain a bound with a multiplicative factor of $e^{c_0}$.  The bound on the inner expectation is then a function of $\Fcal_{k-1}$, so the conclusion follows by induction.  This completes the proof of Lemma \ref{measurebdd2}.

\section{Proof of Lemma \ref{progress}}\label{proofprogress}

%\noindent{\bf {Proof of Lemma \ref{progress}:}}
%The claim regarding the increase can be checked by induction as follows.
For $k\!=\!1$, the $q_{1,1} = g(0)-\eta$ is at least $gap-\eta$.
Consider $q_{1,k} =g_L(q_{k-1}^{adj,tot}) -\eta$, for $k > 1$. Notice that
$$q_{k-1}^{adj,tot} \geq \sum_{k' = 1}^{k-1} q_{k'} - (k-1)f,$$
using $q/(1 + f/q) \geq q - f$. Now, from the definition of $q_k$ in \eqref{eq:qk}, one has
$$\sum_{k' = 1}^{k-1} q_{k'} = q_{1,k-1} - (k-1)(f +1/L_\pi).$$
Consequently,
\eqq{q_{k-1}^{adj,tot} \geq q_{1,k-1} - (k-1)(2f +1/L_\pi)\label{eq:lowerbddqkadj}.}
Denote $m$ as the first $k$ for which $q_{k-1}^{adj,tot}$ exceeds $x_r$. For any $k < m$, as
$q_{k-1}^{adj,tot} \leq x_r$, using the fact that $g_L$ is accumulative till $x_r$, one gets that
 $$q_{1,k} \geq q_{k-1}^{adj,tot} + gap - \eta.$$
Accordingly, using \eqref{eq:lowerbddqkadj}, one gets that
\eqq{ q_{1,k} \geq q_{1,k-1} - (k-1)(2f +1/L_\pi) + gap - \eta,
}
or in other words, for $k < m$, one has
$$q_{1,k} - q_{1,k-1} \geq -m(2f +1/L_\pi) + gap -\eta.$$
We want to arrange the difference $q_{1,k} - q_{1,k-1}$ to be at least a positive quantity which we denote by $\tL$. Notice that this gives $m \leq 1/\tL$, since the $q_{1,k}$'s are bounded by 1. Correspondingly, we solve for $\tL$ in,
$$\tL = -(1/\tL)(2f +1/L_\pi) + gap -\eta,$$
and see that the solution is
\eqq{\tL = \frac{(gap - \eta)}{2}\left[1 + \left(1 - 4\frac{(2f + 1/L_\pi)}{(gap - \eta)^2}\right)^{1/2}\right],\label{eq:tL}}
which is well defined since $f$ satisfies \eqref{eq:fcriterion}. Also notice that from \eqref{eq:tL} that
$\tL \geq (gap - \eta)/2$, making $m \leq 2/(gap -\eta)$.
Also $q_{1,m} = g_L(q_{m-1}^{adj,tot}) - \eta$, which is at least $g_L(x_r) - \eta$ since $g_L$ is increasing. The latter quantity is at least $x_r + gap - \eta$.

\section{Proof of Lemma \ref{lem:paccritlem}} \label{appfpaccritlem}

We prove the lemma by first showing that  \eqq{{A}_{1,m}  \subseteq  \tilde{A}_{1,m}\cup B_{1,m} \cup \ah{1,m}\label{eq:firstshow}.}
Next, we prove that $B_{1,m}$ is contained in $\tilde{B}_{1,m}$. This will prove the lemma.

We start by showing \eqref{eq:firstshow}. We first show that on the set
 %\eqq{A_k = \{\hat q_{1,k}< q_{1,k} \}\cup \{\fkup{k} > f_k \}\cup A \label{eq:akset}}
 \eqq{\{\hat q_{1,k} \geq q_{1,k}\}\cap \aee{k-1}^c\cap\ah{k}^c,\label{eq:setundercond}}
 condition \eqref{eq:pacingworkcond}, that is,
 \eqq{\sum_{j \in dec_{1,k-1}} \pi_j + \sum_{j \in J - dec_{1,k-1}} \pi_j 1_{\{\Zcal_{k,j}^{comb} \geq \tau\}} \geq q_{1,k},\label{eq:pacingworkcondp}}
 is satisfied. Following the arguments of subsection \ref{subsec:modpacing} regarding pacing the steps, this will ensure that the size of the decoded set after $k$ steps, that is $size_{1,k}$, is near $q_{1,k}$, or more precisely
\eqq{q_{1,k} -1/L_\pi < size_{1,k} \leq q_{1,k},\label{eq:paccriterionp}}
as given in \eqref{eq:paccriterion}.

Notice that the left side of \eqref{eq:pacingworkcondp} is at least
$$\sum_{j \in \sent} \pi_j 1_{\{\Zcal_{k,j}^{comb} \geq \tau\}},$$
since the sum in \eqref{eq:pacingworkcondp} is over all terms in $j$, including those in $\sent$, and further, for each term $j$, the contribution to the sum is at least $\pi_j 1_{\{\Zcal_{k,j}^{comb} \geq \tau\}}$.

Further, using the fact that
 $$H_{k,j} \subseteq \{\Zcal_{k,j}^{comb} \geq \tau\} \quad \text{on $\aee{k-1}^c\cap\ah{k}^c$}$$
   from \eqref{eq:hkjsentlow}, one gets that,
   $$\sum_{j \in \sent} \pi_j 1_{\{\Zcal_{k,j}^{comb} \geq \tau\}} \geq \hat q_{1,k}\quad\quad\text{on  $\aee{k-1}^c\cap\ah{k}^c$ }.$$
Correspondingly, on the set \eqref{eq:setundercond} the inequality \eqref{eq:pacingworkcondp}, and consequently the relation \eqref{eq:paccriterionp} also holds.

Next, for each $k$, denote
\eqq{\aeet{k} = \tilde{A}_{1,k}\cup\af{1,k} \cup\ah{1,k}\label{eq:tildee}.}
We claim that for each $k = 1,\ldots, m$, one has$$\aeet{k}^c\, \subseteq \,\aq{1,k}^c.$$ We prove the claim through induction on $k$.
Notice that the claim for $k = m$ is precisely statement \eqref{eq:firstshow}. Also, the claim implies that
$\aeet{k}^c\subseteq\aee{k}^c$, for each $k$, where recall that $\aee{k} = \aq{1,k}\cup\af{1,k}\cup\ah{1,k}$.

We first prove the claim for $k= 1$. We see that,
$$\aeet{1}^c = \{\hat q_{1,1} \geq q_{1,1}\} \cap\{\hat f_1 \leq f\}\cap \ah{1}^c.$$
Using the arguments above, we see that on $\{\hat q_{1,1} \geq q_{1,1}\}\cap\ah{1}^c$, the relation $q_{1,1} - 1/L_{\pi} <  size_{1,1}$ holds. Now, since $size_{1,1} = \hat q_1 + \hat f_1$, one gets that
$$\hat q_1 \geq q_{1,1} - \hat f_1 - 1/L_{\pi}\quad\text{on $\{\hat q_{1,1} \geq q_{1,1}\}\cap\ah{1}^c$}.$$
The right side of the aforementioned inequality is at least $q_1$ on $\aeet{1}^c$, using $\hat f_1 \leq f$. Consequently, the claim is proved for $k = 1$.

Assume that the claim holds till $k -1$, that is, assume that $\aeet{k-1}^c \subseteq \aq{1,k-1}^c$. We now prove that $\aeet{k}^c \subseteq \aq{1,k}^c$ as well. Notice that
$$\aeet{k}^c = \{\hat q_{1,k} \geq q_{1,k}\}\cap \aeet{k-1}^c\cap\ah{k}^c\cap\{\hat f_k \leq f\},$$
which, using $\aeet{k-1}^c \subseteq \aee{k-1}^c$ from the induction hypothesis, one gets that
$$q_{1,k} - 1/L_{\pi} <  size_{1,k} \quad\text{ on $\{\hat q_{1,k} \geq q_{1,k}\}\cap \aeet{k-1}^c\cap\ah{k}^c$} .$$
Accordingly,
$$q_{1,k} - 1/L_{\pi} <  size_{1,k} \leq q_{1,k} \quad\text{ on $\aeet{k}^c$} .$$
Further, as $\aeet{k}^c$ is contained in $\aeet{k-1}^c$, one gets that
$$q_{1,k-1} - 1/L_\pi <  size_{1,k-1} \leq q_{1,k-1} \quad\text{on $\aeet{k}^c$}.$$
Consequently, combining the above, one has
\algg{size_{1,k} - size_{1,k-1} &= \hat q_k + \hat f_k\\
          &\geq q_{1,k} - q_{1,k-1} - f - 1/L_{\pi}\quad\text{on $\aeet{k}^c$}.
}
Consequently, on $\aeet{k}^c$, we have $\hat q_k \geq q_k$,
using the expression for $q_k$ given in \eqref{eq:qk}, and the fact that $f_k \leq f$ on $\aeet{k}^c$. Combining this with the fact that $\aeet{k}^c$ is contained in $A_{1,k-1}^c$, since $\aeet{k}^c \subseteq \aeet{k-1}^c$ and $\aeet{k-1}^c \subseteq A_{1,k-1}^c$ from the induction hypothesis, one gets that
$\aeet{k}^c \subseteq A_{1,k}^c$.

This proves the induction hypothesis. In particular, it holds for $k = m$, which, taking complements, proves the statement \eqref{eq:firstshow}.

Next, we show that $B_{1,m} \subseteq \tilde{B}_{1,m}$. This is straightforward since recall that from subsection \ref{subsec:targetfalse} that one has $\hat f_k \leq \fkup{k}$, for each $k$. Correspondingly, $B_{1,m}$ is contained in $\tilde{B}_{1,m}$.

Consequently, from \eqref{eq:firstshow} and the fact that $B_{1,m} \subseteq \tilde{B}_{1,m}$, one gets that
$\aee{m}$ is contained in $\tilde{A}_{1,m} \cup \tilde{B}_{1,m} \cup \ah{1,m}$. This proves the lemma.

\section{Tails for weighted Bernoulli sums} \label{aptailber}

\vspace{0.3cm}
\begin{lem}\label{bernoullisums} Let $W_j$, $1 \leq j \leq N$ be $N$  independent $\mbox{Bernoulli}(r_j)$ random variables. Furthermore, let $\alpha_j,\, 1\leq j \leq N$ be non-negative weights that sum to 1 and let $ N_\alpha = 1/\max_{j}\alpha_j$.
Then the weighted sum $\hat{r} = \sum_{j}\alpha_j W_j$ which has mean given by $r^* = \sum_{j} \alpha_j r_j$, satisfies the following large deviation inequalities.
For any $r$ with $0< r < r^*$,
\eqs{P(\hat{r} < r ) \leq \exp\left\{-N_\alpha D(r\|r^*)\right\}
}
and for any $\tilde{r}$ with  $r^* < \tilde{r} < 1$,
\eqs{P(\hat{r} > \tilde{r} ) \leq \exp\left\{-N_\alpha D(\tilde{r}\|r^*)\right\}
}
where $D(r\|r^*)$ denotes the relative entropy between Bernoulli random variables of success parameters $r$ and $r^*$.
\end{lem}

\vspace{0.3cm}
\noindent {\bf Proof of Lemma \ref{bernoullisums}:} Let's prove the first part. The proof of the second part is similar.

Denote the event $$\mathcal{A} = \{\underline{W} : \sum_{j}\alpha_j W_j \leq r \}$$ with $\underline{W}$ denoting the
$N$-vector of $W_j$'s. Proceeding as in \citet{csiszar1984sanov} we have that
$$P\left( \mathcal{A} \right) \,=\, \exp\{-D\big(P_{\underline{W}|\mathcal{A}}\|P_{\underline{W}}\big)\}\quad\quad\quad\quad$$
$$\quad\quad \leq \, \exp\big\{-\sum_j D\big(P_{W_j|\mathcal{A}}||P_{W_j}\big)\big\}$$
Here $P_{\underline{W}|\mathcal{A}}$ denotes the conditional distribution of the vector $\underline{W}$ conditional on the event $\mathcal{A}$ and  $P_{W_j|\mathcal{A}}$ denotes the associated marginal distribution of $W_j$ conditioned on $\mathcal{A}$. Now
\eqs{\sum_j D\big(P_{W_j|\mathcal{A}}\|P_{W_j}\big) \geq N_\alpha \sum_j \alpha_{j}D\big(P_{W_j|\mathcal{A}}\|P_{W_j}\big).
}
Furthermore, the convexity of the relative entropy implies that
\eqs{\sum_{j} \alpha_j D(P_{W_j|\mathcal{A}}\parallel P_{W_j}) \geq D\left(\sum_j \alpha_j P_{W_j|\mathcal{A}}\parallel\sum_j \alpha_j P_{W_j}\right) .
}
The sums on the right denote $\alpha$ mixtures of distributions $P_{W_j|\mathcal{A}}$ and $P_{W_j}$, respectively, which are distributions on $\{0,1\}$, and hence these mixtures are also distributions on $\{0,1\}$.  In particular, $\sum_j \alpha_j P_{W_j}$ is the Bernoulli($r^*$) distribution and
%Denoting $\tilde{r}_j = P(W_j = 1 | \mathcal{A})$, the quantity
$\sum_{j}\alpha_j P_{W_j|\mathcal{A}}$ is the
$\mbox{Bernoulli}(r_e)$ distribution where $$r_e =  \mbox{E}\big[\,\sum_j\alpha_j W_j \,\big| \,\mathcal{A}\,\big]=
\mbox{E}\big[\,\hat r \,\big| \,\mathcal{A}\,\big] .$$
But in the event $\mathcal{A}$ we have $\hat r \le r$ so it follows that $r_e \leq r$. As $r < r^*$ this yields $D(r_e \parallel  r^*) \geq
D(r\parallel r^*)$.  This completes the proof of Lemma \ref{bernoullisums}.

\section{Lower Bounds on $D$} \label{lowerboundD}

\vspace{0.3cm}
\begin{lem}\label{lowerbounds} For $p\ge p^*$, the relative entropy between Bernoulli$(p)$ and Bernoulli$(p^*)$ distributions has the succession of lower bounds
$$D_{Ber}(p\|p^*) \ge D_{Poi}(p\|p^*)\ge 2\big(\sqrt p - \!\sqrt {p^*}\,\big)^2 \ge \frac{(p-p^*)^2}{2p}$$
where $D_{Poi}(p\|p^*) = p\log p/p^* + p^*-p$ is also recognizable as the relative entropy between Poisson distributions of mean $p$ and $p^*$ respectively.
\end{lem}
%\vspace{0.3cm}
%\noindent{\bf Remark a:}  There are analogous statements for pairs of probability distributions $P$ and $P^*$ on a measurable space $\Xcal$ with densities $p(x)$ and $p^*(x)$ with respect to a dominating measure $\mu$. The relative entropy $D(P\|P^*)$ which is $\int p(x)\log p(x)/p^*(x) \mu(dx)$ may be written as the integral of the non-negative integrand $p(x)\log p(x)/p^*(x) + p^*(x)-p(x)$, which exceeds $(1/2)\big(p(x)-p^*(x))^2/\max\{p(x),p^*(x)\}$.  It is familiar that $D(P\|P^*)$ exceeds the squared Hellinger distance $H^2(P,P^*)=\int \big( \sqrt{p(x)}-\sqrt{p^*(x)}\,\big)^2 \mu(dx)$. That fact arises for instance via Jensen's inequality, from which $D$ exceeds $2 \log 1/(1-(1/2)H^2)$ which in turn is at least $H^2$. However, we have not been able to get $D\ge 2 H^2$ in general.  The above lemma with the factor of $2$ is for $p>p^*$, not for general integrands.

%\noindent{\bf Remark b:}  When $\hat p$ is the relative frequency of occurrence in $N$ independent Bernoulli trials it has the bound $P\{\hat p > p\}\le e^{-ND_{Ber}(p\|p^*)}$ on the upper tail of the Binomial distribution of $N\hat p$ for $p > p^*$.  In accordance with the Poisson interpretation of the lower bound on the exponent, one sees that this upper tail of the Binomial is in turn bounded by the corresponding large deviation expression that would hold if the random variables were Poisson.

%\vspace{0.3cm}
\begin{proof}
   The Bernoulli relative entropy may be expressed as the sum of two positive terms, one of which is $p \log p/p^* +p^*-p$, and the other is the corresponding term with $1\!-\!p$ and $1\!-\!p^*$ in place of $p$ and $p^*$, so this demonstrates the first inequality.  Now suppose $p>p^*$. Write $p\log p/p^* +p^*-p$ as $p^* F(s)$ where $F(s)=2 s^2 \log s + 1-s^2$ with $s^2=p/p^*$ which is at least $1$.  This function $F$ and its first derivative $F'(s)=4s\log s$ have value equal to $0$ at $s=1$, and its second derivative $F''(s)=4+4\log s$ is at least $4$ for $s\ge 1$.  So by second order Taylor expansion $F(s)\ge 2(s-1)^2$ for $s\ge 1$.  Thus $p\log p/p^* +p^*-p$ is at least $2\big(\sqrt{p}-\sqrt{p^*}\,\big)^2$.  Furthermore $2(s-1)^2 \ge (s^2-1)^2/(2s^2)$ as, taking the square root of both sides, it is seen to be equivalent to $2(s-1) \ge s^2-1$, which, factoring out $s-1$ from both sides, is seen to hold for $s\ge 1$.
From this we have the final lower bound $(p-p^*)^2/(2p).$
 \end{proof}

\section*{Acknowledgment}

We thank Dan Spielman, Edmund Yeh, Mokshay Madiman and Imre Teletar for helpful conversations.  We thank David Smalling who completed a number of simulations of earlier incarnations of the decoding algorithm for his Yale applied math senior project %\cite{Smalling2009}
in spring term of 2009 and Yale statistics masters student Creighton Hauikulani who took the simulations further in 2009 and 2010. %Their work was instrumental to us in recognizing that without modification, direct convex projection methods have a rate threshold substantially below capacity when the signal-to-noise ratio is high.

%\bibliographystyle{/Users/mokshay/Documents-ACADEMIC/WRITINGS/CommonResources/IEEEtranBST/IEEEtran}
%\bibliography{/Users/mokshay/Documents-ACADEMIC/WRITINGS/CommonResources/poi,/Users/mokshay/Documents-ACADEMIC/WRITINGS/CommonResources/ik}

\end{document}